\documentclass[nopreprintline,12pt]{elsarticle}
\usepackage{amssymb}
\usepackage[T1]{fontenc}
\usepackage{graphicx}
\usepackage{mathtools}
\usepackage{amsmath,amssymb,stmaryrd}
\usepackage{hyperref,doi,url}
\hypersetup{colorlinks=true, linkcolor=blue, breaklinks=true, urlcolor=blue}
\usepackage{figures/parking/carshapes}

\usepackage{stackengine}

\usepackage{enumitem}
\usepackage{algorithm}
\usepackage[noend]{algpseudocode}
\usepackage{multirow}
\usepackage{subcaption}
\usepackage{pgfplots}
\usepackage{pgfplotstable}
\pgfplotsset{compat=1.17}
\usepgfplotslibrary{fillbetween}
\allowdisplaybreaks

\usepackage{scalerel}
\newcommand\scale[2]{\vstretch{#1}{\hstretch{#1}{#2}}}

\usepackage{listofitems} % for \readlist to create arrays
\usetikzlibrary{arrows.meta} % for arrow size
\usepackage[outline]{contour} % glow around text
\contourlength{1.4pt}
\usepackage{tikz}
\usetikzlibrary{shapes, arrows.meta}
\usetikzlibrary{automata,petri,positioning,calc}
\usetikzlibrary{automata,positioning,arrows,quotes,through}
\tikzset{>=latex} % for LaTeX arrow head
\usepackage{xcolor}
\usepackage{siunitx}
\pgfplotsset{scaled y ticks=false}
\usetikzlibrary{pgfplots.groupplots,patterns}
\usepackage{xypic}

\newtheorem{defi}{\textbf{Definition}}
\newtheorem{thom}{\textbf{Theorem}}
\newtheorem{asp}{\textbf{Assumption}}
\newtheorem{rek}{\textbf{Remark}}

\newtheorem{lema}{\textbf{Lemma}}

\newdefinition{examp}{\textbf{Example}}
\newenvironment{proof}{\paragraph{Proof}}{\hfill$\square$ \vskip8pt}
\makeatletter
\renewcommand*{\@opargbegintheorem}[3]{\trivlist
      \item[\hskip \labelsep{\bfseries #1\ #2}] \textbf{(#3)}\ \itshape}
\makeatother

\newcommand{\defiref}[1]{Definition~\ref{#1}}

\newcommand{\thomref}[1]{Theorem~\ref{#1}}
\newcommand{\aspref}[1]{Assumption~\ref{#1}}

\newcommand{\algoref}[1]{Algorithm~\ref{#1}}

\newcommand{\lemaref}[1]{Lemma~\ref{#1}}
\newcommand{\tabref}[1]{Table~\ref{#1}}
\newcommand{\figref}[1]{Fig.~\ref{#1}}
\newcommand{\sectref}[1]{Section~\ref{#1}}
\newcommand{\eqnref}[1]{(\ref{#1})}
\newcommand{\egref}[1]{Example~\ref{#1}}

\renewcommand{\emptyset}{\varnothing}

\DeclareMathOperator*{\argmax}{argmax}

\newcommand{\Act}{\mathit{Act}}
\newcommand{\cO}{\mathcal{O}}
\newcommand{\obs}{\mathit{obs}}

\newcommand{\Loc}{\mathit{Loc}}
\newcommand{\Per}{\mathit{Per}}

\newcommand{\loc}{\mathit{loc}}

\newcommand{\per}{\mathit{per}}
\newcommand{\pomdp}{\mathsf{P}}

\newcommand{\fpath}{\mathit{FPath}}
\newcommand{\last}{\mathit{last}}
\newcommand{\mdp}{\mathsf{M}}

\newcommand{\agent}{\mathsf{Ag}}
\newcommand{\sem}[1]{\llbracket {#1} \rrbracket}

\newcommand{\supp}{\mathit{supp}}
\newcommand{\tr}{\mathit{tr}}
\newcommand{\ad}{\mathit{ad}}

\makeatletter
\tikzset{
  xyz frame/.code n args={3}{%
    \begingroup
    \tikz@scan@one@point\pgfutil@firstofone#1\relax
    \pgf@xa=\pgf@x
    \pgf@ya=\pgf@y
    \tikz@scan@one@point\pgfutil@firstofone#2\relax
    \pgf@xb=\pgf@x
    \pgf@yb=\pgf@y
    \tikz@scan@one@point\pgfutil@firstofone#3\relax
    \edef\tikz@marshall{\noexpand\endgroup\noexpand\pgfsetxvec{\noexpand\pgfpoint{\the\pgf@xa}{\the\pgf@ya}}%
      \noexpand\pgfsetyvec{\noexpand\pgfpoint{\the\pgf@xb}{\the\pgf@yb}}%
      \noexpand\pgfsetzvec{\noexpand\pgfpoint{\the\pgf@x}{\the\pgf@y}}}%
    \tikz@marshall
  },
  on layer/.code={
    \pgfonlayer{#1}\begingroup
    \aftergroup\endpgfonlayer
    \aftergroup\endgroup
  },
}
\newcommand{\grid}[2]{
    \node at (#1 + 3.95, #2 + -0.3) {$4$};
    \node at (#1 + -0.2, #2 + -0.3) {$0$};
    \node at (#1 + -0.25, #2 + 3.95) {$4$};
    \node at (#1 + -0.25, #2 + 3) {$3$};
    \node at (#1 + -0.25, #2 + 2) {$2$};
    \node at (#1 + -0.25, #2 + 1) {$1$};
    \node at (#1 + 3, #2 + -0.3) {$3$};
    \node at (#1 + 2, #2 + -0.3) {$2$};
    \node at (#1 + 1, #2 + -0.3) {$1$};
    \draw[black, very thick] (#1, #2 + 0) rectangle (#1 + 4, #2 + 4);
    \draw[black, thick] (#1, #2 + 1) -- (#1 + 4, #2 + 1);
    \draw[black, thick] (#1, #2 + 2) -- (#1 + 4, #2 + 2);
    \draw[black, thick] (#1, #2 + 3) -- (#1 + 4, #2 + 3);
    \draw[black, thick] (#1 + 1, #2 + 0) -- (#1 + 1, #2 + 4);
    \draw[black, thick] (#1 + 2, #2 + 0) -- (#1 + 2, #2 + 4);
    \draw[black, thick] (#1 + 3, #2 + 0) -- (#1 + 3, #2 + 4);
}
\makeatother

\newcommand{\startpara}[1]{{%
\vskip5pt\noindent
{\bf #1.}}}

\newcommand{\startparazero}[1]{{%
\vskip0pt\noindent
{\bf #1.}}}

\newcommand{\marta}[1]{\rm {\raggedright\color{orange}\textsf{MK: #1}\marginpar{$\star$}}} 
\newcommand{\dave}[1]{\rm {\raggedright\color{teal}\textsf{DP: #1}\marginpar{$\star$}}}

\newcommand{\daveM}[1]{{\marginpar{\color{teal}\textsf{DP: #1}}}} 
\newcommand{\martaM}[1]{{\marginpar{\color{orange}\textsf{MK: #1}}}}

\newcommand{\revise}{}

\journal{Artificial Intelligence}

\begin{document}

\begin{frontmatter}



\title{Point-Based Value Iteration for {\revise POMDPs \\ with Neural Perception Mechanisms}}


\author[Oxford]{Rui~Yan}\ead{rui.yan@cs.ox.ac.uk}
\author[Oxford]{Gabriel~Santos}\ead{gabriel.santos@cs.ox.ac.uk}
\author[Oxford,Glasgow]{Gethin~Norman}\ead{gethin.norman@glasgow.ac.uk}
\author[Oxford]{David~Parker}\ead{david.parker@cs.ox.ac.uk}
\author[Oxford]{Marta~Kwiatkowska}\ead{marta.kwiatkowska@cs.ox.ac.uk}

\address[Oxford]{Department of Computer Science, University of Oxford, Oxford, OX1 3QD, UK}
\address[Glasgow]{School of Computing Science, University of Glasgow, Glasgow, G12 8QQ, UK}

% \affiliation{organization={},%Department and Organization

\begin{abstract}
{\revise The increasing trend to integrate neural networks and conventional software components in safety-critical settings calls for methodologies for their formal modelling, verification and correct-by-construction policy synthesis.}
In this paper, we introduce neuro-symbolic partially observable Markov decision processes (NS-POMDPs), {\revise a variant of continuous-state POMDPs with discrete observations and actions, in which the agent}
perceives a continuous-state environment using a neural {\revise perception mechanism} and makes decisions symbolically. 
{\revise The perception mechanism classifies inputs such as images and sensor values into symbolic percepts, which are used in decision making.}

We study the problem of optimising discounted cumulative rewards {\revise for NS-POMDPs.}
{\revise Working directly with the continuous state space, we exploit the underlying structure of the model and the neural perception mechanism to} 
propose a novel 
piecewise linear and convex representation (P-PWLC) in terms of polyhedra covering the state space and value vectors,
and extend Bellman backups to this representation. We prove the convexity and continuity of value functions and present two value iteration algorithms that ensure finite representability.   
The first is a classical (exact) value iteration algorithm extending the $\alpha$-functions of Porta {\em et al} (2006) 
to the P-PWLC representation for continuous-state spaces. The second is a point-based (approximate) method called NS-HSVI, which uses the P-PWLC representation and belief-value induced functions to approximate value functions from below and above for two types of beliefs, particle-based and region-based.
Using a prototype implementation, we show the practical applicability of our approach on two case studies that employ (trained) ReLU neural networks as perception functions, dynamic car parking and an aircraft collision avoidance system, by synthesising (approximately) optimal strategies. 
An experimental comparison with the finite-state POMDP solver SARSOP demonstrates that NS-HSVI is more robust to particle disturbances.
\end{abstract}

\begin{keyword}
Neuro-symbolic systems \sep continuous-state POMDPs \sep {\revise neural perception} \sep point-based value iteration \sep heuristic search value iteration



\end{keyword}

\end{frontmatter}



\section{Introduction}

An emerging trend in artificial intelligence is to integrate traditional symbolic techniques with data-driven components in sequential decision making and optimal control. Application domains include mobile robotics \cite{JKG-ME-MK:17}, visual reasoning \cite{SA-HP-AP-YH-KK:20}, autonomous driving \cite{SSS-SS-AS:16} and aircraft control \cite{MEA-EB-PK-AL:20}. In real-world autonomous navigation systems, agents rely on unreliable sensors to perceive the environment, typically represented using continuous state spaces, and planning and control must deal with environmental uncertainty. Neural networks (NNs) have proven effective in these complex settings at providing fast data-driven perception mechanisms capable of performing tasks such as object detection or localisation.
{\revise They are increasingly often deployed in conjunction with conventional controllers based on symbolic approaches, which can provide high interpretability, provable correctness guarantees and ease of inserting human expert knowledge~\cite{MKS-LZ-AE-PH}. }
Because of the potential applicability in safety-critical domains, there is growing interest in methodologies for {\revise formal modelling, verification and correct-by-construction policy synthesis for such settings,}
but methods for these are currently lacking.

\emph{Partially observable Markov decision processes} (POMDPs) are a convenient mathematical framework to plan under uncertainty. %, which can be modelled in probabilistic fashion. 
Solving POMDPs in a scalable and efficient manner is 
% already
challenging for finite-state models~\cite{CHP-JNT:87,CCGK16}, but significant progress has been made, e.g., through point-based methods~\cite{GS-JP-RK:13}, which extend the classic value iteration algorithm for MDPs by applying it to a selected set of \emph{belief states} of the POMDP.
Typically, a belief state is a distribution over the states of the model
representing an agent's knowledge about the current state. Since the resulting belief MDP is infinite-state, conventional value iteration cannot be directly applied and instead point-based methods rely on a so-called {\em $\alpha$-vector} parameterisation, a linear function characterised by its values in the vertices of the belief simplex, which is finitely representable since the value function is piecewise linear and convex. 

Compared to finite-state POMDPs, solving continuous-state POMDPs suffers from additional challenges due to the uncountably infinite underlying state space. The common approach {\revise of discretising
and then using techniques for finite-state models
can yield exponential growth of the state space,
depending on the granularity and time horizon.}
{\revise An alternative
method that has been shown to outperform discretisation relies on exploiting structure in the underlying model
and working directly with the continuous state space~\cite{ZF-RD-NM-RW:04} through 
a piecewise constant representation of the value function,
based on a partition of the state space created dynamically during solution.}

\iffalse{
\marta{Drop this para}Additionally, belief spaces for continuous-state POMDPs have infinitely many dimensions, which further complicates the problem. Since functions over continuous spaces can have arbitrary forms
not amenable to computation, a key challenge is finding an efficient representation of the value function that allows closed-form belief updates and Bellman backups for the 
% underlying 
(parameterisable) transition and reward functions. This problem was addressed by Porta {\em et al} in~\cite{JMP-NV-MTS-PP:06}, where it was proved that continuous-state POMDPs with discrete observations and actions have a piecewise linear and convex value function and admit a finite representation in terms of so-called {\em $\alpha$-functions}, which generalise  $\alpha$-vectors by replacing weighted summation with integration. Working with a representation in terms of linear combinations of Gaussian mixtures, they derive point-based value iteration and implement it by randomly sampling belief points to approximate the value function. 
}
\fi

{\revise In this paper, we propose \emph{neuro-symbolic POMDPs} (NS-POMDPs),
a variant of continuous-state POMDPs with discrete observations and actions. 
In NS-POMDPs, the agent observes the environment using a data-driven \emph{neural perception} mechanism, which classifies inputs such as images and sensor values into a finite set of symbolic \emph{percepts}, and makes decisions symbolically. 
We constrain the interface between the neural perception and symbolic decision making mechanisms so that the agent transitions to its next local state based on the current local state and percept, rather than accessing the state of the environment directly.
This enables knowledge acquisition of the learnt concepts from the neural perception mechanism.}
Our model is expressive enough for realistic perception functions, {\revise such as ReLU NNs,} while being sufficiently tractable to solve.

We address the problem of synthesising policies
that optimise discounted cumulative rewards for NS-POMDPs.
Working directly with continuous state spaces, 
we propose a novel finite representation of the value function inspired by %the 
\emph{$\alpha$-functions}, {\revise introduced by Porta {\em et al} in~\cite{JMP-NV-MTS-PP:06}, which generalise  $\alpha$-vectors by replacing weighted summation with integration. 
Our representation exploits the fact that
the 
neural perception functions
are classifiers, and thus
% ReLU NN classifiers
induce a finite decomposition
of the continuous environment into regions.} 

We prove convergence and continuity of the value function, and present two algorithms for this representation: classical value iteration (VI) and a variant of the HSVI (Heuristic Search Value Iteration) algorithm~\cite{TS-RS:04}. 
We first demonstrate
that, by exploiting the structure of NS-POMDPs, one can indeed find 
an $\alpha$-function representation, namely {\em piecewise linear and convex representation under piecewise constant $\alpha$-functions (P-PWLC)}, that has a simple parameterisation and is closed with respect to belief updates and the Bellman operator. More specifically, we show that value functions can be represented using pointwise maxima of piecewise constant $\alpha$-functions (a finite set of polyhedra and a value vector), 
in conjunction with mild assumptions that ensure closure with respect to the transition and reward functions of NS-POMDPs.

Since $\alpha$-functions for VI increase exponentially in the number of observations,
we propose a variant of HSVI, called NS-HSVI,
which approximates the value function from above and below. 
Starting with the polyhedral preimage of the model's %NN 
perception function, NS-HSVI works by progressively subdividing the continuous state space during value backups to compute lower bounds.  
We use a lower $K$-Lipschitz envelope of a convex hull to approximate an upper bound. We formulate two 
% $\alpha$-function
representations of the belief space, which have closed forms for the quantities of interest: \emph{particle-based}, which relies on sampling of individual points, and \emph{region-based}, which places a (uniform) distribution over a region of continuous space. 

We develop a prototype implementation of the techniques and provide experimental results for strategy (policy) synthesis for particle- and region-based beliefs on two case studies: a dynamic car parking scenario and an aircraft collision avoidance system{\revise, both with ReLU NN perception mechanisms}. We find that region-based values are more robust to disturbance than particle-based. We also compare our particle-based NS-HSVI to a finite-state POMDP approximation of an NS-POMDP model using SARSOP~\cite{HK-DH-WL:08}, and observe that our method consistently yields tighter lower bound values, at a higher computational cost due to expensive polyhedra computations, because the accuracy of SARSOP's lower bound depends 
on the length of the horizon considered when building the model.

\startpara{Contributions} 
In summary, this paper makes the following contributions.

\begin{enumerate}
\item	We propose neuro-symbolic POMDPs, {\revise a variant} of continuous-state POMDPs with discrete observations and actions, and observation functions {\revise defined by neural networks whose outputs are stored locally as symbolic percepts and used in the agent's decision making}.
\item 	We propose a novel piecewise constant $\alpha$-function representation of the value function (as a pointwise maximum function over a set of piecewise constant $\alpha$-functions defined over the continuous state space). We show that this representation admits a finite polyhedral representation and is closed with respect to the Bellman operator. 
\item	We prove continuity and convexity of the value function for discounted cumulative rewards and derive a value iteration (VI) algorithm. 
\item	We present a new point-based method called NS-HSVI for approximating values of NS-POMDPs, proving that 
piecewise constant $\alpha$-functions
are a suitable representation for lower bound approximations of values. 
We develop two variants of the algorithm, one based on the popular particle-based beliefs and the other on novel region-based beliefs, and show they have closed forms for computing the quantities of interest.
\item	We provide experimental results to demonstrate the applicability of NS-HSVI in practice for neural 
{\revise perception mechanisms} %networks 
whose preimage (or that of their approximation) is in polyhedral form. 
\end{enumerate}

\startpara{Structure of the paper} 
The remainder of the paper is structured as follows.
Section~\ref{background-sect} provides the relevant background material.
Section~\ref{nspomdp-sect} proposes our model of neuro-symbolic POMDPs, together with its belief MDP, and gives an illustrative example. Section \ref{sec:value-iteration} introduces piecewise constant representations for functions in NS-POMDPs, and shows that they have a finite representation (P-PWLC) that ensures closure under the Bellman operator. A new value iteration (VI) algorithm is also proposed, and we prove the convexity and continuity of the value function. Section \ref{sec:HSVI} presents a new HSVI algorithm for NS-POMDPs, which uses P-PWLC functions and belief-value induced functions to approximate the value function from below and above, and considers two belief representations for the implementation. Section \ref{sec:case-study} presents a prototype implementation and experimental evaluation of our approach
on two case studies.
{\revise Section~\ref{rw-sect} discusses related work
and} Section \ref{sec:conclusions} concludes the paper.
To ease presentation, proofs of the theorems and lemmas have been placed in the Appendix.

\section{Background}\label{background-sect}
This section introduces notation and preliminaries concerning {\revise neural networks,} Markov decision processes (MDPs), their partially observable variant (POMDPs)
and the construction of the (fully observable) belief MDP. 

\startpara{Notation}
The space of probability measures on a Borel space $X$ is denoted $\mathbb{P}(X)$, and the space of bounded real-valued functions on $X$ is denoted $\mathbb{F}(X)$.
{\revise 
A \emph{finite connected partition (FCP)} of $X$, denoted $\Phi$, is a finite collection of disjoint connected subsets (regions) that cover $X$. 

\begin{defi}[PWC function]\label{defi:PWC-func}
A function $f: X \to \mathbb{R}$ is piecewise constant (PWC) if there exists an FCP $\Phi$ of $X$ such that $f : \phi \to \mathbb{R}$ is constant for all $\phi \in \Phi$. Such an FCP $\Phi$ is called a constant-FCP of $X$ for $f$. We denote by $\mathbb{F}_{C}(X)$ the subset of PWC functions of $\mathbb{F}(X)$. 
\end{defi}
\begin{defi}[PWL function]\label{defi:PWL-Borel-func}
A function $f:X\to \mathbb{R}$ is  piecewise linear (PWL) if there exists a FCP $\Phi$ of $X$ such that  $f : \phi \to \mathbb{R}$ is linear and bounded for all $\phi \in \Phi$.
\end{defi} 
A function $f: X \to \mathbb{R}$ is \emph{piecewise continuous} if there exists an FCP $\Phi$ of $X$ such that $f : \phi \to \mathbb{R}$ is continuous for all $\phi \in \Phi$.
}

{\revise \startpara{Neural networks} 
A \emph{neural network (NN)} is a real vector-valued function $f:\mathbb{R}^m \to \mathbb{R}^c$, where $m,c \in \mathbb{N}$, and is said to be a \emph{classifier} for a set of classes $C$ of size $c$ if, for any input $x \in \mathbb{R}^m$, the output $f(x) \in \mathbb{R}^c$ is a probability vector, where the $i$th element of $f(x)$ represents the confidence probability of the $i$th class of $C$, i.e., a classifier is a function $f : \mathbb{R}^m \rightarrow \mathbb{P}(C)$.

Let $f^{\max} : \mathbb{R}^m \rightarrow C$ denote a function that returns the class with the largest confidence probability in $f(x)$, and call $f^{\max}(x)$ the \emph{class} of $x$. To allow for situations where the class with the highest probability returned by $f$ is not unique,
% and hence $f^{\max}(x)$ would be undefined,
we assume the classifier includes
a \emph{tie-breaking rule} defined by a function $\kappa :2^{C} \to C$ which, given a set of classes, i.e., those with the highest probability, returns the selected class.
}

{\revise
Given an NN classifier $f$ with the tie-breaking rule $\kappa$, the \emph{preimage} of $f$ divides $\mathbb{R}^m$ into an FCP $\Phi$ of $\mathbb{R}^m$, i.e., for any $\phi \in \Phi$, there exists a class $y$ such that $f^{\max}(x)=y$ for all $x \in \phi$.
For an NN classifier with PWL activation functions, this FCP $\Phi$ can be extracted, or approximated, by analysing its preimage \cite{KM-FF:20}, which can be computed offline for a trained NN.}

\startpara{MDPs}
We focus on (Borel measurable) continuous-state MDPs, which model a single agent executing in a continuous environment by transitioning probabilistically between states.
Formally, an MDP
is given as a tuple $\mdp=(S,\Act,\Delta,\delta)$, where $S$ is a Borel measurable set of states, $\Act$ a finite set of actions, $\Delta : S \rightarrow 2^\Act$ an available action function and $\delta : (S {\times} \Act) \rightarrow \mathbb{P}(S)$ a probabilistic transition function.  

\iffalse
\begin{defi}[MDP]
An MDP is a tuple $\mdp=(S,\Act,\Delta,\delta)$, where $S$ is a Borel measurable set of states, $\Act$ a finite set of actions, $\Delta : S {\rightarrow} 2^\Act$ an available action function and $\delta : (S {\times} \Act) {\rightarrow} \mathbb{P}(S)$ a probabilistic transition function.
\end{defi}
\fi

When in state $s$ of MDP $\mdp$, the agent has a choice
between the available actions $\Delta(s)$ and, if $a \in \Delta(s)$ is chosen, then the probability of moving to state $s'$ is $\delta(s,a)(s')$. A path of $\mdp$ is a sequence $\pi = s_0 \xrightarrow{a_0} s_1 \xrightarrow{a_1} \cdots$ such that $s_i \in S$, $a_i \in \Delta(s_i)$ and $\delta(s_i,a_i)(s_{i+1})>0$ for all $i$.  We let $\pi(i)=s_i$ and $\pi[i]=a_i$ for all $i$. $\fpath_\mdp$ is the set of finite paths of $\mdp$ and $\last(\pi)$ is the last state of $\pi$ for any $\pi \in \fpath_\mdp$.

A {\em strategy (policy)} of $\mdp$ resolves the choices in each state based on the execution so far. Formally, a strategy $\sigma$ is a Borel measurable mapping $\sigma : \fpath_\mdp \rightarrow \mathbb{P}(\Act)$ such that,  if $\sigma(\pi)(a)>0$, then $a \in \Delta(\last(\pi))$. We denote by $\Sigma_\mdp$ the set of strategies of $\mdp$. A strategy is memoryless if the choice depends only on the last state of each path.

\startpara{POMDPs}
POMDPs are an extension of MDPs, in which the agent cannot perceive the underlying state but instead must infer it based on observations. 
Formally, a POMDP is a tuple $\pomdp=(S,\Act,\Delta,\delta,\cO,\obs)$,  where $(S,\Act,\Delta,\delta)$ is an MDP, $\cO$ is a finite set of \emph{observations} and $\obs : S \rightarrow \cO$ {\revise is a (deterministic) \emph{observation function}, i.e., $\obs(s)$ is the observation made upon entering state $s$. We require that,
}
for any $s,s' \in S$, if $\obs(s)=\obs(s')$ then $\Delta(s)=\Delta(s')$. Note that the underlying state space of the POMDP is uncountably infinite with a continuous-state structure.

\iffalse
\begin{defi}[POMDP]
A POMDP is a tuple $\pomdp=(S,\Act,\Delta,\delta,\cO,\obs)$,  where $(S,\Act,\Delta,\delta)$ is an MDP, $\cO$ is a finite set of observations and $\obs : S \rightarrow \cO$ is a labelling of states with observations such that, for any $s,s' \in S$, if $\obs(s)=\obs(s')$ then $\Delta(s)=\Delta(s')$.
\end{defi}
\fi
%
When in a state $s$ of a POMDP $\pomdp$, a strategy cannot determine $s$ directly, but only the observation $\obs(s)$. 
The definitions of paths and strategies for $\pomdp$ carry over from MDPs. However, the set of strategies $\Sigma_\pomdp$ of $\pomdp$  includes only {\em observation-based strategies}. Formally, a strategy $\sigma$ is observation-based if, for paths $\pi = s_0 \xrightarrow{a_0}  \cdots \xrightarrow{a_{n-1}} s_n$ and $\pi' = s_0' \xrightarrow{a_0}  \cdots \xrightarrow{a_{n-1}} s_n'$ such that $\obs(s_i)=\obs({\revise s_i'})$ for $0 \leq i \leq n$, then we have $\sigma(\pi)=\sigma(\pi')$.

\startpara{Objectives, values and optimal strategies}
We focus on the \emph{discounted cumulative reward} objectives, since they balance the importance of immediate rewards compared to future rewards, and allow optimisation of the behaviour over an infinite horizon. 
We note that the problem of undiscounted reward objectives is undecidable even for finite-state POMDPs.
For a \emph{reward structure} $r=(r_A,r_S)$, where $r_A:(S {\times} \Act)\to \mathbb{R}$ and $r_S:S\to \mathbb{R}$ are action and state bounded reward functions, the discounted cumulative reward for a path $\pi$ of a POMDP $\pomdp$ is given by:
\[  \begin{array}{c}  Y(\pi)= \mbox{$\sum_{k=0}^{\infty}$} \, \beta^k \big(r_A(\pi(k),\pi[k])+r_S(\pi(k)) \big)
\end{array} \]
where $\beta\in(0,1)$ is the discount factor. Given a state $s$ and strategy $\sigma$ of $\pomdp$, $\mathbb{E}_{s}^{\sigma}[Y]$ denotes the expected value of $Y$ when starting from $s$ under $\sigma$. {\revise Solving $\pomdp$ means finding a strategy $\sigma^{\star} \in \Sigma_\pomdp$, called an {\em optimal strategy}, that maximises the expected value of $Y$, }
and the (optimal) {\em value function} $V^\star : S \to \mathbb{R}$, which is defined as $V^\star(s) = \mathbb{E}_{s}^{\sigma^\star}[Y]$ for $s \in S$. 

\startpara{Belief MDP}
A strategy of POMDP $\pomdp$ can infer the current state from the observations and actions performed. The usual way of representing this knowledge is as a \emph{belief} $b \in \mathbb{P}(S)$. In general, observation-based strategies are %represent
more informative than belief-based strategies. However, since we focus on  discounted cumulative rewards, under the Markov assumption belief-based strategies carry sufficient information to plan optimally
\cite{DB:12}, and therefore, for a given objective $Y$, there exists an {\em optimal (observation-based) strategy} $\sigma$ of $\pomdp$, which can be represented as $\sigma : \mathbb{P}(S) \to \Act$. The strategy updates its belief $b$ to $b^{a,o}$ via Bayesian inference based on the executed action $a$ and observation $o$, i.e.\ for $s' \in S$:
\[
\begin{array}{c}
    b^{a,o} (s') = (P(o \mid s')/P(o \mid b, a)) \int_{s \in S} \delta(s,a)(s')b(s) \textup{d} s \, .
\end{array}
\]
Using this update we can define the corresponding (fully observable) \emph{belief MDP} in a standard way \cite{LPK-MLL-ARC:98}, from which an optimal strategy can be derived. We remark that belief spaces for continuous-state POMDPs are continuous and have infinitely many dimensions.

\section{Neuro-Symbolic POMDPs}\label{nspomdp-sect}

In this section we introduce our model of neuro-symbolic POMDPs, aimed at scenarios where the agent perceives its environment using a data-driven perception mechanism {\revise and then makes decisions symbolically based on the outcome}. 
We also give an illustrative example and then describe how a (fully observable) belief MDP can be obtained for an NS-POMDP.

\startpara{NS-POMDPs}
The model of \emph{neuro-symbolic POMDPs}
comprises a neuro-symbolic \emph{agent} acting in a continuous-state environment. 
The agent has finitely many local states and actions, and 
{\revise 
observes the environment through a (data-driven) \emph{neural perception} mechanism (called the agent's perception function), which can depend on the agent's current local state, while relying on a \emph{symbolic} decision-making mechanism (the agent's transition function). During execution, the agent alternates between invoking perception and symbolic decisions, where the interface is constrained to enable symbolic reasoning with the (exactly) learnt concepts (regions of the continuous inputs space), which we call \emph{percepts} to distinguish them from local states. When invoking perception, continuous inputs are converted into symbolic percepts, and the agent transitions  to the next local state based on the current local state and percept, rather than the environment state, and can thus model knowledge acquisition from the neural perception mechanism. }

\begin{defi}[Syntax of NS-POMDPs]\label{defi:NS-POMDP}
An NS-POMDP $\pomdp$ comprises
an agent $\agent  = (S_A,\Act,\Delta_A,\obs_A,\delta_A)$ and environment $E=(S_E,\delta_E)$ where:
\begin{itemize}
    \item $S_A = \Loc \times \Per$ is a set of states for $\agent$, where {\revise $\Loc$ and $\Per$}
    are finite sets of local states and percepts, respectively; 

    \item $S_E\subseteq \mathbb{R}^e$ is a closed set of continuous environment states;
  
    \item $\Act$ is a nonempty finite set of actions for $\agent$;

    \item $\Delta_A: S_A \to 2^\Act$ is an available action function for $\agent$;

    \item $\obs_A : (\Loc \times S_E) \to \Per$ is $\agent$'s perception function; 
    
    \item $\delta_A : (S_A \times \Act) \to \mathbb{P}(\Loc)$ is $\agent$'s probabilistic transition function;

    \item $\delta_E: (S_E \times \Act) \to \mathbb{P}(S_E)$ is a finitely-branching probabilistic  transition function for the environment.
\end{itemize}
\end{defi}

\iffalse
\noindent
NS-POMDPs are a subclass of continuous-state POMDPs with discrete observations (i.e., 
agent states $S_A$, which are pairs consisting of a local state and percept) and actions. This model captures a number of key properties of POMDP models that we target. The environment is continuous, as many real-world systems such as robot navigation are naturally modelled by continuous states, and probabilities are used to account for uncertainties. At the same time, the agent's state space is finite to ensure tractability.
\fi

The system executes as follows. 
A (global) state for an NS-POMDP $\pomdp$ comprises an agent state $s_A = (\loc, \per)$, 
where $\loc$ is its local state and $\per$ is the percept, and environment state $s_E$. In state $s=(s_A, s_E)$, the agent $\agent$ chooses an action $a$ available in $s_A$, 
%i.e.\ $a \in \Delta_A(s_A)$, 
then updates its local state to $\loc'$ according to the distribution $\delta_A(s_A,a)$. At the same time, the environment updates its state to $s_E'$ according to $\delta_E(s_E,a)$. Finally, the agent, based on $\loc'$ (since it may require different information regarding the environment depending on its local state), observes $s_E'$ to generate a new percept $\per' = \obs_A(\loc',s_E')$ and $\pomdp$ reaches the state $s'=((\loc', \per'), s_E')$.

While the NS-POMDP model admits any (deterministic) function $\obs_A$ from the continuous environment to a finite set of percepts, 
in this work we focus on neural perception functions {\revise represented by (trained) NN classifiers (see Section~\ref{background-sect}).}
While this restriction of perception to deterministic functions with discrete outputs is limiting, it is well aligned with NNs in applications such as object detection and localisation that we target.

{\revise
Using NN classifiers with PWL activation functions for neural perception also yields a polyhedral decomposition of the continuous state space, which can be obtained by computing the function's preimage \cite{KM-FF:20}.
Furthermore, as will be discussed in \sectref{sect:pwcreprs}, we impose mild restrictions on the transition probability function and reward structure for an NS-POMDP.
Together, these will later allow us to compute finite representations of lower and upper bounds for its value function over a (polyhedral) partition of its state space.}

\begin{rek}\label{rek1}
{\revise We note that our NS-POMDP model is a variant of conventional continuous-state (or hybrid) POMDPs, for which observation functions could also be defined using NNs.
In addition, since our model defines separate transition probability functions for the state space of the agent and the environment, there are also similarities with factored POMDPs~\cite{CB-DP:96}. % {MK-DK:99}
The key distinction between NS-POMDPs and both of these formalisms is the way observations are used, in particular the fact that they are are stored in the agent's local state as percepts and then used for decision making.
These introduce dependencies between the transitions of the agent and the environment that are not present in either conventional POMDPs or factored POMDPs.}
\end{rek}

\iffalse
{\revise
\begin{rek}[Positioning of NS-POMDPs]
    Computing value functions and synthesising strategies for general continuous-state POMDPs is challenging due to uncountable states. Our NS-POMDPs, as a subclass of continuous-state POMDPs with discrete observations (i.e., agent states $S_A$) and actions, use structured state spaces and transitions to allow for the development of efficient methods. This model also captures a number of key properties of POMDP models that we target. The environment is continuous, as many real-world systems such as robot navigation are naturally modelled by continuous states, and probabilities are used to account for uncertainties. At the same time, the agent's state space is finite to ensure tractability. 
    As in \figref{fig:neuro-symbolic-diagram}, if the perception is neural, following from the taxonomy of neuro-symbolic systems in \cite{HK:22}: if the environment state is symbolic, then NS-POMDPs fall into the design of Symbolic Neuro symbolic; if the environment state is nonsymbolic, then NS-POMDPs fall into the design of Neuro | Symbolic.
\end{rek}
}
\fi

\input{figures/neuro-symbolic-diagram}

To motivate our model, we consider a dynamic vehicle parking example, in which an autonomous vehicle uses {\revise a neural perception mechanism} %an NN 
for localisation while {\revise navigating to a preferred} parking spot. We are interested in automated synthesis of an optimal strategy to reach the {\revise chosen} spot.
{\revise Below, we formulate this as an NS-POMDP.
\figref{fig:neuro-symbolic-diagram} shows how its
neural and symbolic components interact.}

\begin{examp} \label{ex:parking:model}
%To illustrate the NS-POMDP model, 
{\revise Consider an agent $\agent$ (a vehicle) in a continuous environment $\mathcal{R} =\{(x, y) \in \mathbb{R}^2 \mid 0 \leq x, y \leq 4\}$, looking for one of two parking spots $\mathit{ps}_1,\mathit{ps}_2 \in \mathcal{R}$ (see \figref{fig:car_parking}, left). The vehicle uses an NN as a perception mechanism (\figref{fig:car_parking}, middle) that subdivides the continuous environment $\mathcal{R}$ into 16 cells, resulting in a grid-like abstraction of the environment. We trained a feed-forward NN classifier $f_\mathcal{R} : \mathcal{R} \rightarrow \mathbb{P}(\mathit{Grid})$, with one hidden ReLU layer and 14 neurons, on randomly generated data to take the coordinates of the vehicle as input and output a distribution over the 16 abstract grid {\revise cells} $\mathit{Grid} = \{(i,j)\mid i,j \in \{1, 2, 3, 4\}\}$. 
We assume that the agent can start from any position, has constant speed and initially has $\mathit{ps}_{1}$ as its preferred parking spot. It probabilistically changes its preferred parking spot if the other spot becomes closer to its current position.

Modelled as an NS-POMDP, the environment's state space corresponds to the continuous coordinates $\mathcal{R}$ of the vehicle.
Observations of its position, obtained from the NN classifier, are stored in the agent's percept
and its local state records the currently preferred parking spot.
Based on these, the agent chooses whether to move \emph{up}, \emph{down}, \emph{left} or \emph{right}, or \emph{park}.
Formally, we have the following.}

\begin{figure}[t]
%\vspace{-0.3cm}
\centering
\begin{subfigure}{0.3\textwidth}
\raisebox{-0.05\height}{\begin{tikzpicture}[scale = 0.6]
	bias = 4.5;
  	%\draw[white, thin, fill = green!30, opacity=0.8] (2, 3) rectangle (3, 4); 
  	%\draw[white, thin, fill = green!30, opacity=0.8] (3, 0) rectangle (4, 1); 

    \filldraw[color=green!40, fill=green!30, very thick](2.5,3.5) circle (0.14);
    \node  at (2,3.25) {\scriptsize $\mathit{ps}_{\scale{.75}{1}}$};

    \filldraw[color=green!40, fill=green!30, very thick](3.5,0.5) circle (0.14);
    \node  at (3,0.25) {\scriptsize $\mathit{ps}_{\scale{.75}{2}}$};

    \draw[black, very thick] (0,0) rectangle (4,4);
    %\draw[black, thick] (2, 4) -- (2, 3);
	%\draw[black, thick] (2, 3) -- (3, 3);
	%\draw[black, thick] (3, 3) -- (3, 4);
    %\draw[black, thick] (3, 0) -- (3, 1);
	%\draw[black, thick] (3, 1) -- (4, 1);
	\node at (3.95, -0.3) {\scriptsize$4$};
	\node at (-0.2, -0.3) {\scriptsize$0$};
	\node at (-0.25, 3.95) {\scriptsize$4$};
 	\node at (-0.25, 3) {\scriptsize$3$};
	\node at (-0.25, 2) {\scriptsize$2$};
	\node at (-0.25, 1) {\scriptsize$1$};
	\node at (3, -0.3) {\scriptsize$3$};
	\node at (2, -0.3) {\scriptsize$2$};
	\node at (1, -0.3) {\scriptsize$1$};
	\node [sedan top,body color=red!30,window color=black!80,minimum width=0.6cm, rotate = 90] at (0.5,0.5) {};
	\end{tikzpicture}}
\end{subfigure}
\hfil
\begin{subfigure}{0.3\textwidth}
\colorlet{myred}{red!80!black}
\colorlet{myblue}{blue!80!black}
\colorlet{mygreen}{green!60!black}
\colorlet{myorange}{orange!70!red!60!black}
\colorlet{mydarkred}{red!30!black}
\colorlet{mydarkblue}{blue!40!black}
\colorlet{mydarkgreen}{green!30!black}
\tikzstyle{node}=[thick,circle,draw=myblue,minimum size=10,inner sep=0.5,outer sep=0.6]
\tikzstyle{node in}=[node,green!20!black,draw=mygreen!30!black,fill=mygreen!25]
\tikzstyle{node hidden}=[node,blue!20!black,draw=myblue!30!black,fill=myblue!20]
\tikzstyle{node convol}=[node,orange!20!black,draw=myorange!30!black,fill=myorange!20]
\tikzstyle{node out}=[node,red!20!black,draw=myred!30!black,fill=myred!20]
\tikzstyle{connect}=[thick,mydarkblue] %,line cap=round
\tikzstyle{connect arrow}=[-{Latex[length=4,width=3.5]},thick,mydarkblue,shorten <=0.5,shorten >=1]
\tikzset{ % node styles, numbered for easy mapping with \nstyle
  node 1/.style={node in},
  node 2/.style={node hidden},
  node 3/.style={node out},
}
% [x=1.1cm,y=0.8cm]
\begin{tikzpicture}[x=1.1cm,y=0.8cm]
  % input nodes
  \node[node 1] (N1-1) at (0, 0) {\scriptsize$x$};
  \node[node 1] (N1-2) at (0, -1) {\scriptsize$y$};

  % hidden nodes 
  \node[node 2] (N2-1) at (1, 1) {};
  \node[node 2] (N2-2) at (1, 0) {};
  \node[node 2] (N2-3) at (1, -1) {};
  \node[node 2] (N2-4) at (1, -2) {};
  \foreach \i in {1,2}{
  \foreach \j in {1,...,4}{
    \draw[connect,white,line width=1.2] (N1-\i) -- (N2-\j);
    \draw[connect] (N1-\i) -- (N2-\j);
        }
  }
  % output nodes
  \node[node 3] (N3-1) at (2, 0.5) {\scriptsize$1$};
  \node[node 3] (N3-2) at (2, -0.5) {\scriptsize$2$};
  \node[node 3] (N3-3) at (2, -1.5) {\scriptsize \scriptsize$16$};
  \foreach \i in {1,...,4}{
  \foreach \j in {1,...,3}{
    \draw[connect,white,line width=1.2] (N2-\i) -- (N3-\j);
    \draw[connect] (N2-\i) -- (N3-\j);
        }
  }

  \path (N2-4) --++ (0,1.1) node[midway,scale=1.2] {$\vdots$};
  \path (N3-3) --++ (0,1.1) node[midway,scale=1.2] {$\vdots$};
  
  % \foreachitem \N \in \Nnod{ % loop over layers
  %   \def\lay{\Ncnt} % alias of index of current layer
  %   \pgfmathsetmacro\prev{int(\Ncnt-1)} % number of previous layer
  %   \foreach \i [evaluate={\y=\N/2-\i; \x=\lay; \n=\nstyle;}] in {1,...,\N}{ % loop over nodes
      
  %     % NODES
  %   \node[node \n] (N\lay-\i) at (\x,\y) {};
  %     % CONNECTIONS
  %     \ifnum\lay>1 % connect to previous layer
  %       \foreach \j in {1,...,\Nnod[\prev]}{ % loop over nodes in previous layer
  %         \draw[connect,white,line width=1.2] (N\prev-\j) -- (N\lay-\i);
  %         \draw[connect] (N\prev-\j) -- (N\lay-\i);
  %         %\draw[connect] (N\prev-\j.0) -- (N\lay-\i.180); % connect to left
  %       }
  %     \fi % else: nothing to connect first layer
      
  %   }
  %   \ifnum\lay>1
  %   \path (N\lay-\N) --++ (0,1+\yshift) node[midway,scale=1.5] {$\vdots$};
  %   \fi
  % }
\end{tikzpicture}
\end{subfigure}
\hfil
\begin{subfigure}{0.24\textwidth}
\includegraphics[width=\textwidth] {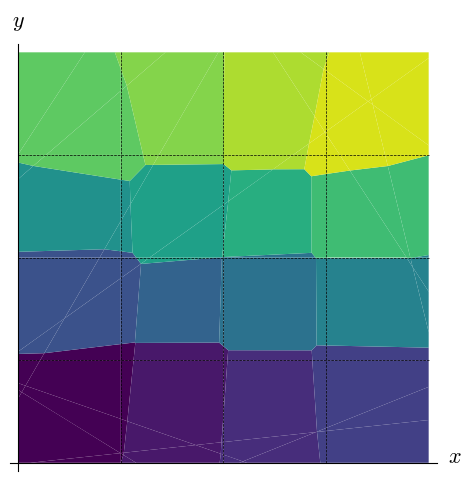}
\put(-96,-3){\scriptsize $0$}
\put(-71,-3){\scriptsize $1$}
\put(-52,-3){\scriptsize $2$}
\put(-32,-3){\scriptsize $3$}
\put(-13,-3){\scriptsize $4$}
\put(-97,22){\scriptsize $1$}
\put(-97,42){\scriptsize $2$}
\put(-97,62){\scriptsize $3$}
\put(-97,82){\scriptsize $4$}
\end{subfigure}
\vspace*{-0.0cm}
\caption{Car parking example, perception NN and perception FCP of its preimage consisting of 62 polygons and 16 classes.}\label{fig:car_parking}
\vspace*{-0.0cm}
\end{figure}
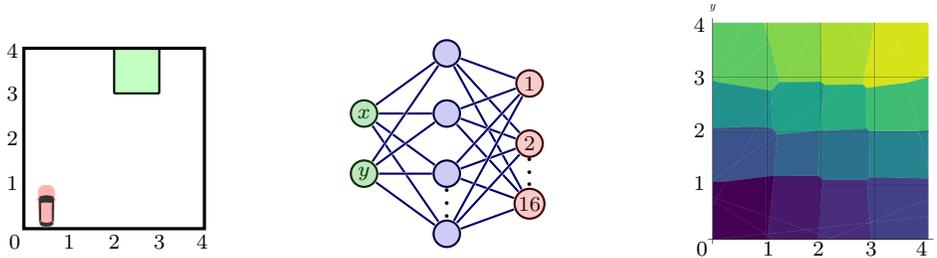

\begin{itemize}
	\item $S_A= \Loc \times \Per$, where {(local states) \revise $\Loc = \{ \mathit{ps}_{1} , \mathit{ps}_{2} \}$ are the two parking spots and (percepts) $\Per = \mathit{Grid}$} are the 16 abstract grid {\revise cells}.
 % which are ordered according to \tabref{tab:advisory-parking}.

	\item $S_E = \mathcal{R}$.
	
	\item $Act = \{ \mathit{up}, \mathit{down}, \mathit{left}, \mathit{right}, \mathit{park} \}$.

	\item {\revise For $(\mathit{ps},\per) \in S_A$, we have $\Delta_A(\mathit{ps}, \per) = \Act$ if $f_\mathcal{R}^{\max}(\mathit{ps}) = \per$ and $\Delta_A(\mathit{ps}, \per) = \{ \mathit{up}, \mathit{down}, \mathit{left}, \mathit{right} \}$ otherwise,  i.e., the agent can only choose to park when it perceives it is in its preferred parking spot.}

  	 \item {\revise For each $\mathit{ps} \in \Loc$ and $(x,y) \in S_E$, we have $\obs_A(\mathit{ps},(x, y))= f_\mathcal{R}^{\max}(x,y)$, i.e., independently of the local state of the agent, the perception function returns the perceived grid cell of the agent under the classifier $f_\mathcal{R}$  (\figref{fig:car_parking}, middle). The boundary coordinate is resolved by assigning the grid cell at the bottom left.}

    \item {\revise For $s_A=(\mathit{ps}, \per) \in S_A$ and $a \in \Act$,
       to define $\delta_A$ we have the following two cases to consider:
        \begin{itemize}
        \item
        if $\| f_\mathcal{R}^{\max}(\mathit{ps}) - \per \|_2 > \|  f_\mathcal{R}^{\max}(\mathit{ps}') - \mathit{per} \|_2$, where $\mathit{ps}'$ is the parking spot not currently preferred by the agent and $\|\cdot\|_2$ is the Euclidean norm,
        then $\delta_A(s_A, \alpha)(\mathit{ps}) = \delta_A(s_A, \alpha)(\mathit{ps}') = 0.5$,
        i.e., if the currently preferred parking spot is further away than the other spot, the agent changes its preferred spot with probability 0.5;
        \item
        otherwise $\delta_A(s_A, \alpha)(\mathit{ps}) = 1$, i.e., $\agent$ sticks with the preferred parking spot. 
        \end{itemize}}

    % where $\lambda = 0.5$. \gethinM{do we need $\lambda$ here?}
    
    \item {\revise For $(x,y),(x',y') \in S_E$ and $a \in \Act$, we define:
     % \[ \begin{array}{rcl}
     % x'' & = & x - \Delta t d_{\mathit{ax}} \\
     % y'' & = & y + \Delta t d_{\mathit{ay}}
     % \end{array} \]  
    \[
     \delta_E((x, y),a)(x', y') = \left\{ \begin{array}{cl}
     1 & \mbox{if $a = \mathit{up}$ and $(x',y')=(x,y{+} \Delta t)$} \\
     1 & \mbox{if $a = \mathit{down}$ and $(x',y')=(x,y{-} \Delta t)$} \\
     0.8 & \mbox{if $a = \mathit{left}$ and $(x',y')=(x{-} \Delta t, y)$} \\
     0.2 & \mbox{if $a = \mathit{left}$ and $(x',y')=(x{-} \Delta t,y+ \Delta t)$} \\
     0.8 & \mbox{if $a = \mathit{right}$ and $(x',y')=(x{+} \Delta t, y)$} \\
     0.2 & \mbox{if $a = \mathit{right}$ and $(x',y')=(x{+} \Delta t,y{+} \Delta t)$} \\
     0 & \mbox{otherwise}
     \end{array}     \right.
     \]
      where $\Delta t = 1.0$ is the time step. Here the movement \emph{left} and \emph{right} is probabilistic, as it can result in the agent instead moving on the diagonal with probability 0.2. We have also assumed that the agent's coordinates remain within the environment boundaries. 
      }   \hfill$\blacksquare$
\end{itemize}
\end{examp}
{\revise To simplify the presentation, the perception function of the agent in \egref{ex:parking:model} is independent of its local state. However, our modelling formalism admits more complex scenarios where this is not the case, for example, if the agent were to switch to a more accurate sensor as it approaches its preferred parking spot to ensure that it parks there.}

\iffalse
\begin{itemize}

	\item The agent state space consists of 5 trust levels $\Loc = \{ 1, 2, \dots, 5\}$ (local states) and 16 abstract grid points $\Per = \{ 1, 2, \dots, 16 \}$ (percepts).

	\item The set of environment states is $\mathcal{R}$ corresponding to the continuous coordinate vector of the vehicle.

	\item The agent's actions are to move \emph{up}, \emph{down}, \emph{left} or \emph{right}, or \emph{park}, with the restriction that \emph{park} can only be taken when the parking spot is reached.s

	 \item The observation function of the agent is implemented via an NN with one ReLU hidden-layer, takes the coordinate vector of the vehicle as input and the outputs one of the 16 abstract grid points (\figref{fig:car_parking}, middle). 

    \item The transition function of the agent increases the trust level if the percept is compliant with the executed action and decreases probabilistically otherwise.

    \item The transition function of the environment corresponds to the vehicle moving in the direction specified by the agent for a fixed time step, unless the vehicle moves outside of $\mathcal{R}$ in which case it does not move. 
\end{itemize}
\fi

\noindent
\startpara{NS-POMDP semantics}
The semantics of an NS-POMDP $\pomdp$ is a POMDP $\sem{\pomdp}$ over the product of the (discrete) states of the agent and the (continuous) states of the environment, except that we restrict those to states that are \emph{percept compatible}.
A state $s = ((\loc, \per), s_E)$ is percept compatible if $\per = \obs_A(\loc, s_E)$. {\revise Percept compatibility indicates that the agent always accesses its percept via the perception function.}
% and accesses the environment state only through the perception function.} 
The semantics of an NS-POMDP is closed with respect to percept compatible states. {\revise We would like to emphasise that (see Remark~\ref{rek1}) NS-POMDPs differ from conventional POMDPs in that the observations, $S_A$, are stored in the state space.}

\begin{defi}[Semantics of NS-POMDPs]\label{semantics-def}
Given an NS-POMDP $\pomdp$, the semantics of $\pomdp$ is the POMDP $\sem{\pomdp} = (S,\Act,\Delta,\delta,S_A,\obs)$ where:
\begin{itemize}
    \item $S \subseteq S_A \times S_E$ is the set of percept compatible %, i.e., if $s$ is percept compatible, then any state $s'$ reachable from $s$ is also percept compatible.}) 
    states,
    which contain both discrete and continuous elements;
    
    \item $\Delta(s_A,s_E) = \Delta_A(s_A)$ for $(s_A,s_E)\in S$;
  
    \item 
    % if $s_A'=(\loc',\per')$ is percept compatible, 
    % such that $\per'=\obs_A(\loc',s_E')$, 
    {\revise $\delta(s, a)(s') = \delta_A(s_A, a)(\loc') \delta_E(s_E, a)(s_E')$ and $\delta(s,a)(s')=0$ otherwise for $s=(s_A,s_E),s'=(s_A',s_E') \in S$ and $a \in \Delta(s)$;}
    
    \item $\obs(s_A,s_E)=s_A$ for $(s_A,s_E)\in S$. 
\end{itemize}
\end{defi}

\noindent
{\revise Note that the transition probability function $\delta$ has finite branching, since $S_A$ is finite and $\delta_E$ is required to have finite branching. Nevertheless, the underlying state space of an NS-POMDP is uncountable.}

\begin{examp} \label{ex:parking:rew}
{\revise We return to the vehicle parking NS-POMDP of \egref{ex:parking:model}. Suppose the objective of the agent is to try and park at its currently preferred
parking spot as quickly as possible. We can represent this scenario using a discounted cumulative reward objective where, in the corresponding
reward structure, all action rewards are zero, and for the state rewards there is a positive reward if the agent is perceived to have reached its preferred parking spot and is zero otherwise. Formally, we have for any $s = ((\mathit{ps},\per),s_E) \in S$: 
\[
r_S(s) = \left\{ \begin{array}{cl}
1000 & \mbox{if $f_\mathcal{R}^{\max}(\mathit{ps}) = \per$} \\
0 & \mbox{otherwise.}
\end{array} \right.
\]
Decreasing the value of the discount factor $\beta$ will lead to greater loss in the accumulated rewards if the agent fails to park quickly. \hfill$\blacksquare$}
\end{examp}

\startpara{NS-POMDP strategies}
As $\sem{\pomdp}$ is a POMDP, we consider \emph{observation-based strategies},  which can be represented by memoryless strategies over its belief MDP $\sem{\pomdp}_B$. 
Given agent state $s_A = (\loc, \per)$, we let $S_E^{s_A}= \{ s_E \in S_E \mid \obs_A(\loc, s_E) = \per \}$, i.e., the environment states generating percept $\per$ given $\loc$. 
Since agent states are observable and states of $\sem{\pomdp}$ are percept compatible, beliefs can be represented as pairs $(s_A,b_E)$, where $s_A \in S_A$ is an agent state, $b_E \in \mathbb{P}(S_E)$ is a belief over environment, and $b_E(s_E) = 0$ for all $s_E \in S_E \setminus S_E^{s_A}$, i.e., those states that are not percept compatible. 
\iffalse
Since the agent's state is observable and the states of $\pomdp$ are percept compatible, we can represent a belief as a pair $(s_A,b_E)$, where $s_A \in S_A$ and $b_E: S_E \to \mathbb{R}_{\ge 0}$ is such that $\int_{s_E \in S_E} b_E(s_E) \textup{d} s_E  = 1$ and $b_E(s_E) = 0$ for all $s_E \notin S_E^{s_A}$, such that, if $s_A = (\loc, \per)$, then $S_E^{s_A}$ is the set of environment states generating the percept $\per$ given the local state $\loc$, i.e., $S_E^{s_A} = \{ s_E \in S_E \mid \obs_A(\loc, s_E) = \per \}$. 
\fi

Before giving the definition of $\sem{\pomdp}_B$, we consider how beliefs are updated in this setting. % of NS-POMDPs. 
Therefore, suppose $s_A$ is the current agent state, i.e., {\revise stores} what is observable, and $b_E$ is the current belief about the environment. Then if action $a$ is executed and $s_A'$ is observed, the updated belief %, denoted $b_E^{s_A,a,s_A'}$, 
is such that for any $s_E' \in S_E$:
\begin{equation}\label{eq:belief-update1}
\!\! b_E^{s_A,a,s_A'}(s_E') = \frac{P((s_A',s_E') \mid (s_A,b_E), a)}{P(s_A' \mid (s_A,b_E), a)} \; \mbox{if  $s_E' \in S_E^{s_A'}$ and $0$ otherwise.} \!\!
\end{equation}

\startpara{Belief MDP and belief updates}
We can now derive the belief MDP of an NS-POMDP, which follows through a standard construction~\cite{LPK-MLL-ARC:98} while relying on Borel measurability of the underlying uncountable state space of the NS-POMDP.

\begin{defi}[Belief MDP]\label{beliefmdp-def}
The belief MDP of an NS-POMDP $\pomdp$ is the MDP $\sem{\pomdp}_B = (S_B,\Act,\Delta_B,\delta_B)$, where:
\begin{itemize}
\item
$S_B \subseteq S_A \times \mathbb{P}(S_E)$ is the set of percept compatible beliefs;
\item
$\Delta_B(s_A,b_E)= \Delta_A(s_A)$ for $(s_A,b_E) \in S_B$;
\item
for $(s_A,b_E),(s_A',b_E') \in S_B$, 
and  $a \in \Delta_B(s_A,b_E)$:
\[
\delta_B((s_A,b_E),a)(s_A',b_E') = \left\{ \begin{array}{cl}
P(s_A' \mid (s_A,b_E), a) & \mbox{if $b_E' = b_E^{s_A,a,s_A'}$} \\
0 & \mbox{otherwise.}
\end{array} \right.
\]
\end{itemize}
\end{defi}

\noindent
Finally, in this section we discuss how the beliefs and probabilities of \defiref{beliefmdp-def} can be computed. 
For any $(s_A,b_E),(s_A',b_E') \in S_B$ and $s_A'=(\loc',\per')$, we have that $P(s_A' \mid (s_A,b_E), a)$ equals:
{ \revise 
\begin{align}\label{eq:probability-obs-agent-state}
\delta_A(s_A, a)(\loc') \left(\mbox{$\int_{s_E \in S_E}$} b_E(s_E) \mbox{$\int_{s_E' \in S_E^{s_A'}}$}  \delta_E(s_E, a)(s_E') \textup{d} s_E \right).
\end{align}
}
Furthermore, $P((s_A',s_E') \mid (s_A,b_E), a)$ equals:
\begin{equation} \label{eq:new-agent-obs}
\delta_A(s_A, a)(\loc')  \left(\mbox{$\int_{s_E \in S_E}$} b_E(s_E) \delta_E(s_E,a)(s_E') \textup{d} s_E \right) \end{equation}
if $s_E' \in S_E^{s_A'}$ and 0 otherwise.
Thus, using \eqnref{eq:belief-update1} we have that $b_E^{s_A,a,s_A'} (s_E')$ equals:
{\revise
% for any $s_E' \in S_E$, $b_E^{s_A,a,s_A'} (s_E')$ equals
\begin{align}\label{eq:belief-update}
\frac{\int_{s_E \in S_E} b_E(s_E) \delta_E(s_E,a)(s_E') \textup{d} s_E}{\int_{s_E \in S_E}  b_E(s_E) \mbox{$\int_{s_E'' \in S_E^{s_A'}}$} \delta_E(s_E, a)(s_E'') \textup{d} s_E} \mbox{ if  $s_E' \in S_E^{s_A'}$ and $0$ otherwise.}
\end{align}
} 
We note that the belief MDP $\sem{\pomdp}_B$ is continuous and infinite-dimensional, with finite branching. Thus, solving it exactly is intractable as closed-form operations and parametric forms for continuous functions are required. For efficient computation, beliefs also need to be in closed form.

\section{Value Iteration}\label{sec:value-iteration}

A common approach to solving continuous-state POMDPs is to discretise or approximate the continuous components with a grid and use methods for finite-state POMDPs. As this may compromise accuracy and leads to an exponential growth in the number of
states, we instead aim to operate directly in the continuous domain. 
Since
functions over continuous spaces can have arbitrary forms not amenable to
computation, we will extend $\alpha$-functions to the setting of NS-POMDPs, aided by the theoretical formulation of~\cite{JMP-NV-MTS-PP:06}, where it was proved that continuous-state POMDPs
with discrete observations and actions have a piecewise linear and convex
value function.
Rather than work with Gaussian mixtures as in~\cite{JMP-NV-MTS-PP:06}, which would require approximations, we will directly exploit the structure of the model to induce a finite (polyhedral) representation of the value function,
{\revise similarly to~\cite{ZF-RD-NM-RW:04}, which has been shown to outperform discretisation.}

More specifically, in this section we show that {\em piecewise constant} representations for the perception, reward and transition functions are sufficient for NS-POMDPs under mild assumptions, in the sense that they offer a finite representation and are closed with respect to belief update and the Bellman operator. We next propose a value iteration (VI) algorithm that utilises piecewise constant $\alpha$-functions, which does not scale but serves as a basis for designing a practical point-based algorithm in Section~\ref{sec:HSVI}. We conclude this section by investigating the convexity and continuity of the value function.

\subsection{Value functions}
We work with the belief MDP $\sem{\pomdp}_B = (S_B,\Act,\Delta_B,\delta_B)$ of an NS-POMDP $\pomdp$ and consider discounted cumulative reward objectives $Y$. 
%Before giving a fixed point characterisation of 
The {\em value function} is given by $V^\star : S_B \to \mathbb{R}$, where $V^\star(s_A, b_E) = \mathbb{E}_{(s_A,b_E)}^{\sigma^\star}[Y]$ for all $(s_A, b_E) \in S_B$.
We require the following notation to evaluate beliefs through a function over the state space $S$.
Given $f: S \to \mathbb{R}$ and belief $(s_A,b_E)$, let: 
\begin{equation}\label{expectation-eq}
\langle f, (s_A,b_E) \rangle = \mbox{$\int_{s_E \in S_E}$} f(s_A,s_E) b_E(s_E) \textup{d}s_E
\end{equation}
for which an integral over $S_E$ is required.

Recall that $\mathbb{F}(S_B)$ denotes the space of functions over the beliefs.

\begin{defi}[Bellman operator]\label{max-def}
Given $V \in \mathbb{F}(S_B)$, the operator $T : \mathbb{F}(S_B) \rightarrow \mathbb{F}(S_B)$ is defined as follows: $[TV](s_A, b_E)$ equals
\begin{equation}
\begin{aligned}
\max\limits_{a  \in \Delta_A(s_A)} \left\{ \langle R_{a}, (s_A, b_E) \rangle + \beta \mbox{$\sum_{s_A' \in S_A}$}  P(s_A' \mid (s_A, b_E), a) V ( s_A',b^{s_A,a,s_A'}_E ) \right\}
\label{eq:max-operator}
\end{aligned}
\end{equation}
for $(s_A, b_E) \in S_B$, where $R_{a}(s) = r_A(s, a) + r_S(s)$ for $s \in S$.
\end{defi}
Since $\sem{\pomdp}_B$, the semantics of NS-POMDP $\pomdp$, is a continuous-state POMDP with discrete observations and actions, according to \cite{JMP-NV-MTS-PP:06} the value function $V^\star$ is the unique fixed point of the operator $T$, and thus, theoretically, value iteration can be used to compute $V^\star$. However, as the functions involved are defined over probability density functions from $\mathbb{P}(S_E)$ and $S_E$ is a continuous space, to ensure feasible computation we require a finite parameterisable representation for the value function. To this end, we will extend the class of $\alpha$-functions with special structure introduced for continuous-state POMDPs in \cite{JMP-NV-MTS-PP:06}, which generalise  $\alpha$-vector representations for finite-state POMDPs~\cite{EJS:78}.

\subsection{PWC Representations}\label{sect:pwcreprs}

{\revise We first recall that the class of perception functions we consider (see Section~\ref{background-sect})} induces a finite partition 
% (FCP) 
of the continuous state space, consisting of connected and observationally-equivalent {\em regions}, and obtained as the preimage of the perception function.
We then impose mild assumptions on the NS-POMDP structure (\aspref{asp:transitions-rewards}) to ensure that 
the agent and environment transition functions preserve the PWC properties of this partition, and on the reward function to ensure region-based reward accumulation.

\begin{lema}[Perception FCP]\label{lema:perception-fcp}
     There exists a smallest FCP of $S$, called the perception FCP, denoted $\Phi_{P}$, such that all states in any $\phi \in \Phi_{P}$ are observationally equivalent, i.e., if $(s_A,s_E),(s_A',s_E')\in \phi$, then $s_A=s_A'$ and we let $s_A^\phi= s_A$.
\end{lema}
The perception FCP $\Phi_P$ can be used to find the set $S_E^{s_A'}$ for any agent state $s_A' \in S_A$ over which we integrate beliefs 
in closed form; see e.g., \eqnref{eq:probability-obs-agent-state} and \eqnref{eq:belief-update}.

{\revise The assumption we make about an NS-POMDP's structure requires that, given a finite decomposition $\Phi'$ of the state space into regions (i.e., an FCP of the state space), there exists another finite decomposition $\Phi$, called the \emph{preimage FCP}, such that states in regions of $\Phi$ have the same rewards and transition probabilities into regions of  $\Phi'$.
Furthermore, we require that transitions of the (continuous) environment must also be decomposable into regions.}
{\revise Our assumption is formally stated below.}

{\revise
\begin{asp}[Transitions and rewards]\label{asp:transitions-rewards} For any given FCP $\Phi'$ of $S$, there exists an FCP $\Phi$ of $S$, called the \emph{preimage FCP} of $\Phi'$, such that for all $s \in \phi \in \Phi$, $a \in \Act$ and $s' \in \phi' \in \Phi'$,
\begin{itemize}
    \item $\delta(s, a)(s') = \delta_{\Phi}(\phi, a)(\phi')$;

    \item $r_A(s, a) = r_{\Phi,A}(\phi, a)$;

    \item $r_S(s) = r_{\Phi,S}(\phi)$;
\end{itemize}
for functions $\delta_{\Phi} : \Phi \times \Act \to \mathbb{P}(\Phi'), r_{\Phi, A} : \Phi \times \Act \to \mathbb{R}$ and $r_{\Phi, S} : \Phi \to \mathbb{R}$. Furthermore, there exists $n \in \mathbb{N}$ such that $\delta_E = \sum_{i=1}^{n} \mu_i \delta_E^i$, where $\sum_{i = 1}^{n} \mu_i = 1$, $\mu_i \geq 0$ and $\delta_E^i : (S_E {\times} \Act) \to S_E$ is piecewise continuous for each action in $\Act$ and $1 \leq i \leq n$.
\end{asp}
This assumption allows us to compute finite representations of lower and upper bounds of an NS-POMDP's value function over a (polyhedral) partition of the state space. This partition is created dynamically during the iterations of the solution, using a preimage-based splitting operation (see Algorithm~\ref{alg:ISPP-backup}). In this construction, using Assumption~\ref{asp:transitions-rewards}, we can assume for any action $a$ that there exists a smallest FCP of $S$, called the \emph{reward FCP under action $a$} and denoted $\Phi_{R}^a$, such that all states in any $\phi \in \Phi_{R}^a$ have the same state rewards and action rewards when $a$ is performed.}

We emphasize that, although the states in any region of the perception FCP are observationally equivalent, by \aspref{asp:transitions-rewards} {\revise and, since states in the same region of the preimage FCP of the perception FCP have the same transitions and rewards to these regions, the perception FCP and the preimage FCP need not be the same. This means that} such states can still have different values,  since taking the same actions can yield paths that need not be observationally equivalent. Therefore, the value function $V^{\star}$ may not be piecewise constant. 
Our results demonstrate that analysing NS-POMDPs under these PWC restrictions remains challenging.
\begin{examp} \label{ex:parking:fcp}
\figref{fig:car_parking} (right) shows an FCP representation for the preimage of the perception function of \egref{ex:parking:model}. The FCP was constructed via the exact computation method from \cite{KM-FF:20}, and is composed of 62 polygons. Each colour indicates one of the grid cells as perceived by the agent. 

{\revise We remark that Assumption~\ref{asp:transitions-rewards} holds for this example since the agent's transition function and the reward functions are PWC and the environment's transition function is PWL.} \hfill$\blacksquare$
\end{examp}

\subsection{PWC \texorpdfstring{$\alpha$}{α}-Function Value Iteration}

We can now show, utilising the results %$\alpha$-function representations 
for continuous-state POMDPs~\cite{JMP-NV-MTS-PP:06}, that $V^{\star}$ is the limit of a sequence of $\alpha$-functions, called \emph{piecewise linear and convex under PWC $\alpha$-functions (P-PWLC)}, where each such function can be represented by a (finite) set of PWC functions (concretely, as a finite set of FCP regions and a value vector).

\begin{defi}[P-PWLC function]\label{defi:PWLC}
A function $V: S_B \to \mathbb{R}$ is \emph{piecewise linear and convex under PWC $\alpha$-functions (P-PWLC)} if there exists a finite set $\Gamma \subseteq \mathbb{F}_C(S)$ such that $V(s_A, b_E) = \max_{\alpha \in \Gamma} \langle \alpha, (s_A, b_E) \rangle$ for all $(s_A, b_E) \in S_B$, 
where the functions in $\Gamma$ are called PWC \emph{$\alpha$-functions}.
\end{defi}
\defiref{defi:PWLC} implies that, if $V \in \mathbb{F}(S_B)$ is P-PWLC, then it can be represented by a set $\Gamma$ of PWC continuous functions over $S$. 
{\revise Similarly to~\cite{ZF-RD-NM-RW:04}, where a rectangular PWC representation of value functions was proposed and proved to be closed under the Bellman backup for a class of structured continuous finite-horizon MDPs, we demonstrate that, for NS-POMDPs satisfying Assumption~\ref{asp:transitions-rewards}, 
% a PWC $\alpha$-function 
a P-PWLC representation of value functions is closed under the Bellman operator and the value iteration algorithm converges.}

\begin{thom}[P-PWLC closure and convergence]\label{thom:PWC-consistency}
If $V \in \mathbb{F}(S_B)$ and P-PWLC, then so is $[TV]$. If $V^0 \in \mathbb{F}(S_B)$ and P-PWLC, then the sequence $(V^t)_{t = 0}^{\infty}$, such that $V^{t+1} = [TV^t]$, is P-PWLC and converges to $V^\star$. 
\end{thom}

\noindent
We remark that an implementation of this exact value iteration is feasible {\revise in principle}, since each $\alpha$-function involved is PWC and thus allows for a finite representation. 
However, as the number of $\alpha$-functions grows exponentially in the number of agent states, 
it is not scalable {\revise in practice}.

\subsection{Convexity and Continuity of the Value Function}

In \sectref{sec:HSVI}, we will derive a variant of HSVI for lower and upper bounding of the value function, which is more scalable. To this end, the following properties will be required.

Using \thomref{thom:PWC-consistency} the value function can be represented as a pointwise maximum $V^{\star}(s_A, b_E) = \sup_{\alpha \in \Gamma} \langle \alpha, (s_A, b_E) \rangle$ for $(s_A, b_E) \in S_B$, where $\Gamma \subseteq \mathbb{F}_C(S)$ may be infinite. We now show that $V^{\star}$ is convex and continuous for any fixed $s_A \in S_A$.  
Since we assume bounded reward functions, the value function $V^{\star}$ has lower and upper bounds:
\begin{equation}\label{lu-eqn}
L = \min\nolimits_{s \in S, a \in \Act} R_a(s) /(1- \beta) \quad \mbox{and} \quad
U = \max\nolimits_{s \in  S, a \in \Act} R_a(s)/(1- \beta)  \, .
\end{equation}

\begin{thom}[Convexity and continuity]\label{thom:continuity}
For any $s_A \in S_A$, the value function $V^{\star}(s_A, \cdot) : \mathbb{P}(S_E) \rightarrow \mathbb{R}$ is convex and for any $b_E, b_E' \in  \mathbb{P}(S_E)$:
\begin{equation}\label{eq:K-definition}
    |V^{\star}(s_A, b_E) - V^{\star}(s_A, b_E')| \leq K(b_E, b_E')
\end{equation}
{\revise where $K(b_E, b_E') = \frac{1}{2}(U - L) \int_{s_E \in S_E^{s_{\scale{.75}{A}}}} | b_E(s_E) - b_E'(s_E)| \textup{d}s_E$.}
\end{thom}

\section{Heuristic Search Value Iteration}\label{sec:HSVI}

Value iteration with point-based updates has been proposed for finite-state POMDPs \cite{JMP-NV-MTS-PP:06,TS-RS:04,GS-JP-RK:13,TS-RS:05,KH-BB-KC:18}, relying on the fact that performing many fast approximate updates often results in a more useful value function than performing a few exact updates. HSVI \cite{TS-RS:04} approximates 
$V^\star$ at a given initial belief via lower and upper bound functions, which are updated through heuristically generated beliefs. SARSOP \cite{HK-DH-WL:08} improves efficiency,
but sacrifices convergence guarantees due to aggressive pruning. 
These methods focus on finite-state POMDPs and are not directly applicable to continuous-state NS-POMDPs, as they rely on discretisation or approximation. 

We now present a new HSVI algorithm for NS-POMDPs, which uses P-PWLC functions and
belief-value induced functions to approximate $V^{\star}$ from below and above.  This HSVI algorithm progressively subdivides the continuous state space during value backups, to obtain a piecewise constant lower bound and a lower $K$-Lipschitz envelope of a convex hull upper bound on $V^{\star}$ that itself may not be piecewise constant.

We first introduce the representations of the lower and upper bound functions to the value function, then present point-based updates followed by our HSVI algorithm, and finally consider two belief representations for the implementation, both with closed forms for the quantities of interest,
one based on particles (individually sampled points) and the other on regions (polyhedra) of the continuous space.

\subsection{Lower and Upper Bound Representations}\label{subsec:lb_up}
\startpara{Lower bound function} 
Selecting an appropriate representation for $\alpha$-functions requires closure properties with respect to the Bellman operator, which includes both the transition function and the reward function. 
Rather than relying on Gaussian mixtures~\cite{JMP-NV-MTS-PP:06}, which require both the transition and reward functions to be in this form,  
we represent the lower bound $V_{\mathit{lb}}^{\Gamma} \in \mathbb{F}(S_B)$ as a P-PWLC function
% a point-wise maximum
for the finite set $\Gamma \subseteq \mathbb{F}_C(S)$ of PWC $\alpha$-functions (see \defiref{defi:PWLC}),
for which closure is guaranteed by \thomref{thom:PWC-consistency}. This is finitely representable as each $\alpha$-function is PWC. 
In contrast to Gaussian mixtures, our P-PWLC representation is designed to match the NS-POMDP perfectly, with the necessary closure properties ensured by exploiting the structure of the NS-POMDP.

\startpara{Upper bound function} The upper bound $V_{\mathit{ub}}^{\Upsilon} \in \mathbb{F}(S_B)$ is represented by a finite set of belief-value points $\Upsilon = \{ ((s_A^i, b_E^i), y_i) \mid i \in I \}$, where $y_i$ is an upper bound of $V^{\star}(s_A^i, b_E^i)$. Since $V^{\star}(s_A, \cdot)$ is convex by \thomref{thom:continuity}, letting $I_{s_{\scale{.75}{A}}} = \{ i \in I \mid s_A^i=s_A \}$, for any $\lambda_i \ge 0$ such that $\sum_{i\in I_{s_{\scale{.75}{A}}}} \lambda_i = 1$, we have:
\begin{equation}\label{eq:v-convexity-upper-bound}
   V^{\star}(s_A, \mbox{$\sum_{i\in I_{s_{\scale{.75}{A}}}}$} \lambda_i b_E^i) \leq \mbox{$\sum_{i\in I_{s_{\scale{.75}{A}}}}$} \lambda_i V^{\star}(s_A^i, b_E^i) \leq \mbox{$\sum_{i\in I_{s_{\scale{.75}{A}}}}$} \lambda_i y_i \, .
\end{equation}
This fact is used in HSVI for finite-state POMDPs \cite{TS-RS:04}, as any new belief is a convex combination of the beliefs in $\Upsilon$, and therefore the convexity of $V^{\star}(s_A, \cdot)$ yields an upper bound.
{\revise However, since $\Upsilon$ is a finite set and in NS-POMDPs each belief %in $S_B$ 
is over infinitely many states, any convex combination of beliefs in $\Upsilon$ cannot cover the belief space $S_B$, and hence \eqref{eq:v-convexity-upper-bound} cannot be used to generate an upper bound.}
{\revise We therefore instead define} the upper bound $V_{\mathit{ub}}^{\Upsilon}$ as the lower envelope of the lower convex hull of the points in $\Upsilon$ satisfying the following problem:
\begin{align}
     V_{\mathit{ub}}^{\Upsilon}(s_A, b_E)  = &\textup{ minimize} \; \mbox{$\sum\nolimits_{i \in I_{s_{\scale{.75}{A}}}}$}\lambda_i y_i + K_{\mathit{ub}}(b_E, \mbox{$\sum\nolimits_{i \in I_{s_{\scale{.75}{A}}}}$} \lambda_i b_E^i) \nonumber  \\
     &\textup{ subject to: }  \lambda_i \ge 0, \mbox{$\sum\nolimits_{i \in I_{s_{\scale{.75}{A}}}}$}  \lambda_i = 1 \; \mbox{for all $(s_A, b_E) \in S_{B}$} \label{eq:new-ub}
\end{align} 
where $K_{\mathit{ub}} : \mathbb{P}(S_E) \times \mathbb{P}(S_E) \to \mathbb{R}$ measures the difference between two beliefs such that, if $K$ is from \thomref{thom:continuity} showing the continuity of the value function, then for any $b_E, b_E' \in \mathbb{P}(S_E)$:
\begin{align}\label{eq:K-UB-condition}
    K_{\mathit{ub}}(b_E, b_E') \ge K(b_E, b_E') \quad \mbox{and} \quad  K_{\mathit{ub}}(b_E, b_E) = 0 \, .
\end{align}
It can be seen that \eqref{eq:new-ub} is close to the classical upper bound function used in regular HSVI for finite-state spaces, except for the function $K_{\mathit{ub}}$ that measures the difference between two beliefs (two functions). 
We require that $K_{\mathit{ub}}$ satisfies \eqref{eq:K-UB-condition} to ensure that \eqref{eq:new-ub} is an upper bound after a value backup, as stated in \lemaref{lema:upper-bound-update} below.

\startpara{Bound initialisation} The lower bound $V_{\mathit{lb}}^{\Gamma}$ is initialised using the lower bound of the blind strategies of the form ``always choose action $a \in \Act$'', which is
given by $\sum_{k = 0}^{\infty} \beta^k \inf_{s \in S}R_a(s)$. Therefore, a lower bound for $V_{\mathit{lb}}^{\Gamma}$ is given by:
\[ R_\mathit{LB} = 
\max\nolimits_{a \in \Act} \left( \mbox{$\sum\nolimits_{k = 0}^{\infty}$} \beta^k \inf\nolimits_{s \in S}R_a(s) \right) 
= 1/(1 - \beta) \max\nolimits_{a \in \Act}\inf\nolimits_{s \in S}R_a(s) \, .
\]
The PWC $\alpha$-function set $\Gamma$ for the initial $V_{\mathit{lb}}^{\Gamma}$ contains a single PWC function $\alpha$,  where $\alpha(s) = R_\mathit{LB}$ for all $s \in S$ and the associated FCP is the perception FCP $\Phi_P$. We initialise the upper bound $V_{\mathit{ub}}^{\Upsilon}$ by sampling a set of initial beliefs $\{(s_A^i, b_E^i)\}_{i \in I}$ and letting $y_i = U$ for all $(s_A^i, b_E^i)$.

\subsection{Point-Based Updates}
\startpara{Lower bound updates} For the lower bound $V_{\mathit{lb}}^{\Gamma}$, 
in each iteration we add a new PWC $\alpha$-function $\alpha^{\star}$ to $\Gamma$ leading to $\Gamma'$
at a belief $(s_A, b_E) \in S_{B}$ such that:
\begin{equation}\label{eq:update-lb-condition}
    \langle \alpha^{\star}, (s_A, b_E) \rangle = [TV_{\mathit{lb}}^{\Gamma}](s_A, b_E) \, .
\end{equation}
To that end, let $a$ be an action maximising the Bellman backup \eqref{eq:max-operator} at $(s_A, b_E)$, i.e., $a$ is a maximiser when computing $[TV_{\mathit{lb}}^{\Gamma}](s_A, b_E)$. If action $a$ is taken, then $\bar{S}_A = \{ s_A' \in S_A \mid P(s_A' \mid (s_A, b_E), a) > 0 \}$ are agent states that can be observed. If $s_A'$ is observed, then the backup value at belief $(s_A, b_E)$ from an $\alpha$-function $\alpha \in \Gamma$ equals $\int_{s_E \in S_E} \mathit{bval}((s_A, s_E), a, s_A', \alpha) b_E(s_E) \textup{d}s_E$, where for any $s_E \in S_E$:
\begin{equation*} %\label{eq:alpha-new-middle}
\mathit{bval}((s_A, s_E), a, s_A', \alpha) = \beta \delta_A(s_A, a)(\loc')  \mbox{$\int_{s_E' \in S_E^{s_{\scale{.75}{A}}'}}$} \delta_E(s_E, a)(s_E') \alpha(s_A', s_E')  \,.
\end{equation*}
For $s_A' \in \bar{S}_A$, let $\alpha^{s_A'} \in \Gamma$ be an $\alpha$-function maximising the backup value, i.e., $\alpha^{s_A'} \in \argmax_{\alpha \in \Gamma} \int_{s_E \in S_E} \mathit{bval}((s_A, s_E), a, s_A', \alpha) b_E(s_E) \textup{d}s_E$.

Using $a$, $\alpha^{s_A'}$ for $s_A' \in \bar{S}_A$ and the perception FCP $\Phi_P$, \algoref{alg:point-based-update-belief} computes a new $\alpha$-function $\alpha^{\star}$ at belief $(s_A, b_E)$. To guarantee \eqref{eq:update-lb-condition} and improve the efficiency, we only compute the backup values for regions $\phi \in \Phi_P$ over which $(s_A, b_E)$ has positive probabilities, i.e. $s^\phi_A=s_A$ (recall $s^\phi_A$ is the unique agent state appearing in $\phi$) and $\int_{(s_A, s_E) \in \phi} b_E(s_E) \textup{d} s_E > 0$
and assign the trivial lower bound $L$ otherwise. More precisely, for each such region $\phi$ and $(\hat{s}_A,\hat{s}_E) \in \phi$: 
\begin{equation}\label{eq:backup-value}
    \alpha^{\star}(\hat{s}_A,\hat{s}_E) =  R_{a}(\hat{s}_A,\hat{s}_E) + \mbox{$\sum\nolimits_{s_A' \in S_A}$} \mathit{bval}((\hat{s}_A,\hat{s}_E), a, s_A', \alpha^{s_A'})
\end{equation}
where if $s_A' \notin \bar{S}_A$, then $\alpha^{s_A'}$ can be any $\alpha$-function in $\Gamma$.
Computing the backup values \eqref{eq:backup-value} state by state is computationally intractable, as region $\phi$ contains an infinite number of states. 
However, the following lemma shows that $\alpha^{\star}$ is PWC, thus resulting in a tractable region-by-region backup. The lemma also shows that the lower bound function increases uniformly, is valid after each update, and performs no worse than the Bellman backup at the current belief.

\begin{algorithm}[tb]
\caption{Point-based $\mathit{Update}(s_A, b_E)$ of $(V_{\mathit{lb}}^{\Gamma}, V_{\mathit{ub}}^{\Upsilon})$}
\label{alg:point-based-update-belief}
\begin{algorithmic}[1] %[1] enables line numbers
\State $a \leftarrow$ the maximum action in computing $[TV_{\mathit{lb}}^{\Gamma}](s_A, b_E)$
\State $\bar{S}_A \leftarrow \{ s_A' \in S_A \mid P(s_A' \mid (s_A, b_E), a) > 0 \}$
\State $\alpha^{s_A'} \leftarrow \argmax_{\alpha \in \Gamma} \int_{s_E \in S_E} \mathit{bval}((s_A, s_E), a, s_A', \alpha) b_E(s_E) \textup{d}s_E $ for all $s_A' \in \bar{S}_A$
\For{$\phi \in \Phi_P$}
\If{$s^\phi_A = s_A$ and $\int_{(s_A, s_E) \in \phi} b_E(s_E) \textup{d} s_E > 0$}
\State Compute $\alpha^{\star}(\hat{s}_A,\hat{s}_E)$ by \eqref{eq:backup-value} for $(\hat{s}_A,\hat{s}_E) \in \phi$ \Comment{ISPP backup} 
\Else $\;\alpha^{\star}(\hat{s}_A,\hat{s}_E) \leftarrow L$ for $(\hat{s}_A,\hat{s}_E) \in \phi$
\EndIf
\EndFor
\State $\Gamma \leftarrow \Gamma \cup \{\alpha^{\star}\}$
\State $p^{\star} \leftarrow [TV_{\mathit{ub}}^{\Upsilon}](s_A, b_E)$ 
\State $\Upsilon \leftarrow \Upsilon \cup \{((s_A, b_E), p^{\star}) \}$
\end{algorithmic}
\end{algorithm}

\begin{algorithm}[tb]
\caption{Image-Split-Preimage-Product (ISPP) backup over a region}
\label{alg:ISPP-backup}
\textbf{Input}: region $\phi$, action $a$, PWC $\alpha^{s_A'}$ for all $s_A' \in S_A$
\begin{algorithmic}[1] %[1] enables line numbers
\State $\Loc' \leftarrow \{ \loc' \in \Loc \mid \delta_A(s^\phi_A, a)(\loc') > 0 \} $, $\Phi_{\textup{product}} \leftarrow \phi$
\For{$\loc' \in \Loc', i = 1, \dots, n$} 
\State $\phi_{E}' \leftarrow \{ \delta_E^i(s_E, a) \mid (s^\phi_A, s_E) \in \phi \}$ \Comment{Image}
\State $\Phi_{\textup{image}} \leftarrow \textup{divide } \phi_{E}' \textup{ into regions over } S \textup{ by } \obs_A(\loc', \cdot )$
\State $\Phi_{\textup{split}} \leftarrow \emptyset$ \Comment{Split}
\For{$\phi_{\textup{image}} \in \Phi_{\textup{image}}$}
\State $\Phi_{\alpha} \leftarrow \textup{a constant-FCP of } S \textup{ for the PWC function } \alpha^{s_A^{\phi_{\scale{.75}{\textup{image}}}}}$
\State $\Phi_{\textup{split}} \leftarrow \Phi_{\textup{split}} \cup  \{ \phi_{\textup{image}} \cap \phi' \mid \phi' \in \Phi_{\alpha} \}$
\EndFor
\State $\Phi_{\textup{pre}} \leftarrow \emptyset$  \Comment{Preimage}
\For{$\phi_{\textup{image}} \in \Phi_{\textup{split}}$}
%\State $\phi_{\textup{pre}} \leftarrow \{ (s^\phi_A, s_E) \in \phi \mid t_E^{i}(s_E, a) \in \phi_{\textup{image}} \}$
\State $\Phi_{\textup{pre}} \leftarrow \Phi_{\textup{pre}}  \cup \{ (s^\phi_A, s_E) \in \phi \mid \delta_E^{i}(s_E, a) \in \phi_{\textup{image}} \}$ 
\EndFor
\State $\Phi_{\textup{product}} \leftarrow \{ \phi_1 \cap \phi_2 \mid \phi_1 \in \Phi_{\textup{pre}} \wedge \phi_2 \in \Phi_{\textup{product}} \}$  \Comment{Product}
\EndFor
\State $\Phi_{\textup{product}} \leftarrow \{ \phi_1 \cap \phi_2 \mid \phi_1 \in \Phi_{\textup{product}} \wedge \phi_2 \in \Phi_{R}^{a} \}$
\For{$\phi_{\textup{product}} \in \Phi_{\textup{product}} $} \Comment{Value backup}
\State Take one state $(\hat{s}_A,\hat{s}_E) \in \phi_{\textup{product}}$
\State $\alpha^{\star}(\phi_{\textup{product}}) \leftarrow  R_{a}(\hat{s}_A,\hat{s}_E) + \sum\nolimits_{s_A' \in S_A} \mathit{bval}((\hat{s}_A,\hat{s}_E), a, s_A', \alpha^{s_A'})$
\EndFor
\State \textbf{return:} $(\Phi_{\textup{product}}, \alpha^{\star})$
\end{algorithmic}
\end{algorithm}

\begin{lema}[Lower bound]\label{lema:new-pwc-alpha}
At belief $(s_A, b_E) \in S_{B}$, the function $\alpha^{\star}$ generated by \algoref{alg:point-based-update-belief} is a PWC $\alpha$-function satisfying \eqref{eq:update-lb-condition}, $V_{\mathit{lb}}^{\Gamma} \leq V_{\mathit{lb}}^{\Gamma'} \leq V^{\star}$ and $V_{\mathit{lb}}^{\Gamma'}(s_A, b_E) \ge [TV_{\mathit{lb}}^{\Gamma}](s_A, b_E)$.
\end{lema}
Since $\alpha^{\star}$ is PWC, we next present a new backup for \eqref{eq:backup-value} through finite region-by-region backups.
Recall 
% from \aspref{asp:transitions} 
that $\delta_E$ can be represented as $\sum_{i = 1}^{n} \mu_i \delta_E^i$ {\revise and can preserve a decomposition of the state space into a finite set of regions}. \algoref{alg:ISPP-backup} presents an Image-Split-Preimage-Product (ISPP) backup method to compute \eqref{eq:backup-value} region by region. This method, inspired by \lemaref{lema:new-pwc-alpha}, is to divide a region $\phi$ into subregions, where for each subregion $\alpha^{\star}$ is constant, illustrated in \figref{fig:ispp}. Given any reachable local state $\loc'$ under $a$ and continuous transition function $\delta_E^i$, the image of $\phi$ under $a$ and $\delta_E^i$ to $\loc'$ is divided into \emph{image} regions $\Phi_{\textup{image}}$ such that the states in each region have a unique agent state. Each image region $\phi_{\textup{image}}$ is then split into subregions by a constant-FCP of the PWC function $\alpha^{s_A^{\phi_{\scale{.75}{\textup{image}}}}}$ by pairwise intersections, and thus $\Phi_{\textup{image}}$ is \emph{split}  into a set of refined image regions $\Phi_{\textup{split}}$. An FCP over $\phi$, denoted by $\Phi_{\textup{pre}}$, is constructed by computing the \emph{preimage} of each $\phi_{\textup{image}} \in \Phi_{\textup{split}}$ to $\phi$. Finally, the \emph{product} of these FCPs  $\Phi_{\textup{pre}}$ for all reachable local states and environment functions and {\revise $\Phi_R^{a}$ (the reward FCP under action~$a$)}, denoted $\Phi_{\textup{product}}$, is computed. The following lemma demonstrates that $\alpha^{\star}$ is constant in each region of $\Phi_{\textup{product}}$, and therefore \eqref{eq:backup-value} can be computed by finite backups. % through one state for each region.

\newcommand{\bfcpgrid}[1]{

    \filldraw[draw=black, thick, fill=cyan, opacity=0.4] (2.0,-2.0,#1) -- (2.0,2.0,#1) --  (-2.0,2.0,#1) -- (-2.0,-2.0,#1) -- cycle;

    \draw (1.0,-2.0,#1) -- (1.0,2.0,#1);
    %\draw (0.0,-2.0,#1) -- (0.0,2.0,#1);
    \draw (-1.0,-2.0,#1) -- (-1.0,2.0,#1);

    \draw (-2.0,-1.0,#1) -- (2.0,-1.0,#1);
    %\draw (-2.0,0.0,#1) -- (2.0,0.0,#1);
    \draw (-2.0,1.0,#1) -- (2.0,1.0,#1)
   
}

\begin{figure}[t]
%\vspace{-0.3cm}
\centering
\begin{tikzpicture}[scale=0.6,xyz frame={(0,0,-1)}{(0,1,0)}{(1,0,0)}]
%\hspace{-1.0cm}

%------------%

\bfcpgrid{0};

\node at(1.5,-3,0.25) {$\phi$};

\draw[-{stealth}] (0,0,1) -- (0,0,2);

%------------%

\bfcpgrid{3};

\filldraw[fill=green,opacity=0.6] (0.85,0.15,3) -- (0.85,0.65,3) --  (0.35,0.65,3) -- (0.35,0.15,3) -- cycle;

\filldraw[fill=green,opacity=0.6] (-0.85,-0.85,3) -- (-0.85,-0.35,3) --  (-0.35,-0.35,3) -- (-0.35,-0.85,3) -- cycle;

\draw[-] (0.85,0.15,3) -- (0.85,-0.85,6);
\draw[-] (0.85,0.65,3) -- (0.85,-0.35,6);
\draw[-] (0.35,0.65,3) -- (0.35,-0.35,6);
\draw[-] (0.35,0.15,3) -- (0.35,-0.85,6);

\draw[-] (-0.85,-0.85,3) -- (-0.85,-1.85,6);
\draw[-] (-0.85,-0.35,3) -- (-0.85,-1.35,6);
\draw[-] (-0.35,-0.35,3) -- (-0.35,-1.35,6);
\draw[-] (-0.35,-0.85,3) -- (-0.35,-1.85,6);

\node at(1.5,-3,3.25) {$\Phi_{\textup{image}}$};

%------------%

\bfcpgrid{6};

\filldraw[fill=orange,opacity=0.6] (0.85,-0.85,6) -- (0.85,-0.35,6) --  (0.35,-0.35,6) -- (0.35,-0.85,6) -- cycle;

\filldraw[fill=gray,opacity=0.6] (-0.85,-1.85,6) -- (-0.85,-1.35,6) --  (-0.35,-1.35,6) -- (-0.35,-1.85,6) -- cycle;

\draw[-{stealth}] (0.6,-0.6,6.25) -- (0.6,-0.6,8.75);
\draw[-{stealth}] (-0.6,-1.6,6.25) -- (-0.6,-1.6,8.75);

\node at(1.5,-3,6.25) {$\Phi_{\textup{split}}$};

%------------%

\bfcpgrid{9};

\filldraw[fill=red,opacity=0.6] (0.85,-0.85,9) -- (0.85,-0.35,9) --  (0.35,-0.35,9) -- (0.35,-0.85,9) -- cycle;

\draw (0.6,-0.85,9) -- (0.6,-0.35,9);
\draw (0.85,-0.6,9) -- (0.35,-0.6,9);

\filldraw[fill=blue,opacity=0.6] (-0.85,-1.85,9) -- (-0.85,-1.35,9) --  (-0.35,-1.35,9) -- (-0.35,-1.85,9) -- cycle;

\draw (-0.6,-1.35,9) -- (-0.6,-1.85,9);
\draw (-0.85,-1.6,9) -- (-0.35,-1.6,9);

\draw[dashed] (0.85,-0.85,9) -- (0.85,0.15,12);
\draw[dashed] (0.85,-0.35,9) -- (0.85,0.65,12);
\draw[dashed] (0.35,-0.35,9) -- (0.35,0.65,12);
\draw[dashed] (0.35,-0.85,9) -- (0.35,0.15,12);

\draw[dashed] (-0.85,-1.85,9) -- (-0.85,-0.85,12);
\draw[dashed] (-0.85,-1.35,9) -- (-0.85,-0.35,12);
\draw[dashed] (-0.35,-1.35,9) -- (-0.35,-0.35,12);
\draw[dashed] (-0.35,-1.85,9) -- (-0.35,-0.85,12);

\node at(1.5,-3,9.25) {$\Phi_{\textup{pre}}$};

%------------%

\bfcpgrid{12};

\filldraw[fill=yellow,opacity=0.6] (0.85,0.15,12) -- (0.85,0.65,12) --  (0.35,0.65,12) -- (0.35,0.15,12) -- cycle;

\draw (0.6,0.15,12) -- (0.6,0.65,12);
\draw (0.85,0.4,12) -- (0.35,0.4,12);

\filldraw[fill=magenta,opacity=0.6] (-0.85,-0.85,12) -- (-0.85,-0.35,12) --  (-0.35,-0.35,12) -- (-0.35,-0.85,12) -- cycle;

\draw (-0.6,-0.35,12) -- (-0.6,-0.85,12);
\draw (-0.85,-0.6,12) -- (-0.35,-0.6,12);

\draw[-{stealth}] (0.6,0.4,12.25) -- (0.6,0.4,14.75);
\draw[-{stealth}] (-0.6,-0.6,12.25) -- (-0.6,-0.6,14.75);

\node at(1.5,-3,12.5) {$\Phi_{\textup{product}}$};

%------------%

\bfcpgrid{15};

\filldraw[fill=purple,opacity=0.6] (0.85,0.15,15) -- (0.85,0.65,15) --  (0.35,0.65,15) -- (0.35,0.15,15) -- cycle;

\draw (0.6,0.15,15) -- (0.6,0.65,15);
\draw (0.85,0.4,15) -- (0.35,0.4,15);

\draw (0.35,0.15,15) -- (0.85,0.65,15);
\draw (0.35,0.65,15) -- (0.85,0.15,15);

\filldraw[fill=brown,opacity=0.6] (-0.85,-0.85,15) -- (-0.85,-0.35,15) --  (-0.35,-0.35,15) -- (-0.35,-0.85,15) -- cycle;

\draw (-0.6,-0.35,15) -- (-0.6,-0.85,15);
\draw (-0.85,-0.6,15) -- (-0.35,-0.6,15);

\draw (-0.85,-0.35,15) -- (-0.35,-0.85,15);
\draw (-0.85,-0.85,15) -- (-0.35,-0.35,15);

\end{tikzpicture}
\vspace*{-0.2cm}
\caption{Illustration of the steps taken by the ISPP algorithm.}\label{fig:ispp}
%\vspace*{-0.2cm}
\end{figure}
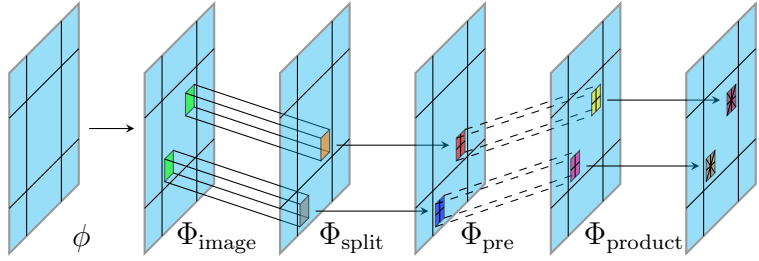

\begin{lema}[ISPP backup]
    The FCP $\Phi_{\textup{product}}$ returned by \algoref{alg:ISPP-backup} is a constant-FCP of $\phi$ for $\alpha^{\star}$ and the region-by-region backup for $\alpha^*$ satisfies~\eqref{eq:backup-value}.
\end{lema}
Computing $\mathit{bval}((\hat{s}_A,\hat{s}_E), a, s_A', \alpha^{s_A'})$ in the value backup requires $\alpha^{s_A'}(s_A', s_E')$. To obtain this value, we need to find the region in the constant-FCP for $\alpha^{s_A'}$ containing $(s_A', s_E')$.
To improve efficiency,
instead of searching,
we record the region connections during ISPP.

\begin{figure}[t]
    \centering
    \begin{tikzpicture}[scale=0.8, transform shape]
        \def\x{0}
        \draw[very thin] (\x+0,0) rectangle (\x+4,4);
        \node[teal, thick, regular polygon, regular polygon sides=5, draw, inner sep=0.25cm] (r) at (\x+1,0.5) {};
        \node[] at(\x+1,-0.35) {\large $\phi$};
        %%%
        \def\x{4.5}
        \draw[very thin] (\x+0,0) rectangle (\x+4,4);
        \draw[very thin] (\x+3,4) -- (\x+2,0);
        \draw[very thin] (\x+0,2) -- (\x+2.5,2);
        \node[orange, thick, densely dotted, regular polygon, regular polygon sides=5, draw, inner sep=0.25cm] (r11) at (\x+1,2) {};
        \node[purple, thick, densely dotted, regular polygon, regular polygon sides=5, draw, inner sep=0.25cm] (r21) at (\x+2.25,0.5) {};
        %\draw[orange, thin, densely dotted,-{Latex[length=1.5mm]}] (r) |- (r11.north west) node [pos=0.8, yshift=0.25cm]{\small $\delta_1$};
        %\draw[purple, thin, densely dotted,-{Latex[length=1.5mm]}] (r.east) -- (r21.west) node [pos=0.45, yshift=0.25cm]{\small $\delta_2$};
        \draw[orange, thin, densely dotted,-{Latex[length=1.5mm]}] (r) |- (r11.north west) node [pos=0.8, yshift=0.25cm]{\small $0.6$};
        \draw[purple, thin, densely dotted,-{Latex[length=1.5mm]}] (r.east) -- (r21.west) node [pos=0.43, yshift=0.25cm]{\small $0.4$};
        \node[] at(\x+0.25,3.75) {$10$};
        \node[] at(\x+0.25,0.25) {$20$};
        \node[] at(\x+3.8,3.75) {$5$};
        \node[] at(\x+2,-0.35) {\large $\Phi_{R}^a$};
        %%%
        \def\x{9.0}
        \node[orange, thick, densely dotted, regular polygon, regular polygon sides=5, draw, inner sep=0.25cm] (r12) at (\x+1,2) {};
        \draw[orange, thick, densely dotted] (r12.west) --(r12.east);
        \node[purple, thick, densely dotted, regular polygon, regular polygon sides=5, draw, inner sep=0.25cm] (r22) at (\x+2.25,0.5) {};
        \draw[purple, thick, densely dotted] (r22.north) -- (r22.corner 3);
        \draw[thin, densely dotted,-{Latex[length=1.5mm]}] (r11.north east) -- (r12.north west) node [pos=0.8, yshift=0.25cm]{}; %{\small $0.6$};
        \draw[thin, densely dotted,-{Latex[length=1.5mm]}] (r21.east) -- (r22.west) node [pos=0.475, yshift=0.25cm]{}; %{\small $0.4$};
        \node[draw, circle, fill=black, inner sep=0.5pt] () at (\x+1,2.15) (v1) {};
        \node[draw, circle, fill=black, inner sep=0.5pt] () at (\x+1,1.85) (v2) {};
        \node[draw, circle, fill=black, inner sep=0.5pt] () at (\x+2.0,0.6) (v3) {};
        \node[draw, circle, fill=black, inner sep=0.5pt] () at (\x+2.4,0.6) (v4) {};
        \draw[thin,] (v1.center) -- (\x+1,2.75) node [pos=1.25] {\small $6$};
        \draw[thin,] (v2.center) -- (\x+1,1.25) node [pos=1.25] {\small $12$};
        \draw[thin,] (v3.center) -- (\x+2.0,1.2) node [pos=1.25] {\small $8$};
        \draw[thin,] (v4.center) -- (\x+2.4,1.2) node [pos=1.25] {\small $2$};
        %%%
        \def\x{14.0}
        \node[teal, thick, regular polygon, regular polygon sides=5, draw, inner sep=0.25cm] (r1) at (\x+1,0.5) {};
        \draw[teal, thick, densely dotted] (r1.west) -- (r1.east);
        \draw[teal, thick, densely dotted] (r1.north) -- (r1.corner 3);
        \draw[very thin, densely dotted, -{Latex[length=1.5mm]}] (r12.north east) -| (r1.north);
        \draw[very thin, densely dotted, -{Latex[length=1.5mm]}] (r22.east) -- (r1.west);
        \node[draw, circle, fill=black, inner sep=0.5pt] () at (\x+0.75,0.65) (v13) {};
        \node[draw, circle, fill=black, inner sep=0.5pt] () at (\x+1.1,0.65) (v14) {};
        \node[draw, circle, fill=black, inner sep=0.5pt] () at (\x+0.75,0.4) (v23) {};
        \node[draw, circle, fill=black, inner sep=0.5pt] () at (\x+1.1,0.4) (v24) {};
        \draw[thin,] (v13.center) -- (\x+0.6,1.0) node [pos=1.6] {\small $14$};
        \draw[thin,] (v14.center) -- (\x+1.25,1.0) node [pos=1.6] {\small $8$};
        \draw[thin,] (v23.center) -- (\x+0.45,0.2) node [pos=1.6] {\small $20$};
        \draw[thin,] (v24.center) -- (\x+1.4,0.2) node [pos=1.6] {\small $14$};
        \node[] at(\x+1,-0.35) {\large $\phi'$};
    \end{tikzpicture}
    \vspace*{-0.2cm}
    \caption{\revise An example of region refinement during the ISSP algorithm (see Example~\ref{ex:refine}).}\label{region-refinement-fig}
\end{figure}

\begin{examp}\label{ex:refine}
{\revise In \figref{region-refinement-fig} we illustrate how a single region can be refined under an action $a$. The first step shows how performing action $a$ leads to two different regions being reached with probability $0.4$ and $0.6$, respectively. In the next step, 
these regions are then split based on $\Phi_R^{a}$, the reward FCP under action $a$, which has three regions, with reward values 5, 10 and 20 respectively, multiplying the probability and reward values for each region. In the last step, we take the intersections of the regions obtained for the different probabilistic transitions and a summation is performed for each subregion, yielding the average reward for these regions.}
\end{examp}

\startpara{Upper bound updates} For the upper bound $V_{\mathit{ub}}^{\Upsilon}$, 
working with the representation given in \eqref{eq:new-ub}, at a belief $(s_A, b_E) \in S_B$ in each iteration we add a new belief-value point $((s_A, b_E), p^{\star})$ to $\Upsilon$ 
such that $p^{\star} = [TV_{\mathit{ub}}^{\Upsilon}](s_A, b_E)$. The following lemma shows that $p^{\star} \ge V^{\star} (s_A, b_E)$ required by \eqref{eq:new-ub}, the upper bound function is decreasing uniformly, is valid after each update, and performs no worse than the Bellman backup at the current belief.

\begin{lema}[Upper bound]\label{lema:upper-bound-update}
Given belief $(s_A, b_E) \in S_{B}$, if $p^{\star} = [TV_{\mathit{ub}}^{\Upsilon}](s_A, b_E)$, then $p^\star$ is an upper bound of $V^{\star}$ at $(s_A, b_E)$, i.e., $p^{\star} \ge V^{\star} (s_A, b_E)$, and if $\Upsilon' = \Upsilon \cup \{ ((s_A, b_E), p^{\star}) \}$, then $V_{\mathit{ub}}^{\Upsilon} \ge V_{\mathit{ub}}^{\Upsilon'} \ge V^{\star}$ and $V_{\mathit{ub}}^{\Upsilon'}(s_A, b_E) \leq [TV_{\mathit{ub}}^{\Upsilon}](s_A, b_E)$.
\end{lema}

\begin{algorithm}[tb]
\caption{NS-HSVI for NS-POMDPs}
\label{alg:NS-HSVI}
\begin{algorithmic}[1] %[1] enables line numbers
\State Initialise $V_{\mathit{lb}}^{\Gamma}$  and $V_{\mathit{ub}}^{\Upsilon}$
\While{$V_{\mathit{ub}}^{\Upsilon}((s_{A}^{\mathit{init}}, b_{E}^{\mathit{init}}) - V_{\mathit{lb}}^{\Gamma}(s_{A}^{\mathit{init}}, b_{E}^{\mathit{init}}) > \varepsilon$} $\mathit{Explore}((s_{A}^{\mathit{init}}, b_{E}^{\mathit{init}}), \varepsilon, 0)$
\EndWhile
\Function{$\mathit{Explore}$}{$(s_A, b_E), \varepsilon, t$}
\If{$V_{\mathit{ub}}^{\Upsilon}(s_A, b_E) - V_{\mathit{lb}}^{\Gamma}(s_A, b_E) \leq \varepsilon \beta^{-t}$} \textbf{return}
\EndIf
\For{$a \in \Delta_A(s_A)$, $s_A' \in S_A$} 
\State $p^{a, s_A'} \leftarrow P(s_A' \mid (s_A,b_E), a)  V_{\mathit{ub}}^{\Upsilon}(s_A', b_E^{s_A, a, s_A'})$
\EndFor
\State $\hat{a} \leftarrow \argmax_{a \in \Delta_A(s_A)} \langle R_a, (s_A, b_E) \rangle + \beta \sum_{s_A' \in S_A} p^{a, s_A'}$
\State $\mathit{Update}(s_A, b_E)$
 
\State $\hat{s}_A \leftarrow \argmax\nolimits_{s_A' \in S_A}  \mathit{excess}_{t+1}(s_A', b_E^{s_A, \hat{a}, s_A'})$
 
\State $\mathit{Explore}((\hat{s}_A, b_E^{s_A, \hat{a}, \hat{s}_A}), \varepsilon, t+1)$
 
\State $\mathit{Update}(s_A, b_E)$
\EndFunction
\end{algorithmic}
\end{algorithm}

% \startpara{NS-HSVI algorithm} 
\subsection{NS-HSVI Algorithm}
\algoref{alg:NS-HSVI} presents the NS-HSVI algorithm for NS-POMDPs. Similarly to the heuristic search in HSVI \cite{TS-RS:04}, the algorithm (lines 5--7) selects an action $\hat{a}$ greedily according to the upper bound at belief $(s_A, b_E) \in S_B$, i.e., $\hat{a}$ is a maximiser when computing $[TV_{\mathit{ub}}^{\Upsilon}](s_A, b_E)$. Furthermore, given $\varepsilon > 0$, it selects an agent state $\hat{s}_A$ (observation) with the highest weighted excess approximation gap (line 9), denoted $\mathit{excess}_{t+1}(s_A', b_E^{s_A, \hat{a}, s_A'})$, which equals:
\[
P(s_A' \mid (s_A,b_E), \hat{a}) \big( V_{\mathit{ub}}^{\Upsilon}(s_A', b_E^{s_A, \hat{a}, s_A'}) -  V_{\mathit{lb}}^{\Gamma}(s_A', b_E^{s_A, \hat{a}, s_A'}) - \varepsilon \beta^{t+1} \big)
\]
where $t$ is the depth of $(s_A, b_E)$ from the initial belief $s_B^{\mathit{init}} = (s_A^{\mathit{init}}, b_E^{\mathit{init}}) \in S_B$.
NS-HSVI has the following convergence guarantees. 

\begin{thom}[NS-HSVI]\label{thom:NS-HSVI}
\algoref{alg:NS-HSVI} will terminate and upon termination:
\begin{enumerate} %[label=(\roman*)]
    
    \item\label{itm:HSVI-2a} $V_{\mathit{ub}}^{\Upsilon}(s_B^{\mathit{init}}) - V_{\mathit{lb}}^{\Gamma}(s_B^{\mathit{init}}) \leq \varepsilon$;
    \item\label{itm:HSVI-2b} $V_{\mathit{lb}}^{\Gamma}(s_B^{\mathit{init}}) \leq V^{\star}(s_B^{\mathit{init}}) \leq V_{\mathit{ub}}^{\Upsilon}(s_B^{\mathit{init}}) $;
    
    \item\label{itm:HSVI-3} $V^{\star}(s_B^{\mathit{init}}) - \mathbb{E}^{\hat{\sigma}}_{s_B^{\mathit{init}}}[Y] \leq \varepsilon$ where $\hat{\sigma}$ is the one-step lookahead strategy from $V_{\mathit{lb}}^{\Gamma}$.
\end{enumerate}
\end{thom}
\begin{proof}
Given belief $(s_A, b_E) \in S_B$, through Lemma~\ref{lema:new-pwc-alpha} after updating a lower bound $V_{\mathit{lb}}^{\Gamma}$ we have:
\begin{equation}\label{convergence1-eqn}
V_{\mathit{lb}}^{\Gamma} \leq V_{\mathit{lb}}^{\Gamma'} \leq V^{\star} \quad \mbox{and} \quad V_{\mathit{lb}}^{\Gamma'}(s_A, b_E) \ge [TV_{\mathit{lb}}^{\Gamma}](s_A, b_E)
\end{equation}
and through Lemma~\ref{lema:upper-bound-update} after updating an upper bound $V_{\mathit{ub}}^{\Upsilon}$, we have:
\begin{equation}\label{convergence2-eqn}
V_{\mathit{ub}}^{\Upsilon} \ge V_{\mathit{ub}}^{\Upsilon'} \ge V^{\star} \quad \mbox{and} \quad V_{\mathit{ub}}^{\Upsilon'}(s_A, b_E) \leq [TV_{\mathit{ub}}^{\Upsilon}](s_A, b_E) \, .
\end{equation}
Now, since $V_{\mathit{lb}}^{\Gamma}$ and $V_{\mathit{ub}}^{\Upsilon}$ are initially bounded and from Lemmas~\ref{lema:new-pwc-alpha} and \ref{lema:upper-bound-update} are uniformly improvable, $\delta$ has finite branching and $\beta \in (0, 1)$, using  \cite[Theorem 6.8]{TS-07} we have that \algoref{alg:NS-HSVI} terminates after finitely many steps. 

Next, combining  \eqnref{convergence1-eqn} and \eqnref{convergence2-eqn}, and using \cite[Section 6.5]{TS-07} both \ref{itm:HSVI-2a} and  \ref{itm:HSVI-2b} follow directly. Finally, concerning \ref{itm:HSVI-3}, by \eqref{eq:lb-alpha-v-1}, we have 
\begin{align}\label{eq:hsvi-alpha-tv}
    \langle \alpha^{\star}, (\hat{s}_A,\hat{s}_E) \rangle \leq [TV_{\mathit{lb}}^{\Gamma}](\hat{s}_A,\hat{s}_E)
\end{align}
for all $(\hat{s}_A,\hat{s}_E) \in S_B$. If $V_{\mathit{lb}}^{\Gamma} \leq [TV_{\mathit{lb}}^{\Gamma}]$, we have $V_{\mathit{lb}}^{\Gamma'} \leq [TV_{\mathit{lb}}^{\Gamma}]$ using \eqref{eq:hsvi-alpha-tv}. Then, since \algoref{alg:NS-HSVI} 
terminates, according to \cite[Theorem 3.18]{TS-07}:
\begin{align*}
   V^{\star}(s_{A}^{\mathit{init}}, b_{E}^{\mathit{init}}) - \mathbb{E}^{\hat{\sigma}}_{(s_{A}^{\mathit{init}}, b_{E}^{\mathit{init}})}[Y] 
   \leq V_{\mathit{ub}}^{\Upsilon}(s_{A}^{\mathit{init}}, b_{E}^{\mathit{init}}) - V_{\mathit{lb}}^{\Gamma}(s_{A}^{\mathit{init}}, b_{E}^{\mathit{init}}) \leq \varepsilon
\end{align*}
which completes the proof. 
\end{proof}

\startparazero{Pruning} We apply the following pruning to speed up \algoref{alg:NS-HSVI}. First, a new $\alpha$-function $\alpha^{\star}$ is added to $\Gamma$ at belief $(s_A, b_E)$ in each update only if  $\alpha^{\star}$ strictly improves the value at $(s_A, b_E)$, i.e., $\langle \alpha^{\star},(s_A, b_E) \rangle > V_{\mathit{lb}}^{\Gamma}(s_A, b_E)$. This leads to fewer $\alpha$-functions in $\Gamma$ without changing convergence, and thus faster lower bound computation. Second, for the heuristic search, since the action $\hat{a}$ (line 6) maximising the upper bound backup may not be unique and different $\hat{a}$ could result in different maximum gaps (line 8), we keep all maximisers and select the pair $(\hat{a}, \hat{s}_A)$ with the largest gap. We find this new excess heuristic to be empirically superior, as it tends to reduce the uncertainty the most.

\startpara{Convergence} Each belief update of \algoref{alg:NS-HSVI} is focused on a single belief, and therefore the number of iterations can be higher than value iteration;
on the other hand, iterations are cheaper to perform. In the finite-state case,
an upper bound on the number of HSVI  iterations required can be calculated \cite[Theorem 6.8]{TS-07}. However, such analysis would be difficult in our setting, as the number of points to update depends on the initial beliefs and which beliefs are updated at each iteration,
% the dynamics of the system, 
and varies as the algorithm progresses.

\subsection{Two Belief Representations}
An implementation of the NS-HSVI algorithm crucially depends on the representations of beliefs, 
as a closed form is needed when computing belief $\smash{b_E^{s_A, a, s_A'}}$, expected values $\langle \alpha, (s_A, b_E) \rangle$ and $\langle R_a, (s_A, b_E) \rangle$, probability $P(s_A' \mid (s_A,b_E), a)$ and upper bound $V_{\mathit{ub}}^{\Upsilon}(s_A,b_E)$. 
We first consider the popular particle-based belief representation and then propose a region-based belief representation to overcome the problem of requiring many particles to converge in the particle-based representation~\cite{DC-AD:02}.  

\startpara{Particle-based beliefs} 
Particle-based representations have been widely used in applications from computer vision \cite{AD-NDF-NJG:01}, robotics \cite{ML-DH-PJ:22,JMP-NV-MTS-PP:06} to machine learning \cite{XM-PK-DH-WL:20}. They can approximate arbitrary beliefs (given sufficient particles), handle nonlinear and non-Gaussian systems, and allow efficient computations.

\begin{defi}[Particle-based belief]\label{defi:particle-based-belief}
A belief $(s_A, b_E) \in S_B$ is represented by a weighted particle set $\{ (s_E^i, w_i) \}_{i=1}^{n_b}$ with normalised weights if
\[
b_E(s_E) = \mbox{$\sum\nolimits_{i = 1}^{n_b}$} w_i D(s_E - s_E^i)
\]
where $w_i >0$, $s_E^i \in S_E$ for all $1 \leq i \leq n_b$ and $D(s_E - s_E^i)$ is a Dirac delta function centred at $0$. Let $B(s_E)$ be a small neighbourhood of $s_E$, and $P(s_E;b_E) = \int_{s_E' \in B(s_E)} b_E(s_E') \textup{d} s_E' $ be the probability of particle $s_E$ under $b_E$.
% We denote by $S_{\mathit{PB}}$ the set of particle-based beliefs.
\end{defi}
Given an initial particle-based belief $(s_{A}^\mathit{init}, b_{E}^\mathit{init})$, the number of states reachable in any finite number of steps is finite, and therefore standard methods for finite-state POMDPs can be used to solve the resulting finite-state game tree, similarly to \cite{YSD+22} under fully-observable strategies.  However, the size of the game tree can increase exponentially as the number of steps increases, particularly given that the reachable states are likely to be distinct due to the continuous-state space.

To implement NS-HSVI given in Algorithm~\ref{alg:NS-HSVI} using particle-based beliefs, we must demonstrate that $V_{\mathit{lb}}^{\Gamma}$ and $V_{\mathit{ub}}^{\Upsilon}$ are eligible representations~\cite{JMP-NV-MTS-PP:06} for particle-based beliefs, that is, there are closed forms for the quantities of interest. For a particle-based belief $(s_A, b_E)$ with weighted particle set $\{ (s_E^i, w_i) \}_{i=1}^{n_b}$, it follows from \eqref{eq:belief-update} that for belief $b_E^{s_A, a, s_A'}$ we have, for any $s_E' \in S_E$, $b_E^{s_A, a, s_A'}(s_E')$ equals:
\begin{equation}\label{eq:pb-belief-update}
\frac{\sum\nolimits_{i = 1}^{n_b}   w_i \delta_E(s_E^i, a)(s_E')}{\sum\nolimits_{i = 1}^{n_b} w_i \sum_{s_{\scale{.75}{E''}} \in S_{\scale{.75}{E}}^{\scale{.75}{{s_A'}}}} \delta_E(s_E^i, a)(s_E'')} \; \mbox{if $s_E' \in S_E^{s_A'}$ and equals 0 otherwise.} 
\end{equation}
Similarly, we can compute $\langle \alpha, (s_A, b_E) \rangle$, $\langle R_a, (s_A, b_E) \rangle$ and $P(s_A' \mid (s_A,b_E), a)$ as simple summations.
\iffalse
Considering any $\alpha \in \Gamma$, we have that:
\begin{eqnarray*}
\langle \alpha, (s_A, b_E) \rangle &=& \mbox{$\sum_{s_E \in \supp(b_E)}$} \alpha(s_A, s_E) b_E(s_E) \\
\langle R_a, (s_A, b_E) \rangle &=& \mbox{$\sum_{s_E \in \supp(b_E)}$} R_a(s_A, s_E) b_E(s_E) \, .
\end{eqnarray*}
Furthermore, using \eqref{eq:probability-obs-agent-state}, probability $P(s_A' \mid (s_A,b_E), a)$ is:
% of observing $s_A'$ at $(s_A, b_E)$ under action $a$ 
\[
\delta_A(s_A, a)(\loc') \left( \mbox{$\sum\limits_{s_E \in \supp(b_E)}$} b_E(s_E) \mbox{$\!\!\!\sum\limits_{s_E' \in S_E^{s_A'}}$} \delta_E(s_E, a)(s_E') \right)
\] 
where $s_A' = (\loc', \per')$.
\fi
It remains to compute $V_{\mathit{ub}}^{\Upsilon}$ in \eqref{eq:new-ub}, which we achieve by designing a function $K_{\mathit{ub}}$ that measures belief differences that satisfy \eqref{eq:K-UB-condition}. However, \eqref{eq:K-UB-condition} is hard to check  as, for beliefs $b_E$ and $b_E'$, calculating $K(b_E,b_E')$ involves the integral over the region $S_E^{s_{\scale{.75}{A}}}$.
For particle-based beliefs, we propose the function $K_{\mathit{ub}}$ where:
\begin{equation}\label{eq:pb-k-ub}
    K_{\mathit{ub}}(b_E, b_E') = (U - L)  n_b \max\nolimits_{s_E \in S_E \wedge b_E(s_E) > 0} |P(s_E; b_E) - P(s_E; b_E')|
\end{equation}
where $n_b$ is the number of particles in $b_E$, which is shown to satisfy \eqref{eq:K-UB-condition} and, given $\Upsilon = \{ ((s_A^i, b_E^i), y_i) \mid i \in I\}$, the upper bound can be computed by solving a linear program (LP) as demonstrated by the following lemma.

\begin{lema}[LP for upper bound]\label{lema:pb-upper-bound}
The function $K_{\mathit{ub}}$ from \eqref{eq:pb-k-ub} satisfies \eqref{eq:K-UB-condition}, and for particle-based belief $(s_A, b_E)$ represented by $\{ (s_E^i, w_i) \}_{i=1}^{n_b}$, we have that $V_{\mathit{ub}}^{\Upsilon}(s_A,b_E)$ is the optimal value of the LP: 
\[    \begin{array}{rl}
    \textup{minimize:} \; & \sum_{k \in I_{s_{\scale{.75}{A}}}} \lambda_k y_k  + (U - L) n_b c \\
  \textup{subject to:} \; & c \ge | w_i - \sum_{k \in I_{s_{\scale{.75}{A}}}} \lambda_k P(s_E^i; b_E^k) |  \textup{ for } 1 \leq i \leq n_b \\
  & \lambda_k \ge 0 \textup{ for } k \in I_{s_{\scale{.75}{A}}} \; \mbox{and} \; \sum_{k \in I_{s_{\scale{.75}{A}}}} \lambda_k = 1\, .
\end{array}
\]
\end{lema}
Since all quantities of interest in \algoref{alg:NS-HSVI} are computed exactly, the convergence guarantee in \thomref{thom:NS-HSVI} holds for any initial particle-based belief.

\startpara{Region-based beliefs} Particle filter approaches \cite{AD-NDF-NJG:01} are required to approximate the updated belief of particle-based representations if the current belief has zero weight at the true state due to partial observations and random perturbations. However, for NS-POMDPs the usual sampling importance re-sampling (SIR) approach \cite{AD-SG-CA:00} requires many particles, which can be computationally expensive. 
Therefore, we propose a new belief representation using {\em regions} of the continuous state space and show that it performs well empirically in handling the uncertainties.  For any connected subset (region) $\phi_E \subseteq S_E$, let $\textup{vol}(\phi_E) = \int_{s_E \in \phi_E} \textup{d} s_E$. 

\begin{defi}[Region-based belief]\label{defi:region-based-belief} A belief $(s_A, b_E) \in S_B$ is represented by a weighted region set $\{ (\phi_E^i, w_i) \}_{i = 1}^{n_b}$ if
$
b_E(s_E) = \mbox{$\sum\nolimits_{i = 1}^{n_b}$} \chi_{\phi_E^i}(s_E) w_i
$, 
where  $w_i >0$, $\phi_E^i$ is a region of $S_E^{s_A}$ and $\chi_{\phi_E^i} : S_E \to \mathbb{R} $ is such that $\chi_{\phi_E^i}(s_E) = 1$ if $s_E \in \phi_E^i$ and $0$ otherwise for $1 \leq i \leq n_b$, and $\sum_{i=1}^{n_b} w_i \textup{vol}(\phi_E^i) = 1$.
\end{defi}
In the case of region-based beliefs, finite-state POMDPs are not applicable even when approximating by finding all reachable states up to some finite depth, as from an initial (region-based) belief this would yield infinitely many reachable states.
Region-based beliefs assume a uniform distribution over each region and allow the regions to overlap. Ensuring that belief updates of region-based beliefs result in region-based beliefs is difficult~\cite{AP-SUP:02}, as even simple transitions of variables with simple distributions can lead to complex distributions. \aspref{asp:transitions-rewards} only ensures a finite partitioning of the state space for the transitions, but not that the updated belief places a uniform distribution over each region. We now provide conditions on the deterministic continuous components $\delta_E^i$, 
% see \aspref{asp:transitions}, 
of the environment transition function, under which region-based beliefs are closed.

\begin{lema}[Region-based belief closure]\label{lema:uniform-closure}
    If $\delta_E^i(\cdot, a) : S_E \rightarrow \delta_E^i(S_E,a)$ is piecewise differentiable and invertible and the Jacobian determinant of the inverse function %, i.e., for any $s_E' \in \delta_E^i(S_E,a)$: $\textup{Jac}(s_E') = \textup{det} ( \textup{d} \delta_E^{i,-1}(s_E', a) / \textup{d} s_E' )$
    is PWC for any $a \in \Act$ and $1 \leq i \leq n$, then region-based beliefs are closed under belief updates.
\end{lema}
We next present an implementation of NS-HSVI using region-based beliefs for environment transition functions satisfying \lemaref{lema:uniform-closure}. 
% we demonstrate that the belief update, expected values and probability are easy to compute. 
For a region-based belief $(s_A, b_E)$, 
%represented by $\{ (\phi_E^i, w_i) \}_{i = 1}^{n_b}$, 
\algoref{alg:rb-belief-update} computes the belief update as the image of each region,
%$\phi_E^i$, 
dividing the images by perception functions into regions of $S_E$, updating weights and selecting the regions with desired observations. The region-based belief update and expected values are summarised in \lemaref{lema:region-belief-update}.

\begin{lema}[Region-based belief update]\label{lema:region-belief-update}
For region-based belief $(s_A, b_E)$ represented by $\{ (\phi_E^i, w_i) \}_{i = 1}^{n_b}$, action $a$ and observation $s_A'$, we have that $(s_A', b_E')$ returned by \algoref{alg:rb-belief-update} is region-based and $b_E' = b_E^{s_A, a, s_A'}$. Furthermore, if $h : S \to \mathbb{R}$ is PWC and $\Phi_E$ is a constant-FCP of $S_E$ for $h$ at $s_A$, then $\langle h, (s_A, b_E) \rangle = \sum_{i = 1}^{n_b} \sum_{\phi_E \in \Phi_E} h(s_A, s_E) w_i \textup{vol}(\phi_E^i \cap \phi_E)$ where $s_E \in \phi_E$.
\end{lema}
\begin{algorithm}[tb]
\caption{Region-based belief update}
\textbf{Input}: $(s_A, b_E)$ represented by $\{ (\phi_E^i, w_i) \}_{i = 1}^{n_b}$,  action $a$, observation $s_A' = (\loc', \per')$
\label{alg:rb-belief-update}
\begin{algorithmic}[1] %[1] enables line numbers
\If{$\delta_A(s_A, a)(\loc') > 0$}
\State $\mathcal{P} \leftarrow \emptyset$
\For{$i=1, \dots, n_b$, $j = 1, \dots, n$}
\State $\phi_E' \leftarrow \{ \delta_E^j(s_E, a) \mid s_E \in \phi_E^i \}$ \Comment{Image}
\State $\Phi_{\textup{image}} \leftarrow$ divide $\phi_E'$ into regions over $S_E$ by $\obs_A(\loc', \cdot)$ 
\State $w' \leftarrow (\textup{vol}(\phi_E^i) / \textup{vol}(\phi_E') ) w_i \mu_j$  \Comment{Weight update}
\State $\mathcal{P} \leftarrow \mathcal{P} \cup \{(\phi_E, w') \mid \phi_E \in \Phi_{\textup{image}} \wedge \phi_E \subseteq S_E^{s_A'} \}$
\EndFor
\State Normalise the weights in $\mathcal{P}$
\State $b_E'(s_E') \leftarrow \sum_{(\phi_E, w) \in \mathcal{P}} \chi_{\phi_E}(s_E') w$ for all $s_E'$
\Else $\; b_E'(s_E') \leftarrow 0$ for all $s_E'$
\EndIf
\State \textbf{return:} $(s_A', b_E')$
\end{algorithmic}
\end{algorithm}
For the upper bound $V_{\mathit{ub}}^{\Upsilon}$, the function $K_{\mathit{ub}}$ has to compare beliefs over regions. 
% The dependency between $S_E^{+}$ and $\lambda_i$ further complicates the computation. 
We let $K_{\mathit{ub}} = K$, and thus \eqref{eq:K-UB-condition} holds. Instead of a computationally expensive exact bound, which involves a large number of region intersections, \algoref{alg:rb-upper-bound} is approximate, based on maximum densities, and involves solving an LP.

\begin{algorithm}[tb]
\caption{Approximate region-based upper bound via maximum density}
\textbf{Input}: $(s_A, b_E)$ represented by $\{ (\phi_E^i, w_i) \}_{i = 1}^{n_b}$, $\Upsilon = \{ ((s_A^k, b_E^k), y_k) \mid k \in I\}$
\label{alg:rb-upper-bound}
\begin{algorithmic}[1] %[1] enables line numbers
\State $I_b \leftarrow \argmax_{I_b \subseteq \{1, \dots, n_b \}} \sum_{i \in I_b} w_i $ subject to: $\cap_{i \in I_b} \phi_E^i \neq \emptyset$ 
\State $\phi_E^{\textup{max}} \leftarrow \cap_{i \in I_b} \phi_E^i $ \Comment{Maximum density}
% \State $I_{s_{\scale{.75}{A}}} \leftarrow \{ k \in I \mid s_A^k = s_A \}$
\State $ p = \textup{minimize}  \sum_{k \in I_{s_{\scale{.75}{A}}}} \lambda_k y_k  + (U - L) c$ 
\State $\;\;\textup{subject to:} \; c \ge 1 - \sum_{k \in I_{s_{\scale{.75}{A}}}}\sum_{j = 1}^{n_{b}^k} \lambda_k  w_{kj} \textup{vol}(\phi_E^{kj} \cap \phi_E^{\textup{max}})$, 

$\;\;\;\;\;\;\;\;\;\;\;\;\lambda_k \ge 0, \; \sum_{k \in I_{s_{\scale{.75}{A}}}} \lambda_k = 1$
\State \textbf{return:} $(\phi_E^{\textup{max}},p)$
\end{algorithmic}
\end{algorithm}

\begin{lema}[Region-based upper bound]\label{rb-upper-bound}
For region-based belief $(s_A, b_E)$ represented by $\{ (\phi_E^i, w_i) \}_{i = 1}^{n_b}$ and $\Upsilon = \{ ((s_A^k, b_E^k), y_k) \mid k \in I\}$, if  $K_{\mathit{ub}} = K$, $(\phi_E^{\textup{max}},p)$ is returned by \algoref{alg:rb-upper-bound}, $b_E' = \sum_{k \in I_{s_{\scale{.75}{A}}}}  \lambda_k^{\star}  b_E^k$ and $b_E(s_E) > b_E'(s_E)$ for all $s_E \in \phi_E^{\textup{max}}$
% $\phi_E^{\textup{max}} \subseteq S_E^{b_E>b_E'}$ 
where $\lambda_k^{\star}$ is a solution to the LP of \algoref{alg:rb-upper-bound}, then $p$ is an upper bound of $V_{\mathit{ub}}^{\Upsilon}$ at $(s_A, b_E)$. Furthermore, if $n_b = 1$, then  $p =  V_{\mathit{ub}}^{\Upsilon}(s_A,b_E)$.
\end{lema}

\section{Implementation and Experimental Evaluation}\label{sec:case-study}
In this section, we present a prototype implementation and experimental evaluation of our NS-HSVI algorithm for solution and optimal strategy synthesis on NS-POMDPs. We first summarise the details of the experimental setup, then discuss the results of two case studies, and conclude the section by discussing  performance comparison.

\subsection{Implementation Overview}

We have developed a prototype Python implementation using the Parma Polyhedra Library \cite{BHZ08} to build and operate over perception FCP representations of preimages of NNs, $\alpha$-functions and reward structures. 
We recall that both $\alpha$-functions and reward functions are piecewise constant over the continuous environment. They can thus be represented by subdividing the entire environment into \emph{regions}, namely polyhedra over
the continuous variables to which we associate a value. 
We remark that, since our method crucially depends on the states in a given region, and those in the subregions arising from subsequent refinements, being mapped to the same percept, arbitrary discretisation is not applicable. We use $h$-representations, which describe polyhedra through linear constraints for intersecting finite half-spaces. Upper bound computation is performed by solving LPs with Gurobi~\cite{gurobi}.
To sample points with polyhedra,
we use the SMT solver Z3~\cite{Z3}.

% For both the car parking example and the VCAS, 
We use the method of  \cite{KM-FF:20} to compute the (exact) preimage of piecewise linear NNs, which iterates backwards through the layers. This method is %naturally 
only applicable when the NN has piecewise linear decision boundaries, for which the basic building blocks are polytopes. This includes NNs with ReLU or linear layers, but can also be applied to approximations of NNs obtained via, for example, linear relaxation.  With this preimage, we then construct a polyhedral representation of the environment space corresponding to the perception FCP. Regarding boundary points, we order regions and then assign boundary points to the region with the highest order, resolving ties via a measurable rule. 

\subsection{Car Parking Case Study}

\begin{figure}[t]
\vspace*{-0.4cm}
\centering
\begin{subfigure}{0.25\textwidth}
\raisebox{0.3\height}{\scalebox{.6}{\begin{tikzpicture}[scale = 0.7]
	bias = 4.5;
%  	\draw[white, thin, fill = green!30, opacity=0.8] (2, 3) rectangle (3, 4); 
    \filldraw[color=green!40, fill=green!30, very thick](2.5,3.5) circle (0.14);
     \draw[white, thin, fill = black!30, opacity=0.8] (1, 1) rectangle (2, 2);
	\draw[black, very thick] (0,0) rectangle (4,4);
    %\draw[black, thick] (2, 4) -- (2, 3);
	%\draw[black, thick] (2, 3) -- (3, 3);
	%\draw[black, thick] (3, 3) -- (3, 4);
    \draw[black, thick] (1, 2) -- (1, 1);
    \draw[black, thick] (1, 1) -- (2, 1);
    \draw[black, thick] (2, 1) -- (2, 2);
    \draw[black, thick] (1, 1) -- (2, 2);
    \draw[black, thick] (1, 2) -- (2, 1);
    \draw[black, thick] (1, 2) -- (2, 2);
	\node at (3.95, -0.3) {$4$};
	\node at (-0.2, -0.3) {$0$};
	\node at (-0.25, 3.95) {$4$};
 	\node at (-0.25, 3) {$3$};
	\node at (-0.25, 2) {$2$};
	\node at (-0.25, 1) {$1$};
	\node at (3, -0.3) {$3$};
	\node at (2, -0.3) {$2$};
	\node at (1, -0.3) {$1$};
	\node [sedan top,body color=red!30,window color=black!80,minimum width=0.9cm, rotate = 90] at (0.5,1) {};
\end{tikzpicture}}}
\end{subfigure}
\hfil
\begin{subfigure}{0.3\textwidth}
\scalebox{.6}{\begin{tikzpicture}[scale = 0.6]
  	\draw[white, thin, fill = green!30, opacity=0.8] (6, 7) rectangle (8, 8); 
    % obstacles
    \draw[white, thin, fill = black!30, opacity=0.8] (4, 0) rectangle (5, 1);
    \draw[white, thin, fill = black!30, opacity=0.8] (7, 2) rectangle (8, 3);
    \draw[white, thin, fill = black!30, opacity=0.8] (4, 4) rectangle (5, 5);
    \draw[white, thin, fill = black!30, opacity=0.8] (2, 7) rectangle (4, 8);
     \draw[black, very thick] (0,0) rectangle (8,8);
    \draw[black, thick] (6, 8) -- (6, 7);
	\draw[black, thick] (6, 7) -- (8, 7);
    \draw[black, thick] (4, 0) -- (4, 1);
    \draw[black, thick] (4, 1) -- (5, 1);
    \draw[black, thick] (5, 1) -- (5, 0);
    \draw[black, thick] (4, 0) -- (5, 1);
    \draw[black, thick] (4, 1) -- (5, 0);
    \draw[black, thick] (7, 3) -- (7, 2);
    \draw[black, thick] (7, 2) -- (8, 2);
    \draw[black, thick] (8, 3) -- (7, 3);
    \draw[black, thick] (7, 3) -- (8, 2);
    \draw[black, thick] (7, 2) -- (8, 3);
    \draw[black, thick] (4, 5) -- (4, 4);
    \draw[black, thick] (4, 4) -- (5, 4);
    \draw[black, thick] (5, 4) -- (5, 5);
    \draw[black, thick] (5, 5) -- (4, 5);
    \draw[black, thick] (4, 5) -- (5, 4);
    \draw[black, thick] (4, 4) -- (5, 5);
    \draw[black, thick] (2, 8) -- (2, 7);
    \draw[black, thick] (2, 7) -- (4, 7);
    \draw[black, thick] (4, 7) -- (4, 8);
    \draw[black, thick] (3, 8) -- (3, 7);
    \draw[black, thick] (2, 8) -- (3, 7);
    \draw[black, thick] (2, 7) -- (3, 8);
    \draw[black, thick] (3, 8) -- (4, 7);
    \draw[black, thick] (4, 8) -- (3, 7);
	\node at (-0.2, -0.3) {$0$};
    \node at (-0.25, 7.95) {$8$};
 	\node at (-0.25, 7) {$7$};
	\node at (-0.25, 6) {$6$};
	\node at (-0.25, 5) {$5$};
	\node at (-0.25, 4) {$4$};
 	\node at (-0.25, 3) {$3$};
	\node at (-0.25, 2) {$2$};
	\node at (-0.25, 1) {$1$};
  	\node at (7.95, -0.3) {$8$};
	\node at (7, -0.3) {$7$};
	\node at (6, -0.3) {$6$};
	\node at (5, -0.3) {$5$};
 	\node at (4, -0.3) {$4$};
	\node at (3, -0.3) {$3$};
	\node at (2, -0.3) {$2$};
	\node at (1, -0.3) {$1$};
	\node [sedan top,body color=red!30,window color=black!80,minimum width=0.9cm, rotate = 90] at (1.5,1) {};
\end{tikzpicture}}
\end{subfigure}
\caption{Car parking with obstacles.}
\label{fig:car_parking_obstacle}
\end{figure}
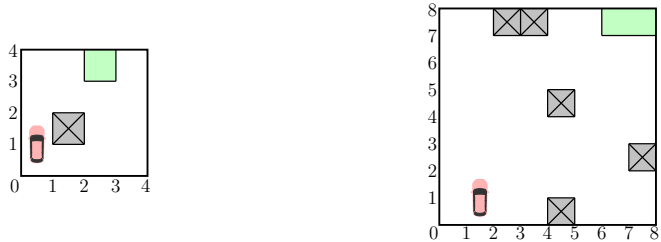

\iffalse
\begin{figure}[t]
\centering
\hspace{-1cm}
\begin{subfigure}{0.45\textwidth}
    \centering
    \includegraphics[height=0.68\textwidth]{figures/parking/4x4_3_initial_obstacle_-1000.png}
    %\subcaption{Collision reward = \num{-1e3}}
\end{subfigure}
\hspace{-1.0cm}
\raisebox{-0.065\height}{
\begin{subfigure}{0.1\textwidth}
    \centering
    \includegraphics[height=3.4\textwidth]{figures/parking/colormap-1000.png}
\end{subfigure}
}
\hfil
\raisebox{-0.1\height}{
\begin{subfigure}{0.45\textwidth}
    \input{figures/parking/4x4_3_initial_obstacle_-1000}
\end{subfigure}
}

\begin{subfigure}{0.45\textwidth}
\centering
\hspace{-1.85cm}
    \includegraphics[height=0.68\textwidth]{figures/parking/4x4_3_initial_obstacle_-5000.png}
    %\subcaption{Collision reward = \num{-5e3}}
\end{subfigure}
\hspace{-1.85cm}
\raisebox{-0.065\height}{
\begin{subfigure}{0.1\textwidth}
    \centering
    \includegraphics[height=3.4\textwidth]{figures/parking/colormap-5000.png}
\end{subfigure}
}
\hfil
\raisebox{-0.175\height}{
\begin{subfigure}{0.45\textwidth}
    \input{figures/parking/4x4_3_initial_obstacle_-5000}
\end{subfigure}
}

\vspace{-0.5cm}
\caption{Paths and values for car parking (obstacle indicated with black border, $\beta = 0.8$, collision rewards equal to $-1000$ (top) and $-5000$ (bottom)).}
\label{fig:car_parking_safe_unsafe}
\end{figure}
\fi

\begin{figure}
    \centering
    \begin{tabular}{lcr}
    \includegraphics[height=0.3\textwidth]{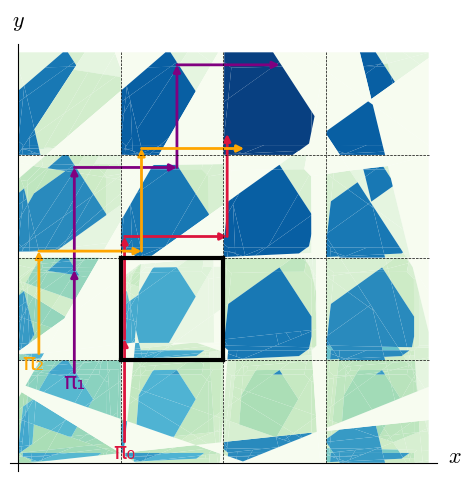}
    & 
    \raisebox{-0.065\height}{
    \includegraphics[height=0.32\textwidth]{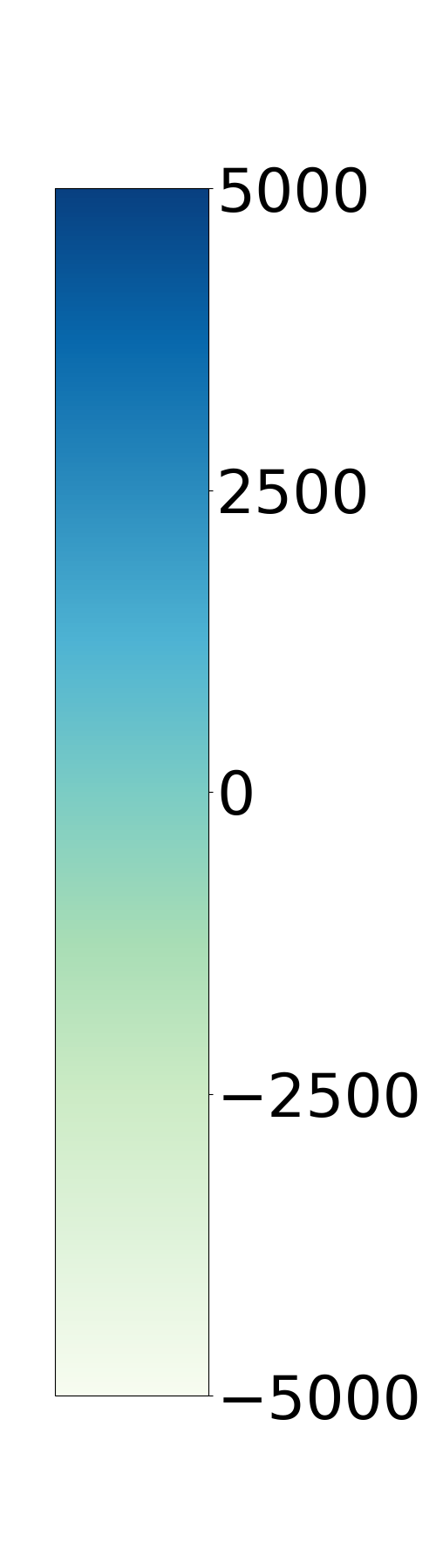}
    }
    & 
    \raisebox{-0.1\height}{
    %\begin{figure}%[h!]
%\centering
\scriptsize{
%\hspace{-0.0cm}
\begin{tikzpicture}
\begin{axis}[
    %no markers,
    title style={yshift=-2ex},
    title={\emph{value for the initial belief}},
    ylabel={},
    y label style={at={(axis description                cs:0.075,.5)},anchor=south,
    /pgf/number format/fixed,
    /pgf/number format/precision=4
    },
    xlabel={$k$},
    xmin=0, xmax=20,
    xtick={0,5,10,15,20},
    ymin=-5000, ymax=5000,
    ytick={-5000, -4000, -3000, -2000, -1000, 0,
            1000, 2000, 3000, 4000, 5000},
    ymajorgrids=true,
    grid style=dashed,
    grid=both,
    height=4.5cm,
    width=0.45\textwidth,
    legend entries={
                {\emph{lower bound}},
                {\emph{upper bound}}
                },
    legend style={at={(0.95,0.05)},
                anchor=south east, 
                nodes={scale=0.9, transform shape}}           
]
\addlegendimage{mark=square*,cyan,mark size=1.5pt}
\addlegendimage{mark=*,teal,mark size=1.5pt}
]

\addplot[mark=square*,cyan,opacity=1.0,mark size=1.0pt] table [x=k, y=lb, col sep=comma]{figures/parking/4x4_3_initial_obstacle_-1000.csv};
\addplot[mark=*,teal,opacity=1.0,mark size=1.0pt] table [x=k, y=ub, col sep=comma]{figures/parking/4x4_3_initial_obstacle_-1000.csv};

\end{axis}
\end{tikzpicture}
}
%\caption{}
%\label{}
%\end{figure}
    }
    \\
    \includegraphics[height=0.3\textwidth]{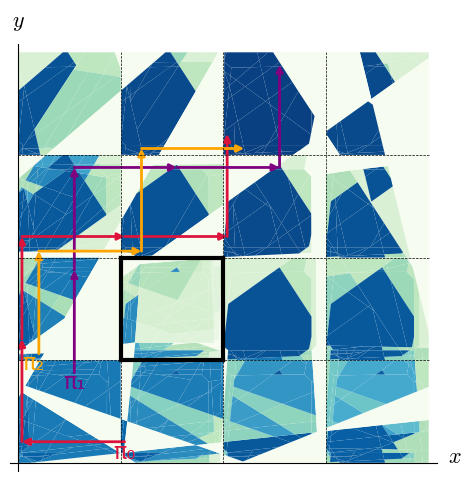}
    & 
    \raisebox{-0.065\height}{
    \includegraphics[height=0.32\textwidth]{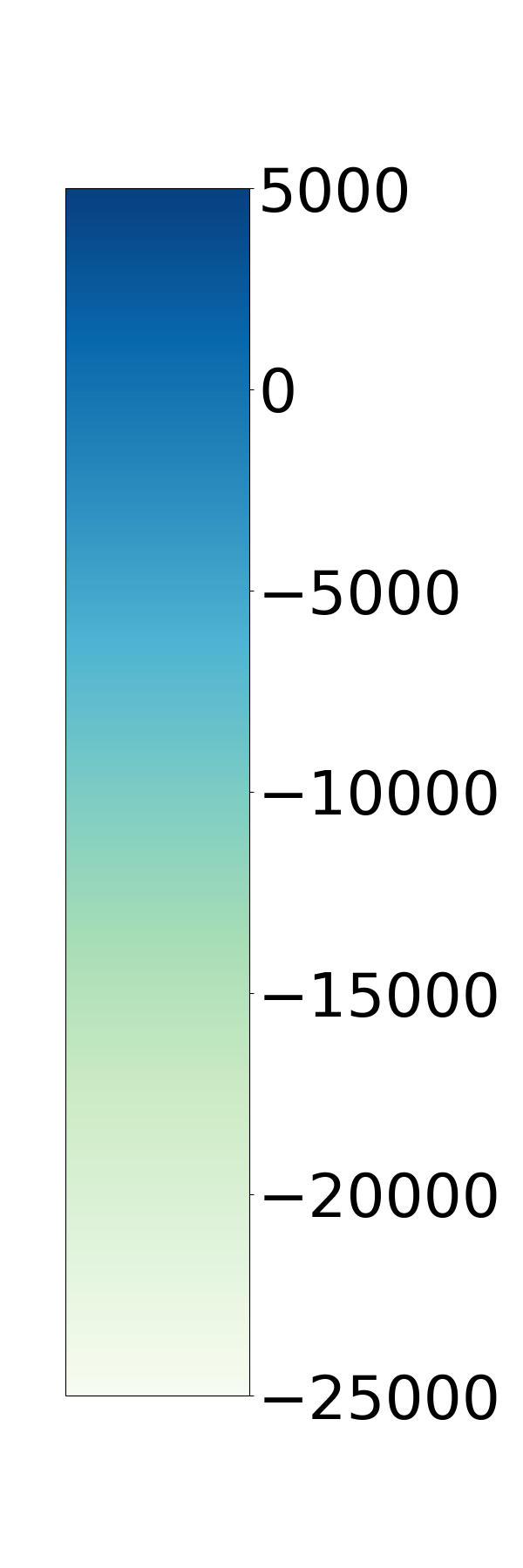}
    }
    & 
    \raisebox{-0.1\height}{
    %\begin{figure}%[h!]
%\centering
\scriptsize{
%\hspace{-0.0cm}
\begin{tikzpicture}
\begin{axis}[
    %no markers,
    title style={yshift=-2ex},
    title={\emph{value for the initial belief}},
    ylabel={},
    y label style={at={(axis description cs:0.075,.5)},anchor=south,
    /pgf/number format/fixed,
    /pgf/number format/precision=4
    },
    xlabel={$k$},
    xmin=0, xmax=16,
    xtick={0, 4, 8, 12, 16},
    ymin=-25000, ymax=5000,
    ytick={5000, 0, -5000, -10000, -15000, -20000, -25000},
    ymajorgrids=true,
    grid style=dashed,
    grid=both,
    height=4.5cm,
    width=0.45\textwidth,
    legend entries={
                {\emph{lower bound}},
                {\emph{upper bound}}
                },
    legend style={at={(0.95,0.05)},
                anchor=south east, 
                nodes={scale=0.9, transform shape}}            
]
\addlegendimage{mark=square*,cyan,mark size=1.5pt}
\addlegendimage{mark=*,teal,mark size=1.5pt}
]

\addplot[mark=square*,cyan,opacity=1.0,mark size=1.0pt] table [x=k, y=lb, col sep=comma]{figures/parking/4x4_3_initial_obstacle_-5000.csv};
\addplot[mark=*,teal,opacity=1.0,mark size=1.0pt] table [x=k, y=ub, col sep=comma]{figures/parking/4x4_3_initial_obstacle_-5000.csv};

\end{axis}
\end{tikzpicture}
}
%\caption{}
%\label{}
%\end{figure}
    }
    \\
    \includegraphics[height=0.3\textwidth]{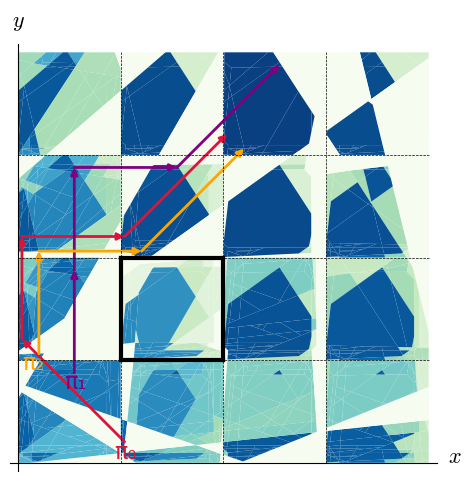}
    & 
    \raisebox{-0.065\height}{
    \includegraphics[height=0.32\textwidth]{figures/parking/colormap-5000.png}
    }
    & 
    \raisebox{-0.1\height}{
    %\begin{figure}%[h!]
%\centering
\scriptsize{
%\hspace{-0.0cm}
\begin{tikzpicture}
\begin{axis}[
    %no markers,
    title style={yshift=-2ex},
    title={\emph{value for the initial belief}},
    ylabel={},
    y label style={at={(axis description cs:0.075,.5)},anchor=south,
    /pgf/number format/fixed,
    /pgf/number format/precision=4
    },
    xlabel={$k$},
    xmin=0, xmax=20,
    xtick={0, 4, 8, 12, 16, 20},
    ymin=-25000, ymax=5000,
    ytick={5000, 0, -5000, -10000, -15000, -20000, -25000},
    ymajorgrids=true,
    grid style=dashed,
    grid=both,
    height=4.5cm,
    width=0.45\textwidth,
    legend entries={
                {\emph{lower bound}},
                {\emph{upper bound}}
                },
    legend style={at={(0.95,0.05)},
                anchor=south east, 
                nodes={scale=0.9, transform shape}}            
]
\addlegendimage{mark=square*,cyan,mark size=1.5pt}
\addlegendimage{mark=*,teal,mark size=1.5pt}
]

\addplot[mark=square*,cyan,opacity=1.0,mark size=1.0pt] table [x=k, y=lb, col sep=comma]{figures/parking/4x4_3_initial_obstacle_prob-5000.csv};
\addplot[mark=*,teal,opacity=1.0,mark size=1.0pt] table [x=k, y=ub, col sep=comma]{figures/parking/4x4_3_initial_obstacle_prob-5000.csv};

\end{axis}
\end{tikzpicture}
}
%\caption{}
%\label{}
%\end{figure}
    }
    \end{tabular}
    \vspace{-.25cm}
    \caption{Paths and values for car parking (obstacle indicated by a black border, $\beta = 0.8$, collision rewards equal to $-1000$ (top) and $-5000$ (middle, bottom)). The top two rows are for deterministic environments, the bottom row uses a probabilistic environment.}
    \label{fig:car_parking_safe_unsafe}
\end{figure}

The first case study is {\revise a modified version of} the dynamic vehicle parking problem from \egref{ex:parking:model}, which we extend both with obstacles and to a larger environment. {\revise We limit the agent to a single parking spot and consider both deterministic and probabilistic environments}. We were able to compute optimal strategies that lead the vehicle to the parking spot while avoiding obstacles (if present).

\startpara{$4 {\times} 4$ environment} To extend this example to the case when there is an obstacle region $\mathcal{R}_O = \{(x, y) \in \mathbb{R}^2 \mid 1 \leq x , y \leq 2\}$ (see \figref{fig:car_parking_obstacle}, left), {\revise we change the state reward function such that, for any $s = ((\mathit{ps},\per),s_E) \in S$: 
\[
r_S(s) = \left\{ \begin{array}{cl}
1000 & \mbox{if $f_\mathcal{R}^{\max}(\mathit{ps}) = \per$} \\
 - k_O & \mbox{if $s_E \in \mathcal{R}_O$} \\
0 & \mbox{otherwise}
\end{array} \right.
\]
where $k_O \in \mathbb{N}$, i.e., we add to the reward structure a negative reward if the vehicle hits the obstacle. The accuracy $\varepsilon$ is $10^{-3}$.}

\iffalse
\begin{figure}[t]
\centering
\hspace{-1.0cm}
\begin{subfigure}{0.45\textwidth}
    \centering
    \includegraphics[height=0.68\textwidth]{figures/parking/4x4_region_obstacle-1000.png}
\end{subfigure}
\hspace{-1.0cm}
\raisebox{-0.065\height}{
\begin{subfigure}{0.1\textwidth}
    \centering
    \includegraphics[height=3.4\textwidth]{figures/parking/colormap-1000.png}
\end{subfigure}
}
\raisebox{-0.1\height}{
\begin{subfigure}{0.45\textwidth}
    \centering
    \input{figures/parking/4x4_region_obstacle-1000.tex}
\end{subfigure}
}
\vspace{-.25cm}
\caption{Region-based paths and values for car parking with the obstacle, $\beta = 0.8$.}
\label{fig:car_parking_region}
\vspace{0.5cm}
\end{figure}
\fi

\begin{figure}
    \centering
    \begin{tabular}{lcr}
    \includegraphics[height=0.3\textwidth]{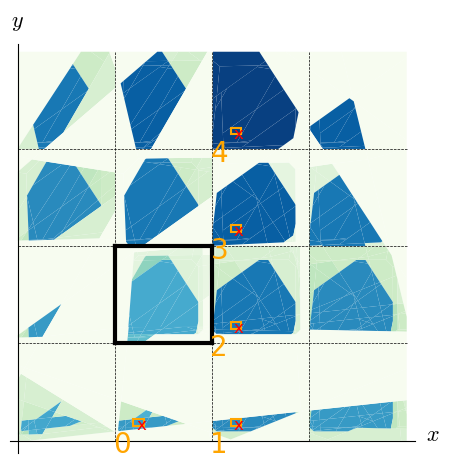}
    & 
    \raisebox{-0.065\height}{
    \includegraphics[height=0.32\textwidth]{figures/parking/colormap-1000.png}
    }
    & 
    \raisebox{-0.1\height}{
    %\begin{figure}%[h!]
%\centering
\scriptsize{
%\hspace{-0.0cm}
\begin{tikzpicture}
\begin{axis}[
    %no markers,
    title style={yshift=-2ex},
    title={\emph{value for the initial belief}},
    ylabel={},
    y label style={at={(axis description                cs:0.075,.5)},anchor=south,
    /pgf/number format/fixed,
    /pgf/number format/precision=4
    },
    xlabel={$k$},
    xmin=0, xmax=6,
    xtick={0,1,2,3,4,5,6},
    ymin=-5000, ymax=5000,
    ytick={-5000, -4000, -3000, -2000, -1000, 0,
            1000, 2000, 3000, 4000, 5000},
    ymajorgrids=true,
    grid style=dashed,
    grid=both,
    height=4.5cm,
    width=0.45\textwidth,
    legend entries={
                {\emph{lower bound}},
                {\emph{upper bound}}
                },
    legend style={at={(0.95,0.05)},
                anchor=south east, 
                nodes={scale=0.9, transform shape}}           
]
\addlegendimage{mark=square*,cyan,mark size=1.5pt}
\addlegendimage{mark=*,teal,mark size=1.5pt}
]

\addplot[mark=square*,cyan,opacity=1.0,mark size=1.0pt] table [x=k, y=lb, col sep=comma]{figures/parking/4x4_region_obstacle-1000.csv};
\addplot[mark=*,teal,opacity=1.0,mark size=1.0pt] table [x=k, y=ub, col sep=comma]{figures/parking/4x4_region_obstacle-1000.csv};

\end{axis}
\end{tikzpicture}
}
%\caption{}
%\label{}
%\end{figure}
    }
    \end{tabular}
    \vspace{-.25cm}
    \caption{Region-based paths and values for car parking with the obstacle, $\beta = 0.8$.}
    \label{fig:car_parking_region}
\end{figure}

\startpara{Strategy synthesis ($4 {\times} 4$)} 
\figref{fig:car_parking_safe_unsafe} presents paths ($\pi_1$, $\pi_2$ and $\pi_3$) for synthesised strategies starting from three particles in a given initial belief in two different scenarios, as well as the corresponding lower bound values for different regions of the environment. It also shows (on the right) the lower and upper bound values computed for the initial belief at each iteration.
In all cases, there is an obstacle, highlighted with black border.
We assume a uniform distribution over the points in the initial belief.

We consider strategies for when the {\revise reward $k_O$ associated to a collision equals $1000$ (\figref{fig:car_parking_safe_unsafe}, top), and when it equals $5000$} (\figref{fig:car_parking_safe_unsafe}, middle and bottom).
We see that, when the negative reward of a collision with the obstacle is increased, \figref{fig:car_parking_safe_unsafe} (middle), all the generated paths avoid the cell with the obstacle. We also see that, in the first step, the action chosen is to move {\em left}; while this is possible for path $\pi_0$ (red), taking that action from the other two initial belief points would take the agent out of the environment, in which case the agent would not move.

For the case {\revise when $k_O =  1000$}, \figref{fig:car_parking_safe_unsafe} (top), since the negative reward associated with a collision with the obstacle is lower, we see that %\marta{rephrased} 
such a reward can be compensated for by the agent afterwards, i.e., it can choose to move upwards from all points in the initial belief, resulting in a possibly unsafe strategy where a collision could happen. Finally, in \figref{fig:car_parking_safe_unsafe} (bottom), we plot the generated paths for the probabilistic environment of \egref{ex:parking:model}. Considering the initial position relative to the parking spot, the probabilistic environment can  be advantageous for the agent, as it is possible to reach the spot with fewer moves than in the deterministic case. This is indeed reflected through a slightly higher expected reward for the initial belief.

Similarly, \figref{fig:car_parking_region} shows values and strategies computed for the same scenario when considering a region-based belief. The regions reached from the initial position until arriving at the parking spot are indicated in orange, with the current state labelled by %\marta{there seem to be several x in red?}
\texttt{x}. %\gabriel{Changed.}. 
The lower and upper bound values
at each iteration are shown on the right-hand side, and the convergence demonstrates that the approximate upper bound for the region-based beliefs is tight if the belief has a unique region (see \lemaref{rb-upper-bound}). We notice that the synthesised strategy avoids the obstacle while also reaching the parking spot with the least number of possible steps, maximising the agent's reward. 

\figref{fig:4x4_alphas_regions} illustrates how computation progresses for Algorithm~\ref{alg:NS-HSVI}. Initially, we have an $\alpha$-function for each local state whose underlying structure is the same as the perception FCP (see \figref{fig:car_parking} right), with all regions initialised with the lower bound as described in Section~\ref{subsec:lb_up}. With each iteration, we refine the representation for the regions containing visited points and update their values. The figure shows the initial representation (left) and the maximum (over all local states) of the first 5, 25, and finally all the generated $\alpha$-functions, coinciding then with the values presented in \figref{fig:car_parking_safe_unsafe} (middle). We can see how the values for the regions progressively increase as the computation proceeds (top row, left to right), as well as how the subsequent representations are refinements of the initial FCP (bottom row).

\begin{figure}[t]
\centering
\vspace{-0.4cm}
\begin{minipage}{0.89\textwidth}
\begin{subfigure}{0.24\textwidth}
    \includegraphics[width=\textwidth]{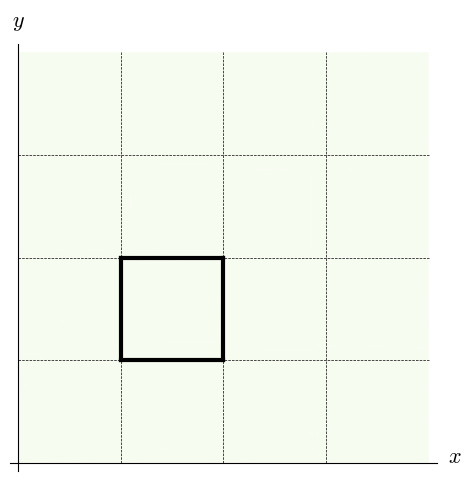}
\end{subfigure}
\begin{subfigure}{0.24\textwidth}
    \includegraphics[width=\textwidth]{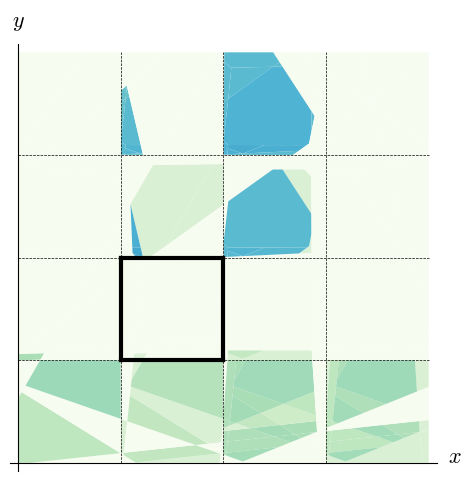}
\end{subfigure}
\begin{subfigure}{0.24\textwidth}
    \includegraphics[width=\textwidth]{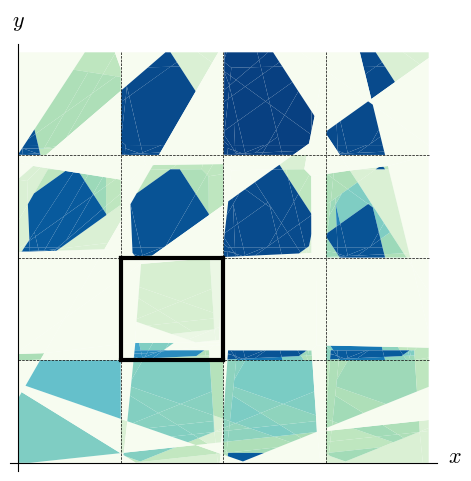}
\end{subfigure}
\begin{subfigure}{0.24\textwidth}
    \includegraphics[width=\textwidth]{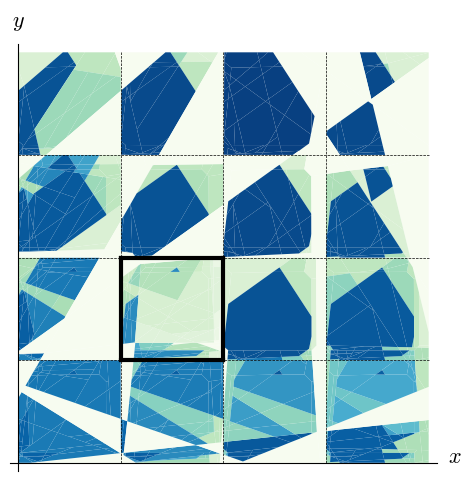}
\end{subfigure}

\begin{subfigure}{0.24\textwidth}
    \includegraphics[width=\textwidth]{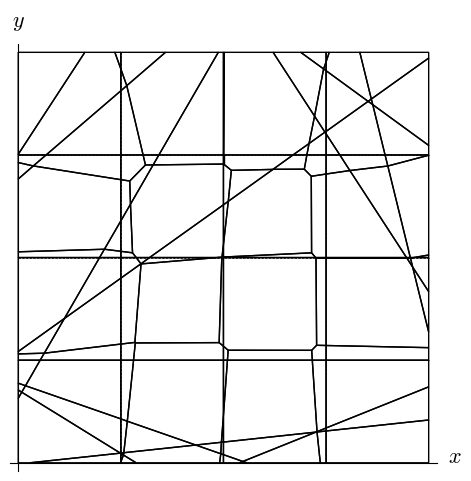}
\end{subfigure}
\begin{subfigure}{0.24\textwidth}
    \includegraphics[width=\textwidth]{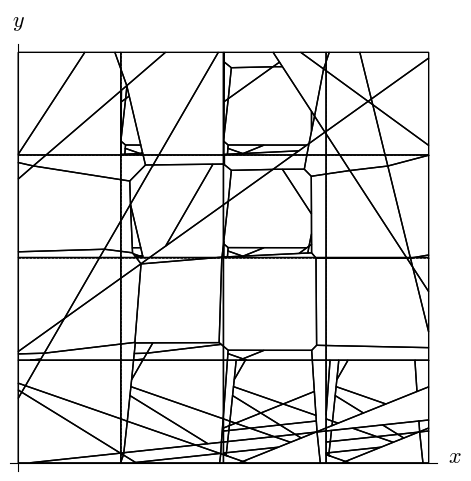}
\end{subfigure}
\begin{subfigure}{0.24\textwidth}
    \includegraphics[width=\textwidth]{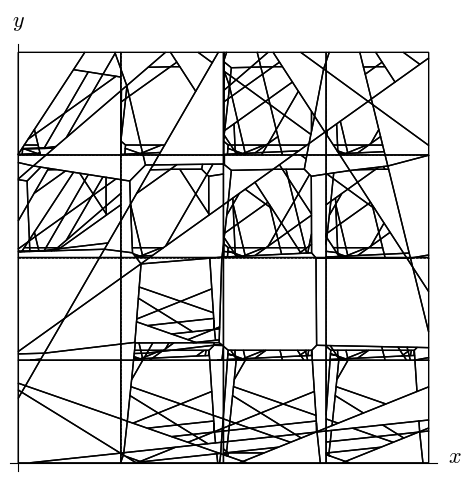}
\end{subfigure}
\begin{subfigure}{0.24\textwidth}
    \includegraphics[width=\textwidth]{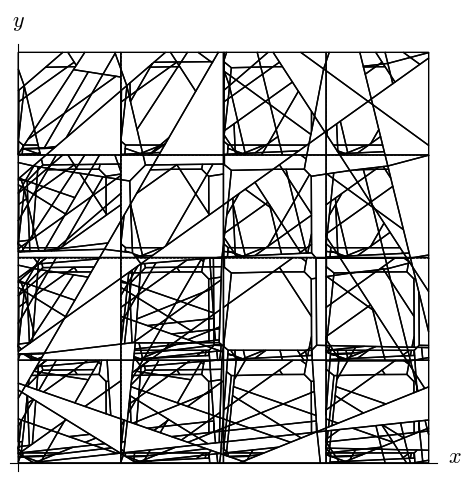}
\end{subfigure}
\end{minipage}
\begin{minipage}{0.1\textwidth}
    \includegraphics[width=\textwidth]{figures/parking/colormap-5000.png}
\end{minipage}

\caption{Values (top) and region outlines (bottom) for the initial and the maximum (over all local states) of the first $5$, $25$ and all the generated $\alpha$-functions (respectively from left to right) for the $4\times 4$ car parking example with an obstacle, $\beta = 0.8$.}
\label{fig:4x4_alphas_regions}
\end{figure}

\startpara{$8 {\times} 8$ environment} We consider a larger $8 {\times} 8$ environment $\mathcal{R}= \{(x, y) \in \mathbb{R}^2 \mid 0 \leq x, y \leq 8\}$ with 4 obstacles (\figref{fig:car_parking_obstacle}, right). In this model the parking spot is given by the region $\mathcal{R}_P = \{ (x,y) \in \mathbb{R}^2 \mid 6 \leq x  \leq 8 \wedge 7 \leq y \leq 8 \}$, and the same reward structure is considered.
To extend the NS-POMDP from Example~\ref{ex:parking:model} to this setting, {\revise there are 64 abstract grid cells (percepts) and the perception function is a feed-forward NN with one hidden ReLU layer and 15 neurons. \figref{fig:car_parking_8x8_VCAS_perceptions} (left) shows the perception FCP for the $8{\times} 8$ environment.}

\begin{figure}[t]
\centering
\begin{subfigure}{0.3\textwidth}
\includegraphics[width=\textwidth]{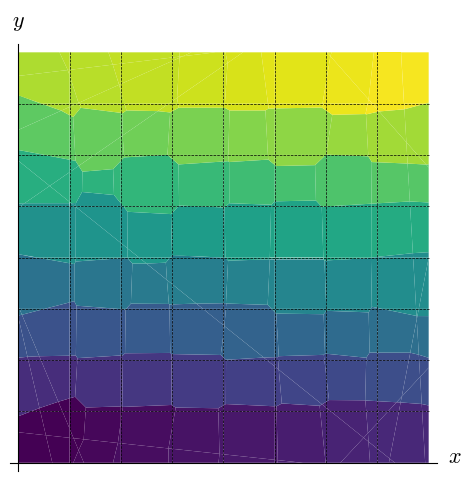}
\put(-120,-2){\scriptsize $0$}
\put(-89,-3){\scriptsize $2$}
\put(-64,-3){\scriptsize $4$}
\put(-38,-3){\scriptsize $6$}
\put(-13,-3){\scriptsize $8$}
\put(-120,28){\scriptsize $2$}
\put(-120,53){\scriptsize $4$}
\put(-120,79){\scriptsize $6$}
\put(-120,103){\scriptsize $8$}
\end{subfigure}
\hfil
\begin{subfigure}{0.48\textwidth}
\includegraphics[width=\textwidth]{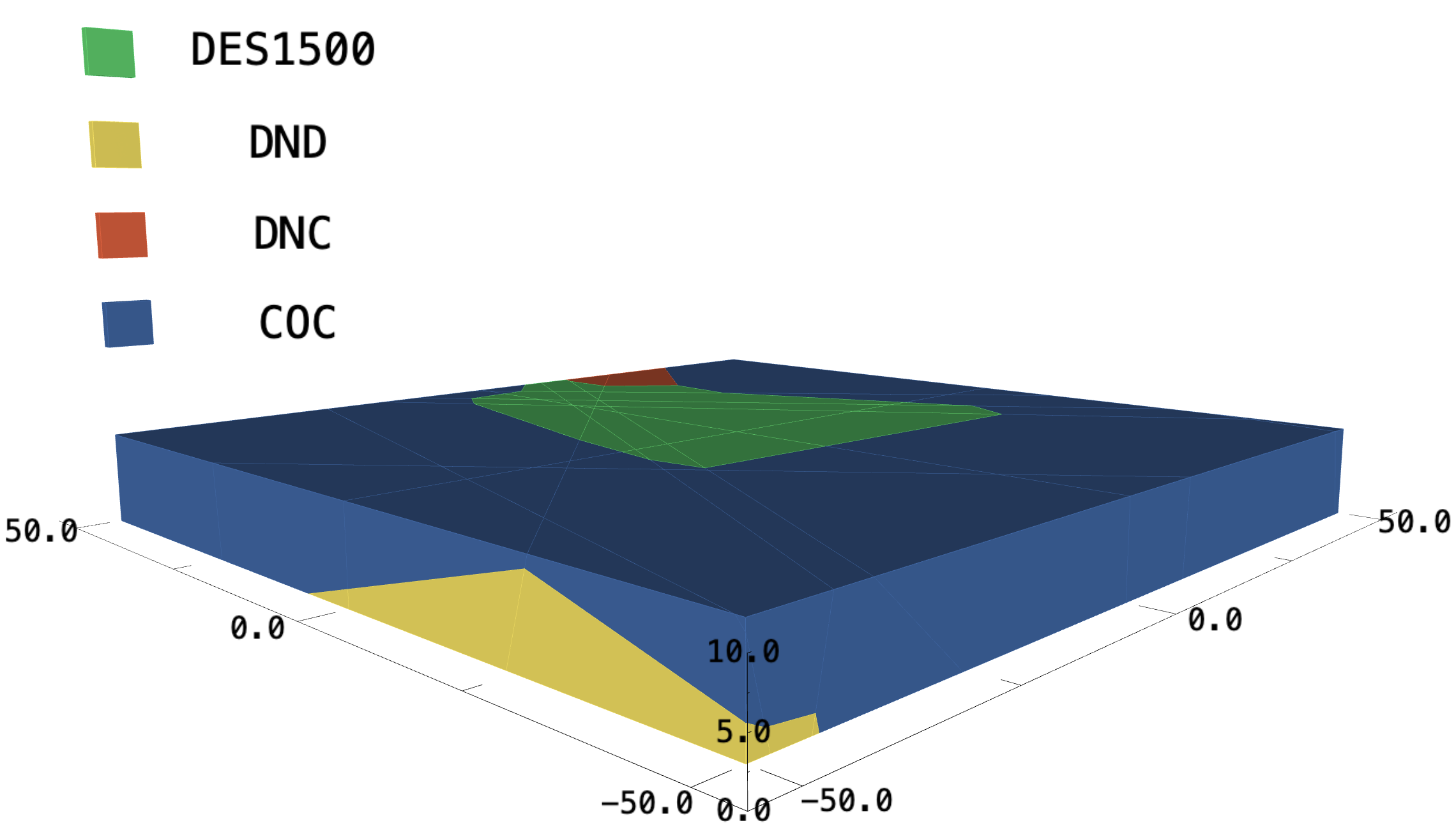}
\put(-180, 45){\scriptsize $t$}
\put(-30, 15){\scriptsize $h$}
\put(-160, 15){\scriptsize $\dot{h}_A$}
\end{subfigure}
\vspace*{0.0cm}
\caption{Perception FCP for car parking ($8{\times} 8$), and a slice of the perception FCP for the COC advisory of the VCAS ($h$ scaled 10:1).}
\label{fig:car_parking_8x8_VCAS_perceptions}
\end{figure}

\begin{figure}[t]
\centering
\vspace{-0.5cm}
\begin{subfigure}{0.42\textwidth}
\includegraphics[width=\textwidth]{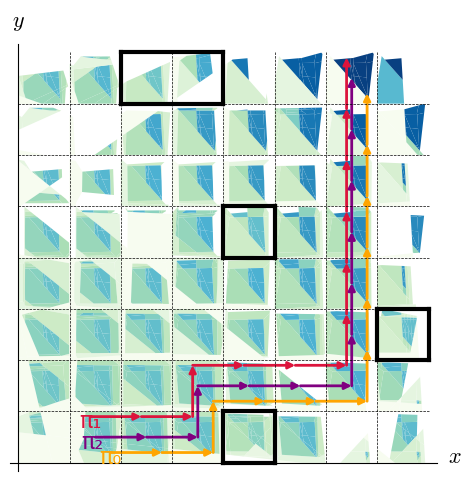}
\end{subfigure}
\hfil
\begin{subfigure}{0.48\textwidth}
\raisebox{-0.05\height}{%\begin{figure}%[h!]
%\centering
\scriptsize{
\hspace{-0.0cm}
\begin{tikzpicture}
\begin{axis}[
    %no markers,
    title style={yshift=-2ex},
    ylabel={},
    y label style={at={(axis description cs:0.075,.5)},anchor=south},
    xlabel={$k$},
    xmin=0, xmax=175, %34
    xtick={0, 50, 100, 150, 175}, %34
    ymin=-5000, ymax=5000,
    ytick={-5000, -3750, -2500, -1250, 0, 1250, 2500, 3750, 5000},
    ymajorgrids=true,
    grid style=dashed,
    grid=both,
    height=6.0cm,
    %height=0.7\textwidth,
    width=1\textwidth,
    legend entries={
                {\emph{lower bound}},
                {\emph{upper bound}}
                },
    legend style={at={(0.45,0.96)},
                anchor=north west, 
                nodes={scale=0.9, transform shape}}            
]
\addlegendimage{mark=square*,cyan,mark size=1.5pt}
\addlegendimage{mark=*,teal,mark size=1.5pt}
]

\addplot[mark=square*,cyan,opacity=1.0,mark size=0.5pt] table [x=k, y=lb, col sep=comma]{figures/parking/bounds_8x8.csv};

\addplot[mark=*,teal,opacity=1.0,mark size=0.5pt] table [x=k, y=ub, col sep=comma]{figures/parking/bounds_8x8.csv};

\end{axis}
\end{tikzpicture}
}
%\caption{}
%\label{}
%\end{figure}
}
\end{subfigure}
\vspace*{-0.2cm}
\caption{Paths and values for car parking ($8{\times}8$, $\beta = 0.8$, partially reconstructed).}
\label{fig:car_parking_8x8_plot}
\end{figure}

\startpara{Strategy synthesis ($8 {\times} 8$)} For this extended model {\revise and $k_O=1000$}, \figref{fig:car_parking_8x8_plot} (left) presents the paths from the three particles in the initial belief for the synthesised strategy, as well as lower bound values for the regions of the environment. As the figure demonstrates, the vehicle is able to reach the parking spot while avoiding the obstacles. 
As the full set of $\alpha$-functions is large (see Table~\ref{tab:strat_stats_car}), to reduce computational effort we show approximate values obtained by maximising over a set of sampled $\alpha$-functions. 
\figref{fig:car_parking_8x8_plot} (right) shows how the lower and upper bound values for the initial belief change as the number of iterations of the NS-HSVI algorithm increases.

\iffalse
\begin{figure}[t]
\centering
%\hspace{-1.0cm}
\raisebox{0.12\height}{
\begin{subfigure}{0.28\textwidth}
\includegraphics[width=\textwidth]{figures/parking/4x4_region_0241.png}
\end{subfigure}
}
\hfil
\raisebox{0.04\height}{
\begin{subfigure}{0.1\textwidth}
\includegraphics[height=3.1\textwidth]{figures/parking/colormap.png}
\end{subfigure}
}
\hfil
\begin{subfigure}{0.5\textwidth}
\input{figures/parking/bounds_region.tex}
\end{subfigure}
\vspace*{-0.2cm}
\caption{Region-based paths and values for car parking (no obstacles).}
\label{fig:car_parking_region_plot}
\end{figure}
\fi

\subsection{VCAS Case Study}

%\subsection{VCAS}.

In this case study there are two commercial aircraft: an ownship aircraft equipped with an NN-controlled vertical collision avoidance system (VCAS) and an intruder aircraft. 
Each second, the avoidance system  gives a vertical climb-acceleration  advisory $\ad$ to the pilot of the ownship to avoid near mid-air collisions (NMACs), which occur when the aircraft are separated by less than 100 ft vertically and 500 ft horizontally. 
The avoidance system extends the classical VCAS~\cite{KDJ-MJK:19}, both by adding trust to measure uncertainty and by allowing for deviations from the advisories arising from the additional belief information.
In contrast to the VCAS model of~\cite{KDJ-MJK:19}, we allow  a non-zero constant climb rate for the intruder. We were able to compute optimal strategies that safely guide the ownship by avoiding the collision zone. 

\iffalse
\begin{figure}[t]
\centering
\includegraphics[width=0.5\textwidth]{figures/fcp_coc_final.png}
\put(-175, 45){\scriptsize $t$}
\put(-30, 15){\scriptsize $h$}
\put(-160, 15){\scriptsize $\dot{h}_A$}
\caption{Slice of the FCP representation for the COC advisory ($h$ scaled 10:1).}
\label{FCP-COC}
\vspace*{-0.0cm}
\end{figure}
\fi

\startpara{VCAS as an NS-POMDP} The input to VCAS is a tuple $(h, \dot{h}_A, t)$, where $h$ is the relative altitude of the two aircraft, $\dot{h}_A$ the climb rate of ownship, and $t$ the time  until the loss of horizontal separation between the aircraft. VCAS is implemented via nine feed-forward NNs, %$f_i:\mathbb{R}^3 \to \mathbb{R}^9$ for $1 \leq i \leq 9$, 
each of which outputs the scores of nine possible advisories (see \tabref{tab:advisory}). Each advisory provides a set of {\revise suggested} acceleration values and the ownship then either accelerates at one of these values or does not accelerate.
Each NN of VCAS has one hidden ReLU layer with 16 neurons, and therefore the regions in its preimage are polytopes. If we had instead considered HorizontalCAS \cite{KDJ-MJK:19-2}, the nonlinear environment transition function twists polytopes into non-polytopes, and 
would destroy our finite representations. 

\begin{table*}[t]
\setlength{\tabcolsep}{5pt}  
\centering
\scriptsize{
\begin{tabular}{|c|l|l|c|} \hline
Label & \multirow{2}{*}{Advisory}  & \multirow{2}{*}{Description}  & \multirow{1}{*}{Actions} \\
$(ad_i)$ & & & \multirow{1}{*}{$\textup{ft/s}^2$}
 \\ \hline \hline	
1 &  COC  & Clear of Conflict & $-3$, $0$, $+3$
\\
2 & DNC  & Do Not Climb & $-9.33$, $-8.33$, $-7.33$ 
\\
3 &  DND & Do Not Descend & $+7.33$, $+8.33$, $+9.33$
\\
4 & DES1500 & Descend at least 1500 ft/min &    $-9.33$, $-8.33$, $-7.33$
\\
5 & CL1500 & Climb at least 1500 ft/min &  $+7.33$, $+8.33$, $+9.33$
\\
6 & SDES1500 & Strengthen Descend to at least 1500 ft/min &  $-11.7$, $-10.7$, $-9.7$
\\
7 & SCL1500 & Strengthen Climb to at least 1500 ft/min &  $+9.7$, $+10.7$, $+11.7$
\\
8 & SDES2500 & Strengthen Descend to at least 2500 ft/min & $-11.7$, $-10.7$, $-9.7$
\\ 
9 & SCL2500 & Strengthen Climb to at least 2500 ft/min & $+9.7$, $+10.7$, $+11.7$
\\ \hline
\end{tabular}}
\vspace*{0.0cm}
\caption{Suggested actions for each advisory of VCAS~\cite{MEA-EB-PK-AL:20}.}
\label{tab:advisory}
\end{table*}

We model VCAS as an NS-POMDP, in which the agent $\agent$ is the ownship.
The agent has four trust levels $\{1, \dots, 4\}$, which represent the trust it has in the previous advisory. These levels increase if the 
{\revise the executed action is \emph{compliant} with the current advisory (i.e., is one of the suggested actions listed in \tabref{tab:advisory})}, and decrease with probability $0.5$ otherwise. A local state of the agent is of the form $(\ad_\mathit{pre}, \tr)$ consisting of the previous advisory and the trust level, and the percept of the agent stores the current VCAS advisory.
An environment state is a tuple $(h, \dot{h}_A, t)$ corresponding to the input of VCAS. Formally, we have:

\begin{itemize}
    \item $S_A= \Loc \times \Per$ with $\Loc = \{1,\dots,9\} \times \{1,\dots,4\}$ and $\Per = \{1,\dots,9\}$;
    
    \item $S_E = [-2000, 2000] \times [-50, 50] \times [0, 20]$;
    
    \item $\Act = \{0, \pm3.0, \pm7.33, \pm8.33, \pm9.33, \pm9.7, \pm10.7, \pm11.7 \}$;
    
     \item $\Delta_A(\loc, \per) = \Act$ for all $\loc \in \Loc$ and $\per \in \Per$;

     \item $\obs_A((\ad_\mathit{pre}, \tr),s_E)=\argmax (f_{\ad_\mathit{pre}}(s_E))$, 
     where $f_{\ad_\mathit{pre}}$ is the NN associated with the previous advisory $\ad_\mathit{pre}$ and the boundary point is resolved by assigning the advisory with the smallest label in \tabref{tab:advisory}; 
     
     \item for $s_A=((\ad_\mathit{pre},  \tr),\ad) \in S_A$, $(\ad', \tr') \in \Loc$ and $a \in \Act$, if $a$ is compliant with $\ad$ (see \tabref{tab:advisory}), then:
     \[
     \delta_A(s_A,a)((\tr',\ad')) = \left\{ \begin{array}{cl}
     1  & \mbox{if $(\tr\leq3) \wedge (\tr'=\tr+1) \wedge (\ad'=\ad)$} \\
     1 & \mbox{if $(\tr=4) \wedge (\tr'=\tr) \wedge (\ad'=\ad)$} \\
     0 & \mbox{otherwise}
     \end{array}  \right.
     \]
    and if $a$ is not compliant with $\ad$, then:
     \[
     \delta_A(s_A,a)((\tr',\ad')) = \left\{ \begin{array}{cl}
     0.5 & \mbox{if $(\tr\geq 2) \wedge (\tr'=\tr-1) \wedge (\ad'=\ad)$} \\
     0.5 & \mbox{if $(\tr\geq 2) \wedge (\tr'=\tr) \wedge (\ad'=\ad)$} \\
     1 & \mbox{if $(\tr=1) \wedge (\tr'=\tr) \wedge (\ad'=\ad)$} \\
     0 & \mbox{otherwise;}
     \end{array}  \right.
     \]
     \item for $s=(h, \dot{h}_A, t),s'=(h', \dot{h}_A', t') \in S$ if
      \[ \begin{array}{rcl}
     h'' & = & h-\Delta t(\dot{h}_A-\dot{h}_{\textup{int}})-0.5\Delta t^2 \ddot{h}_A \\
     \dot{h}_A'' & = & \dot{h}_A+\ddot{h}_A\Delta t \\
     t'' & = & t-\Delta t
     \end{array} \]     
     then
     \[
     \delta_E(s,a)(s') = \left\{
     \begin{array}{cl}
     1 & \mbox{if $(h'', \dot{h}_A'', t'') \in S_E$ and $s'=(h'', \dot{h}_A'', t'')$} \\
     1 & \mbox{if $(h'', \dot{h}_A'', t'') \not\in S_E$ and $s'=s$} \\
     0 & \mbox{otherwise} \\
         \end{array} \right. \] 
     where $\Delta t = 1.0$ is the time step and the intruder is assumed to be a constant climb rate $\dot{h}_{\text{int}} = 30$. 
\end{itemize}
In the reward structure we consider, all action rewards are zero and the state reward function is such that for any $s \in S$: $r_S(s) = -1000$ if $t \in [0,1] \wedge h \in [-100,100]$ and $0$ otherwise,
i.e., there is a negative reward if altitudes of the aircraft are within 100 ft at time 0 or 1. The accuracy $\varepsilon$ is $10^{-1}$.

\startpara{Strategy synthesis} To compute the perception FCP $\Phi_P$, i.e., the preimages of the NNs for this case study, we first trained these NNs.
This involved computing an MDP table policy using local approximate value iteration, reformatting this into training data and training the NNs~\cite{KDJ-SS-JBJ-MJK:19}. To generate the preimages, we adapted the method of \cite{KM-FF:20}, which was used to compute exact preimages for the NNs of HorizontalCAS \cite{KDJ-MJK:19-2}. For example, the preimage for the COC (Clear of Conflict) advisory is shown in \figref{fig:car_parking_8x8_VCAS_perceptions} (right), which shows VCAS next issuing the advisory DES1500 (Descend at least 1500 ft/min) for the environment states in the green region to avoid an NMAC given the small values of $h$ and $t$.

\figref{fig:vcas_safe} shows the paths from the four particles in the initial belief of a synthesised strategy for the VCAS case study. For the particles that would reach the collision zone at time 0 or 1 (coloured green in \figref{fig:vcas_safe}), there is a course correction that enables the ownship to narrowly escape a collision. 

\begin{figure}[t]
% \vspace*{-0.75cm}
\centering
\includegraphics[width=0.95\textwidth]{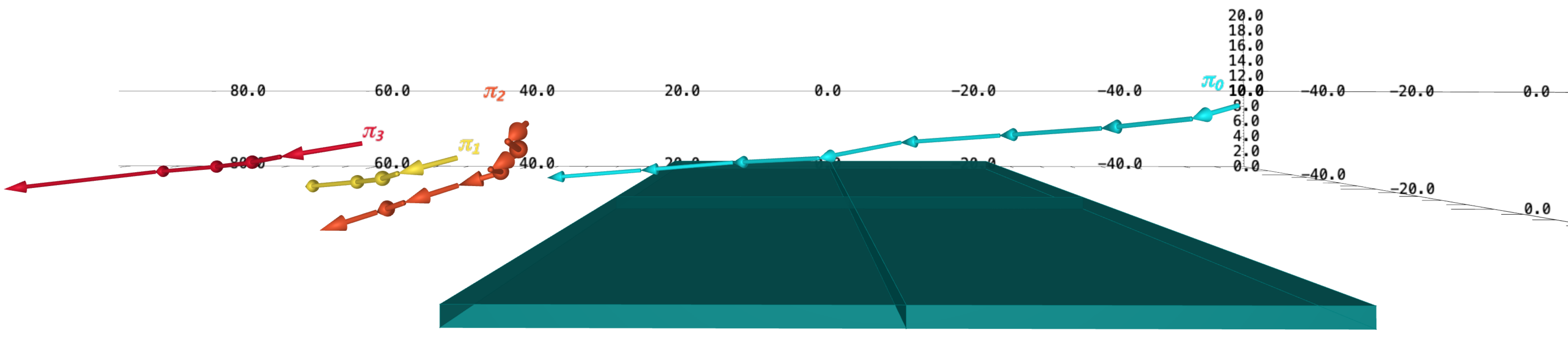}
\put(-88, 75){\small $t$}
\put(-175, -5){\small $h$}
\put(-46, 44){\small $\dot{h}_A$}
\vspace*{-0.2cm}
\caption{Paths from synthesised safe strategies for VCAS ($h$ scaled 5:1).}
\label{fig:vcas_safe}
% \vspace*{-0.4cm}
\end{figure}

\subsection{Performance Analysis}
To conclude the experimental analysis, we first discuss the performance of the implementation based on the statistics for two case studies, and then compare the performance when using particle-based and region-based beliefs, and against SARSOP.
\begin{table}[t]
\scriptsize
\setlength{\tabcolsep}{3.5pt}  
\centering
{
\begin{tabular}{|c|c||c|c|c||c|c|} \hline
\multirow{2}{*}{Model} & \multirow{2}{*}{Belief type} & Initial  & 
Discount & Pts./vol. &  
\multirow{2}{*}{\shortstack[c]{Iter.}} & 
Comput. \\ 
& & pts./regions & factor & updated & & time(s) \\ \hline \hline
\multirow{3}{*}{\shortstack[c]{Car parking \\ {\tiny (no obstacles, $4 {\times} 4$)}}} & \multirow{2}{*}{\shortstack[c]{Particle-based}} & 3 & 0.8 & 205 & 15 & 32.7 \\ \cline{3-7}
& & 5 & 0.8 & 392 & 11 & 36.3 \\ \cline{2-7}
& {\shortstack[c]{Region-based}} & 1 & 0.8 & 7.6 & 15 & 99.1 \\ 
\hline
\multirow{4}{*}{\shortstack[c]{Car parking \\ {\tiny (w/ obstacle, $4 {\times} 4$)}}} & \multirow{3}{*}{\shortstack[c]{Particle-based}} & 3 & 0.8 & 210 & 17 & 46.2 \\ \cline{3-7}
& & {\bf 3} & {\bf 0.8} & {\bf 199} & {\bf 21} & {\bf 87.7} \\ \cline{2-7}
& & 5 & 0.8 & 390 & 15 & 41.9 \\ \cline{2-7}
& {\shortstack[c]{Region-based}} & 1 & 0.8 & 7.6 & 8 & 80.4 \\ \hline
\multirow{3}{*}{\shortstack[c]{Car parking \\ {\tiny (w/ obstacles, $8 {\times} 8$)}}} & \multirow{2}{*}{\shortstack[c]{Particle-based}} & 3 & 0.8 & 960 & 174 & 1820 \\ \cline{3-7}
& & 5 & 0.8 & 1600 & 119 & 1337 \\ \cline{2-7} %\hline
& {\shortstack[c]{Region-based}} & 1 & 0.8 & 43.2 & 59 & 2075 \\ \hline
\multirow{4}{*}{\shortstack[c]{VCAS \\ {\tiny (3 actions)}}} & \multirow{3}{*}{\shortstack[c]{Particle-based}} & 4 & 0.75 & 441 & 40 & 228.2 \\ \cline{3-7}
& & 5 & 0.75 & 649 & 43 & 475.6 \\ \cline{3-7}
& & 6 & 0.75 & 476 & 23 & 1467 \\ \cline{2-7}
& {\shortstack[c]{Region-based}} & 1 & 0.75 & 4725 & 10 & 994.6 \\
\hline
\multirow{4}{*}{\shortstack[c]{VCAS \\ {\tiny (15 actions)}}} & \multirow{3}{*}{\shortstack[c]{Particle-based}} & 4 & 0.75 & 259 & 11 & 183.7 \\ \cline{3-7}
& & 5 & 0.75 & 425 & 15 & 357.7 \\ \cline{3-7}
& & 6 & 0.75 & 228 & 6 & 127.4 \\ \cline{2-7}
& {\shortstack[c]{Region-based}} & 1 & 0.75 & 4059 & 7 & 2419 \\
\hline
\end{tabular}}
\vspace*{0.0cm} 
\caption{Statistics for a set of NS-POMDP solution instances. {\revise The bold entries are for the instance with a probabilistic environment.}}
\label{tab:stats}
% \vspace*{-0.3cm}
\end{table}
\startpara{Experimental results} The experimental results reported in this section were generated on a 2.10GHz Intel Xeon Gold. Our NS-HSVI implementation is able to compute values and strategies for particle-based and region-based instances of the models we considered in less than 1 hour (\tabref{tab:stats}). In the table, we report the model we consider, the belief type, the number of initial points or regions, the discount factor ($\beta$), the number of updated points or the volume of the updated regions (depending on the belief type), and the overall number of iterations of Algorithm~\ref{alg:NS-HSVI} as well as the time taken until convergence. We found that the branching factor of the environment transition function, the number of agent states and actions, and the number of polyhedra in the perception FCP $\Phi_P$ can all have a significant impact on the computation time. \tabref{tab:stats} shows that computation for region-based beliefs normally takes longer because the number of regions of the perception FCP $\Phi_P$ over which the algorithm puts positive probabilities is usually larger, and thus it requires more ISPP backups. Moreover, while the update for particle-based beliefs only involves simple operations, updating region-based beliefs is far more complex due to the need of the polyhedra image computations, intersections and volume calculations. 

\begin{figure}[t]
\centering
\hspace{-.7cm}
\begin{subfigure}{0.4\textwidth}
    \centering
    %\begin{figure}%[h!]
%\centering
\scriptsize{
\hspace{-0.0cm}
\begin{tikzpicture}
\begin{axis}[
    %no markers,
    title style={yshift=-2ex},
    title={},
    ylabel={\emph{time(s)}},
    xlabel={$\beta$},
    xmin=0.5, xmax=0.95,
    xtick={0.5,0.6,0.7,0.8,0.9,0.95},
    xticklabel={
        \pgfmathparse{\tick <= 0.9 ? 10*\tick : 100*\tick} .\pgfmathprintnumber\pgfmathresult
    },
    ymin=15, ymax=90,
    ytick={15,30,45,60,75,90},
    ymajorgrids=true,
    grid style=dashed,
    height=5.0cm,
    width=0.9\textwidth,
    legend entries={
                {\emph{$4{\times4}$, no obstacles}},
                {\emph{$4{\times}4$, with obstacle}},
                {\emph{$8{\times}8$, with obstacles}}
                },
    legend style={at={(0.02,0.96)},
                anchor=north west, 
                nodes={scale=0.9, transform shape}}            
]
]

\addplot[mark=square*,teal,opacity=1.0,mark size=1.5pt,thick] table [x=beta, y=time,col sep=comma]{figures/parking/4x4_3_initial_no_obstacle_discounts.csv};
\addplot[mark=triangle*,red,opacity=1.0,mark size=1.5pt,thick] table [x=beta, y=time,col sep=comma]{figures/parking/4x4_3_initial_obstacle_discounts.csv};
% \addplot[mark=o,orange,opacity=1.0,mark size=1.5pt,thick] table [x=beta, y=time,col sep=comma]{figures/parking/8x8_3_initial_obstacles_discounts.csv};

\end{axis}
\end{tikzpicture}
}
%\caption{}
%\label{}
%\end{figure}
\end{subfigure}
\hspace{-1.0cm}
\begin{subfigure}{0.38\textwidth}
    \centering
    %\begin{figure}%[h!]
%\centering
\scriptsize{
\hspace{-0.0cm}
\begin{tikzpicture}
\begin{axis}[
    %no markers,
    title style={yshift=-2ex},
    xlabel={$\beta$},
    xmin=0.5, xmax=0.95,
    xtick={0.5,0.6,0.7,0.8,0.9,0.95},
    xticklabel={
        \pgfmathparse{\tick <= 0.9 ? 10*\tick : 100*\tick} .\pgfmathprintnumber\pgfmathresult
    },
    ymin=1000, ymax=1500,
    ytick={1000,1100,1200,1300,1400,1500},
    ymajorgrids=true,
    grid style=dashed,
    height=5.0cm,
    width=\textwidth,
    legend entries={
                {\emph{$8 \times 8$, with obstacles}}
                },
    legend style={at={(0.02,0.96)},
                anchor=north west, 
                nodes={scale=0.9, transform shape}}            
]
]

\addplot[mark=o,orange,opacity=1.0,mark size=1.5pt,thick] table [x=beta, y=time,col sep=comma]{figures/parking/8x8_5_initial_obstacles_discounts.csv};

\end{axis}
\end{tikzpicture}
}
%\caption{}
%\label{}
%\end{figure}
\end{subfigure}
\hspace{-1.0cm}
\begin{subfigure}{0.37\textwidth}
    \centering
    %\begin{figure}%[h!]
%\centering
\scriptsize{
\hspace{-0.0cm}
\begin{tikzpicture}
\begin{axis}[
    %no markers,
    title style={yshift=-2ex},
    title={},
    xlabel={$\beta$},
    xmin=0.5, xmax=0.95,
    xtick={0.5,0.6,0.7,0.8,0.9,0.95},
    xticklabel={
        \pgfmathparse{\tick <= 0.9 ? 10*\tick : 100*\tick} .\pgfmathprintnumber\pgfmathresult
    },
    ymin=0, ymax=610,
    ytick={0,120,240,360,480,600},
    ymajorgrids=true,
    grid style=dashed,
    height=5.0cm,
    width=\textwidth,
    legend entries={
                {\emph{VCAS}},
                },
    legend style={at={(0.02,0.96)},
                anchor=north west, 
                nodes={scale=0.9, transform shape}}            
]
]
\addplot[mark=+,cyan,opacity=1.0,mark size=1.5pt,thick] table [x=beta, y=time,col sep=comma]{figures/vcas/vcas_3_initial_discounts.csv};
\end{axis}
\end{tikzpicture}
}
%\caption{}
%\label{}
%\end{figure}
\end{subfigure}
%\vspace*{-0.2cm}
\caption{Solution times for different discount factors (for particle-based beliefs).}
\label{fig:disc}
\end{figure}
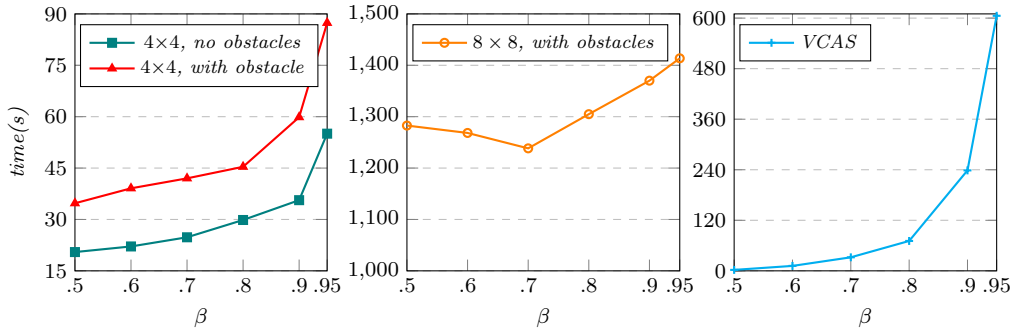

Another crucial aspect is the choice of the discount factor ($\beta$). \figref{fig:disc} shows how verification times vary for the different case studies as a function of that parameter. As expected, the trend we are able to observe is that it takes longer for the algorithm to converge as the value of $\beta$ increases. The small drop in the curve for the $8\times 8$ version of the car parking example for the lower values of $\beta$ can be explained by the inherent nondeterminism of HSVI exploration, especially in the early stages of the computation when many regions may have the same lower and upper bounds. This may lead to the algorithm being indifferent with respect to the actions it takes, and thus constructing paths that have lower impact on the values of the initial belief.

Finally, another element that impacts the running time is the choice of the initial belief and the model's dynamics. This is especially evident when comparing the two instances of VCAS. The beliefs for the version with 15 actions have lower values for $t$ and are thus much closer to the boundaries of the environment, which considerably limits the number of reachable states and makes it possible for the algorithm to converge more quickly despite the higher number of actions.

\begin{table}[t]
\scriptsize
\setlength{\tabcolsep}{3.5pt}  
\centering
{
\begin{tabular}{|c|c||c|c|c||c|c|} \hline
\multirow{2}{*}{Model} & Belief type & Total %/avg. 
 regions & Lower & Upper & Strat.  \\ 
 & \#initial & ($\alpha$-functions) & bound & bound & time (s)  \\ \hline \hline
\multirow{3}{*}{\shortstack[c]{Car parking \\ {\tiny(no obstacles, $4 {\times} 4$)}}} & PB, 3 & 80,494  
& 2389.3309 & 2389.3333 & 19.3 \\ \cline{2-6}
 & PB, 5 & 42,224
& 2047.9989 & 2048.0000 & 14.0  \\ \cline{2-6}
 & RB, 1 & 36,467
& 2047.9992 & 2048.0000 & 50.0  \\ \hline
\multirow{4}{*}{\shortstack[c]{Car parking \\ {\tiny(w/ obstacle, $4 {\times} 4$)}}} & PB, 3 & 99,513
& 2218.6653 & 2218.6666 & 24.5  \\ \cline{2-6}
 & {\bf PB, 3} & {\bf 91,808}
& {\bf 1554.3770} & {\bf 1554.3897} & {\bf 20.0}  \\ \cline{2-6}
 & PB, 5 & 47,719
& 2047.9990 & 2048.0000 & 14.2 \\ \cline{2-6}
 & RB, 1 & 35,751
& 2047.9988 & 2048.0000 & 39.4  \\ \hline
\multirow{3}{*}{\shortstack[c]{Car parking \\ {\tiny(w/ obstacles, $8 {\times} 8$)}}} & PB, 3 & 1,410,799
& 343.5969 & 343.5974 & 338.9  \\ \cline{2-6}
 & PB, 5 & 547,753 
& 343.5970 & 343.5974 & 158.4  \\ \cline{2-6} %\hline
& RB, 1 & 550,685
 & 343.5964 & 343.5974 & 473.8  \\ \hline
\end{tabular}}
\vspace*{0.0cm}
\caption{Further statistics for a set of NS-POMDP solution instances of the car parking case study. {\revise The bold entries are for the instance with a probabilistic environment.}}
\label{tab:strat_stats_car}
%\vspace*{-0.2cm}
\end{table}

\begin{table}[t]
\scriptsize
\setlength{\tabcolsep}{3.5pt}  
\centering
{
\begin{tabular}{|c|c||c|c|c||c|c|c|c|} \hline
\multirow{2}{*}{Model} & Belief type & Total %/avg. 
 regions & Lower & Upper & Strat. & Following & Avg. \\ 
 & \#initial & ($\alpha$-functions) & bound & bound & time (s) & ratio & trust\\ \hline \hline
\multirow{4}{*}{\shortstack[c]{VCAS \\ {\tiny (3 actions)}}} & PB, 4 & 154,009
& -1.2281 & 0.0 & 75.3 & - & - \\ \cline{2-8}
 & PB, 5 & 278,447 
& -1.2398 & 0.0 & 127.5 & - & -  \\ \cline{2-8}
 & PB, 6 & 868,257
& -0.2498 & 0.0 & 400.8 & - & -\\ \cline{2-8}
 & RB, 1 & 22,919
& -0.0715 & 0.0 & 65.5 & - & - \\ \cline{2-8}
\hline
\multirow{4}{*}{\shortstack[c]{VCAS \\ {\tiny (15 actions)}}} & PB, 4 & 32,387
& -0.6718 & 0.0 & 18.7 & 33\% & 1.3 \\ \cline{2-8}
 & PB, 5 & 30,003
& -0.9874 & 0.0 & 21.7 & 0\% & 1.0 \\ \cline{2-8}
 & PB, 6 & 19,218
& -1.0789 & 0.0 & 13.0 & 33\% & 1.3 \\ \cline{2-8}
 & RB, 1 & 21,102
& -0.6133 & 0.0 & 49.9 & 0\% & 1.0 \\ \cline{2-8}
\hline
\end{tabular}}
\vspace*{0.0cm}
\caption{Further statistics for a set of NS-POMDP solution instances of the VCAS case study.}
\label{tab:strat_stats_vcas}
%\vspace*{-0.2cm}
\end{table}
%I left colour black so we know which ones are not from lovelace

Tables~\ref{tab:strat_stats_car} and \ref{tab:strat_stats_vcas} shows, for a number of instances of each case study and for each belief type, particle-based (PB) and region-based (RB):
the total number of polyhedra that make up the $\alpha$-functions computed, the lower and upper bounds on values for the initial belief and the time required for strategy synthesis, i.e., reading $\alpha$-functions, finding maximum actions and updating beliefs.
For the VCAS case study Table~\ref{tab:strat_stats_vcas} also show the compliance ratio with respect to the suggested actions as well as average trust values over 20 paths generated from the synthesised strategies.

For the car parking case study
(recall the accuracy is $10^{-3}$), in general, the more iterations are needed for convergence, the higher the number of $\alpha$-functions generated and 
consequently the total number of regions. Strategy synthesis for region-based beliefs tends to be comparatively slower due to the complexity of the mathematical operations involved.

Regarding VCAS, the statistics in Table~\ref{tab:strat_stats_vcas} 
are for
the accuracy of $10^{-1}$. 
The $\alpha$-functions generally have a large number of regions, as the perception FCP for each of the 9 NNs of VCAS has many regions, and hence many intersections. 
In addition, we note that, for this model, the following ratio and average trust values are low, and in fact have been omitted for the model with 3 actions. This is because (see Table~\ref{tab:advisory}) the number of suggested actions associated to each advisory is only a fraction of the 15 actions we considered and, for a given belief, there are many strategies that can lead to the optimal value. Recall also that it is assumed that the intruder aircraft is always climbing and the beliefs we considered were all reasonably close to the collision zone. We analysed the synthesised strategies and found that, in many cases, the agent chose actions that would at first lead to a faster descent than those suggested in \tabref{tab:advisory}, but then compensated by descending
less, or not at all, at later stages. While the values of the actions differed, all strategies we observed led to the ownship lowering its altitude, which would lead to an increase of the overall height difference so as to escape a potential collision. 
Thus, 
% these measures 
the low following ratios do not reflect an inadequacy of the advisories.

\iffalse
Figs.~\ref{fig:car_parking_no_obstacle_plot} and \ref{fig:car_parking_obstacle_plot} (right) show the lower and upper bound values %for the initial belief 
at each iteration. % (166.66) than the case without obstacle (208.33), 
Although the minimal steps for each particle to reach the parking spot for two cases are the same, the difference in values is a result of the irregular shapes of the perception FCP (see \figref{fig:car_parking}). The shape of top left grid in the perception FCP means that after one step, in \figref{fig:car_parking_obstacle_plot}, the three particles are indistinguishable, and thus, since $\mathit{right}$ would cause the path from $\pi_1$ to hit the obstacle, $\mathit{up}$ is taken. However, $up$ creates a self-loop for the path from $\pi_2$, because the vehicle does not move if the action takes the vehicle outside of $\mathcal{R}$, which leads to a longer parking time. % and thus a smaller optimal value.
\fi

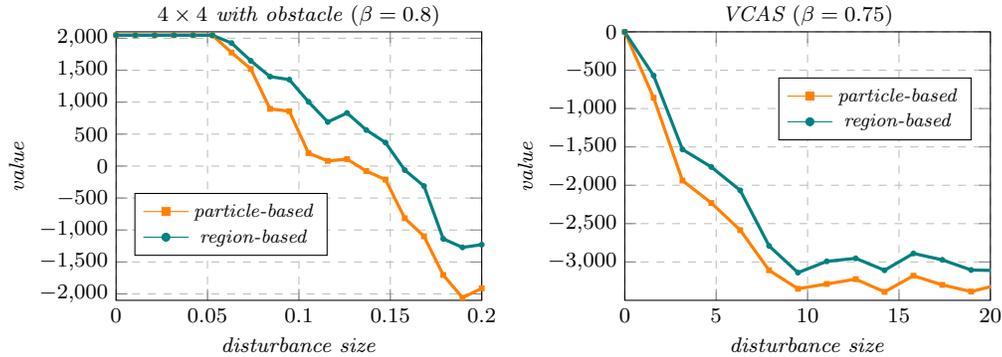
\begin{figure}[t]
\centering
\begin{subfigure}{0.47\textwidth}
%robustness_data_points = "dp, avgp, maxp, minp\n"
%robustness_data_regions = "dr, avgr, maxr, minr\n"
\newcommand{\colorp}{orange}
\newcommand{\colorr}{teal}
\scriptsize{
\hspace{-0.0cm}
\begin{tikzpicture}
\begin{axis}[
    %no markers,
    title style={yshift=-2ex},
    title={\emph{$4 \times 4$ with obstacle $(\beta = 0.8)$}},
    ylabel={\emph{value}},
    xlabel={\emph{disturbance size}},
    xmin=0, xmax=0.20, %34
    xtick={0,0.05,0.10,0.15,0.20,0.25,0.30}, %34
    ymin=-2100, ymax=2100,
    ytick={-2000,-1500,-1000, -500, -0, 500, 1000, 1500, 2000},
    xmajorgrids=true,
        xticklabel style={
        /pgf/number format/fixed,
        /pgf/number format/precision=3
    },
    scaled x ticks=false,
    grid style=dashed,
    grid=both,
    height=0.8\textwidth,
    width=\textwidth,
    legend entries={
                {\emph{particle-based}},
                {\emph{region-based}}
                },
    legend style={at={(0.05,0.4)},
                anchor=north west, 
                nodes={scale=0.9, transform shape}}            
]
\addlegendimage{mark=square*,\colorp,mark size=1.5pt}
\addlegendimage{mark=*,teal,\colorr,mark size=1.5pt}
]

% \addplot[name path=minp, mark=,\colorp,opacity=0.5,mark size=0.5pt] table [x=dp, y=minp, col sep=comma]{figures/parking/robustness_parking_data_points.csv};

\addplot[mark=square*,very thick,\colorp,opacity=1.0,mark size=0.5pt] table [x=dp, y=avgp, col sep=comma]{figures/parking/robustness_parking_data_points.csv};

% \addplot[name path=maxp, mark=,\colorp,opacity=0.5,mark size=0.5pt] table [x=dp, y=maxp, col sep=comma]{figures/parking/robustness_parking_data_points.csv};

% \addplot[opacity=.40,\colorp!20] fill between[of=minp and maxp];

%%

% \addplot[name path=minr, mark=,\colorr,opacity=0.5,mark size=0.5pt] table [x=dr, y=minr, col sep=comma]{figures/parking/robustness_parking_data_regions.csv};

\addplot[mark=*,very thick,\colorr,opacity=1.0,mark size=0.5pt] table [x=dr, y=avgr, col sep=comma]{figures/parking/robustness_parking_data_regions.csv};

% \addplot[name path=maxr, mark=,\colorr,opacity=0.5,mark size=0.5pt] table [x=dr, y=maxr, col sep=comma]{figures/parking/robustness_parking_data_regions.csv};

% \addplot[opacity=.40,\colorr!20] fill between[of=minr and maxr];

\end{axis}
\end{tikzpicture}
}
\end{subfigure}
\hfil
%\raisebox{-0.07\height}{
\begin{subfigure}{0.47\textwidth}
%robustness_data_points = "dp, avgp, maxp, minp\n"
%robustness_data_regions = "dr, avgr, maxr, minr\n"
\newcommand{\colorp}{orange}
\newcommand{\colorr}{teal}
\scriptsize{
\hspace{-0.0cm}
\begin{tikzpicture}
\begin{axis}[
    %no markers,
    title style={yshift=-2ex},
    title={\emph{VCAS $(\beta=0.75)$}},
    ylabel={\emph{value}},
    xlabel={\emph{disturbance size}},
    xmin=0, xmax=20, %30
    xtick={0,5, 10, 15, 20, 25, 30, 35, 40, 45, 50}, %34
    ymin=-3500, ymax=0,
    ytick={0, -500, -1000, -1500, -2000, -2500, -3000},
    xmajorgrids=true,
        xticklabel style={
        /pgf/number format/fixed,
        /pgf/number format/precision=3
    },
    scaled x ticks=false,
    grid style=dashed,
    grid=both,
    height=0.8\textwidth,
    width=\textwidth,
    legend entries={
                {\emph{particle-based}},
                {\emph{region-based}}
                },
    legend style={at={(0.42,0.83)},
                anchor=north west, 
                nodes={scale=0.9, transform shape}}            
]
\addlegendimage{mark=square*,\colorp,mark size=1.5pt}
\addlegendimage{mark=*,teal,\colorr,mark size=1.5pt}
]

% \addplot[name path=minp, mark=,\colorp,opacity=0.5,mark size=0.5pt] table [x=dp, y=minp, col sep=comma]{figures/vcas/robustness_vcas_data_points.csv};

\addplot[mark=square*,very thick,\colorp,opacity=1.0,mark size=0.5pt] table [x=dp, y=avgp, col sep=comma]{figures/vcas/robustness_vcas_data_points_new.csv};

% \addplot[name path=maxp, mark=,\colorp,opacity=0.5,mark size=0.5pt] table [x=dp, y=maxp, col sep=comma]{figures/vcas/robustness_vcas_data_points.csv};

%\addplot[opacity=.40,\colorp!20] fill between[of=minp and maxp];

%%

% \addplot[name path=minr, mark=,\colorr,opacity=0.5,mark size=0.5pt] table [x=dr, y=minr, col sep=comma]{figures/vcas/robustness_vcas_data_regions.csv};

\addplot[mark=*,very thick,\colorr,opacity=1.0,mark size=0.5pt] table [x=dr, y=avgr, col sep=comma]{figures/vcas/robustness_vcas_data_regions_new.csv};

% \addplot[name path=maxr, mark=,\colorr,opacity=0.5,mark size=0.5pt] table [x=dr, y=maxr, col sep=comma]{figures/vcas/robustness_vcas_data_regions.csv};

%\addplot[opacity=.40,\colorr!20] fill between[of=minr and maxr];

\end{axis}
\end{tikzpicture}
}
\end{subfigure}
%}
\vspace*{-0.5cm}
\caption{Comparison between particle-based and region-based values.}
\label{fig:car_parking_robustness}
\end{figure}

\startpara{Performance comparison} Finally, we compare values obtained for particle-based and region-based initial beliefs, where the initial region covers the particles, after they have been disturbed by shifting their position along a sampled direction. This models a realistic scenario, in which the actual initial belief differs from the initial belief used to compute offline lower and upper bound functions, for example due to measurement imprecision.
For a range of disturbance sizes (the distances by which the particles are shifted), the lower bound values for the average of 100 sampled points are presented in \figref{fig:car_parking_robustness}. %for both cases. 
The results show that, in all cases, the region-based belief values are greater than or equal to the particle-based values, and therefore the region-based approach is more robust to disturbance (i.e., generates lower bound values closer to the optimum).

As the number of reachable states for a given number of transitions from an initial particle-based belief is finite, we also compare the robustness of values obtained with our particle-based NS-HSVI and the finite-state POMDP solver SARSOP, for the $4 \times 4$ dynamic vehicle parking without obstacles in \figref{fig:car_parking_sasop_robustness}. For an initial particle-based belief, we build two finite-state POMDPs by unrolling the model's execution when considering $4$ and $6$ transitions, respectively. Note that no new distinct states can be reached for paths whose length exceeds $6$ in this example, as any cell in the grid can be reached from any other cell in $6$ steps. Then, we compute the value function, represented as a set of $\alpha$-vectors, for each finite-state POMDP with SARSOP. Using the value function, we approximate the values of beliefs disturbed by shifting as above, in which each shifted particle takes the value of the closest point in the finite-state space of the unrolled POMDP. The optimal value of each shifted belief is computed by unrolling from the shifted belief for a maximum of $6$ transitions and solving the resulting finite-state POMDP with SARSOP.

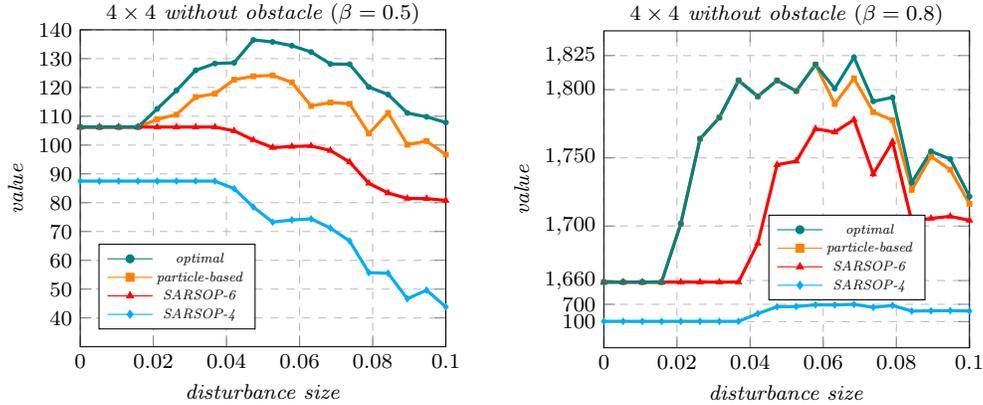
\begin{figure}[t]
\centering
\begin{subfigure}{0.47\textwidth}
%robustness_data_points = "dp, avgp, maxp, minp\n"
%robustness_data_sarsop = "dr, avgr, maxr, minr\n"
%robustness_data_optimal = "dr, avgr, maxr, minr\n"
\newcommand{\colorp}{orange}
\newcommand{\colorseight}{red}
\newcommand{\colorsfour}{cyan}
\newcommand{\coloro}{teal}
\scriptsize{
\hspace{-0.0cm}
\begin{tikzpicture}
\begin{axis}[
    %no markers,
    title style={yshift=-2ex},
    title={\emph{$4 \times 4$  without obstacle $(\beta = 0.5)$}},
    ylabel={\emph{value}},
    xlabel={\emph{disturbance size}},
    xmin=0, xmax=0.10, %34
    xtick={0,0.02,0.04,0.06,0.08,0.1}, %34
    ymin=30, ymax=140,
    ytick={40, 50, 60,70, 80, 90, 100, 110, 120, 130, 140},
    xmajorgrids=true,
        xticklabel style={
        /pgf/number format/fixed,
        /pgf/number format/precision=3
    },
    scaled x ticks=false,
    grid style=dashed,
    grid=both,
    height=0.9\textwidth,
    width=\textwidth,
    legend entries={
                {\emph{optimal}},
                {\emph{particle-based}},
                {\emph{SARSOP-6}},
                {\emph{SARSOP-4}}
                },
    legend style={at={(0.05,0.325)},
                anchor=north west, 
                nodes={scale=0.65, transform shape}}            
]
\addlegendimage{mark=*,teal,\coloro,mark size=1.5pt}
\addlegendimage{mark=square*,\colorp,mark size=1.5pt}
\addlegendimage{mark=triangle*,\colorseight,mark size=1.5pt}
\addlegendimage{mark=diamond*,\colorsfour,mark size=1.5pt}
]

% \addplot[name path=minp, mark=,\colorp,opacity=0.5,mark size=0.5pt] table [x=dp, y=minp, col sep=comma]{figures/parking/robustness_parking_data_points.csv};

\addplot[mark=square*,very thick,\colorp,opacity=1.0,mark size=0.5pt] table [x=dp, y=avgp, col sep=comma]{figures/parking/s-robustness_data_points_0.5.csv};

% \addplot[name path=maxp, mark=,\colorp,opacity=0.5,mark size=0.5pt] table [x=dp, y=maxp, col sep=comma]{figures/parking/robustness_parking_data_points.csv};

% \addplot[opacity=.40,\colorp!20] fill between[of=minp and maxp];

%%

% \addplot[name path=minr, mark=,\colorr,opacity=0.5,mark size=0.5pt] table [x=dr, y=minr, col sep=comma]{figures/parking/robustness_parking_data_regions.csv};

\addplot[mark=triangle*,very thick,\colorseight,opacity=1.0,mark size=0.5pt] table [x=ds, y=avgs, col sep=comma]{figures/parking/s-robustness_data_sarsop_8_0.5.csv};

\addplot[mark=diamond,very thick,\colorsfour,opacity=1.0,mark size=0.5pt] table [x=ds, y=avgs, col sep=comma]{figures/parking/s-robustness_data_sarsop_4_0.5.csv};

% \pgfplotstabletypeset[col sep=comma]{figures/parking/s-robustness_data_sarsop_4.csv}

\addplot[mark=*,teal,very thick,\coloro,opacity=1.0,mark size=0.5pt] table [x=dr, y=avgr, col sep=comma]{figures/parking/s-robustness_data_optimal_0.5.csv};

% \addplot[name path=maxr, mark=,\colorr,opacity=0.5,mark size=0.5pt] table [x=dr, y=maxr, col sep=comma]{figures/parking/robustness_parking_data_regions.csv};

% \addplot[opacity=.40,\colorr!20] fill between[of=minr and maxr];

\end{axis}
\end{tikzpicture}
}
\end{subfigure}
\hfil
\raisebox{-0.06\height}{
\begin{subfigure}{0.47\textwidth}
\newcommand{\colorp}{orange}
\newcommand{\colorseight}{red}
\newcommand{\colorsfour}{cyan}
\newcommand{\coloro}{teal}

\scriptsize{
\begin{tikzpicture} 
\def\yBarrier{1600}
\def\amp{6.5}
\begin{axis}[
    title style={yshift=-2ex},
    title={\emph{$4 \times 4$ without obstacle $(\beta = 0.8)$}},
    ylabel={\emph{value}},
    xlabel={\emph{disturbance size}},
    y coord trafo/.code={%
        \pgfkeys{/pgf/fpu}%
        \pgfmathparse{#1 > \yBarrier ? #1 : (\yBarrier + \amp*ln(#1))}%
    },%
    xmin=0, xmax=0.10, %
    ytick={100,700,1660,1700,1750,1800,1825},
    xtick={0,0.02,0.04,0.06,0.08,0.1}, %
    yticklabel={%
        \pgfkeys{/pgf/fpu}%
        \pgfmathparse{\tick - 50 > \yBarrier ? \tick : 
            (\tick - \yBarrier)/\amp  
        }%
        \ifpgfmathfloatcomparison
            \pgfmathprintnumber\tick
        \else
           \pgfmathparse{e^\pgfmathresult}
           \pgfmathparse{round{\pgfmathresult}}
           \pgfmathprintnumber\pgfmathresult
        \fi
        \pgfkeys{/pgf/fpu=false}%
    },
    xmajorgrids=true,
    xticklabel style={
        /pgf/number format/fixed,
        /pgf/number format/precision=3
    },
    scaled x ticks=false,
    grid style=dashed,
    grid=both,
    height=0.9\textwidth,
    width=\textwidth,
    legend entries={
                {\emph{optimal}},
                {\emph{particle-based}},
                {\emph{SARSOP-6}},
                {\emph{SARSOP-4}}
                },
    legend style={at={(0.45,0.42)},
                anchor=north west, 
                nodes={scale=0.65, transform shape}}    
]
\addlegendimage{mark=*,teal,\coloro,mark size=1.5pt}
\addlegendimage{mark=square*,\colorp,mark size=1.5pt}
\addlegendimage{mark=triangle*,\colorseight,mark size=1.5pt}
\addlegendimage{mark=diamond*,\colorsfour,mark size=1.5pt}
]

\addplot[\colorp,mark=square*,very thick,mark size=0.5pt] table [x index=0,y index=1] {1.csv};
\addplot[\colorseight,mark=triangle*,very thick,mark size=0.5pt] table [x index=0,y index=1] {2.csv};
\addplot[\colorsfour,mark=diamond,very thick,mark size=0.5pt] table [x index=0,y index=1] {3.csv};
\addplot[\coloro,mark=*,very thick,mark size=0.5pt] table [x index=0,y index=1] {4.csv};
\end{axis}
\end{tikzpicture}
}
\end{subfigure}
}
\vspace*{-0.5cm}
\caption{Comparison between particle-based and SARSOP values.}
\label{fig:car_parking_sasop_robustness}
\end{figure}

SARSOP performs better with respect to the computational time taken, which is understandable as SARSOP takes as input a discretised version of the model and does not operate over a continuous abstraction, as NS-HSVI does, requiring expensive operations over polyhedra. Nevertheless, the results shown in \figref{fig:car_parking_sasop_robustness} demonstrate that the values achieved by strategies generated using SARSOP highly depend on how much of the model's execution we are able to construct beforehand, as the impact of missing reward-critical states with a shorter horizon can be considerable. It also shows that particle-based NS-HSVI obtains greater or equal lower bound values compared to SARSOP within a small disturbance range. 
This is due to the fact that, when performing the ISPP backup, we update not only the values for the visited points but also for the regions that contain them. The optimal values of the shifted beliefs indicate that the values of  the particle-based NS-HSVI and SARSOP are both valid lower bounds.
\section{Related work}\label{rw-sect}

Continuous-state POMDPs arise in a multitude of applications, including robotics and autonomous systems. 
Many approaches have been proposed to solve such models, as well as the induced fully-observable belief MDPs, including point-based value iteration \cite{JMP-NV-MTS-PP:06,LB-IL-NRA:19,ZZ-SS-PP-KK:12}, simulation-based policy iteration \cite{XJ-JT-XT-HX:18}, discrete space approximation \cite{SB-TG-RD:13}, locally-valid approximation \cite{JVDB-SP-RA:12} and tree search planning \cite{MHL-CJT-ZNS:21}. However, these approaches focus on traditional symbolic systems and, while extended to continuous transitions via sampling~\cite{JMP-NV-MTS-PP:06}, they are not adapted to data-driven perception functions.

A common approach to solving continuous-state MDPs is by discretising and reduction to finite-state problems, which suffers from exponential growth in the state space size and time horizon. 
An alternative is to work directly with the continuous state space.
The point-based methods of~\cite{JMP-NV-MTS-PP:06,LB-IL-NRA:19,ZZ-SS-PP-KK:12} use $\alpha$-functions, similarly to our approach, but represent value functions as Gaussian mixtures or dynamic Bayes nets, which may result in looser approximation for NNs than our polyhedral representation.
This is because our P-PWLC representation exploits the underlying piecewise constant structure of the continuous-state model and the neural perception mechanism (for which the value function may not be piecewise constant).

Rather than approximating value functions with parametric functions, such as Gaussians, we instead exploit structure in the underlying model
and work directly with the continuous state space, similarly to~\cite{ZF-RD-NM-RW:04}. With suitable restrictions, one can ensure the existence of 
a finite piecewise constant representation of the value function,
based on a partition of the state space created dynamically during solution.
While the existence of a finite representation permits a VI approach, it is rarely scalable. A more efficient (approximate) solution method is HSVI, 
a point-based value iteration for finite-state POMDPs \cite{TS-RS:04,TS-RS:05}, which %we adapt for the first time to the continuous-state POMDPs. HSVI \cite{TS-RS:04,TS-RS:05} 
uses  effective heuristics to guide the forward exploration towards beliefs that significantly reduce the gap between the upper and lower bounds on the optimal value function. 

Alternative solution methods include 
PBVI \cite{JP-GG-ST:03}, the first point-based algorithm to demonstrate good performance on large POMDPs. FSVI \cite{GS-RIB-SES:07}, also a point-based value iteration method, explores the belief space by maintaining the true states, using the optimal value function of the underlying MDP to decide which action to take and then sampling the next states and observations. SARSOP \cite{HK-DH-WL:08}, one of the fastest point-based algorithms, first approximates the optimally reachable belief space in each iteration by sampling a belief according to its stored lower and upper bound functions, then performs backups at selected nodes in the belief tree and prunes the $\alpha$-vectors that are dominated by others over a neighbourhood of the belief tree. 

\iffalse
%either
HSVI is a point-based value iteration for finite-state POMDPs \cite{TS-RS:04,TS-RS:05},
\daveM{could add \url{https://arxiv.org/abs/2406.02871} (UAI24) but don't think it is needed since we don't do reach/verif}
which was recently extended to stochastic games  \cite{KH-BB-VK-CK:23} and works in the continuous belief space, but, to the best our knowledge, has not been applied to continuous-state POMDPs. 

\martaM{revise as in FM}Approaches based on discretisation suffer from loss of accuracy and exponential growth in the number of states and the finite horizon.

While our VI and NS-HSVI algorithms work directly in the continuous state space of the POMDP, most existing approaches rely on constructing a finite-state POMDP to approximate the continuous-state POMDP and then solving the finite-state model. 
\fi

{\revise 
Modelling formalisms and rigorous verification methodologies for systems that incorporate data-driven components are beginning to emerge. This includes verification for non-stochastic neuro-symbolic systems~\cite{MEA-EB-PK-AL:20}, risk verification of the closed-loop stochastic systems for data-driven controllers~\cite{MC-LL-RI-GJP:22}, perception contracts~\cite{Mitra-10.1145/3622875}, probabilistic abstractions for perception modules~\cite{DBLP:journals/tse/CalinescuIMRPSV24}, and synthesis of verified NN-based POMDP policies~\cite{CJT20}. 
This paper is part of an effort to develop effective {\em optimal} policy synthesis for stochastic neuro-symbolic multi-agent systems with neural perception mechanisms, for which solutions for simpler fully-observable variants have been proposed~\cite{YSD+22,nscsgs,kwiatkowska_et_al:LIPIcs.MFCS.2022.4}.
To the best of our knowledge, our approach is the first 
%approximate 
value computation method for partially observable continuous-state POMDPs with neural perception.
Since this paper was submitted, a point-based value iteration algorithm was presented for the more general setting of one-sided neuro-symbolic partially observable stochastic games (NS-POSGs) in~\cite{RY-GS-GN-DP-MK:23-2}, which allows partial observability under similar assumptions to those used here (in fact, our model is a single-agent variant of~\cite{RY-GS-GN-DP-MK:23-2} and shares it syntax). 
The methods of~\cite{RY-GS-GN-DP-MK:23-2} exploit a finite representation that generalises $\alpha$-vectors to approximate value computation, for which \emph{online} strategy synthesis methods have also since been developed~\cite{YSN+24}.
}

\section{Conclusions}\label{sec:conclusions}
We have introduced NS-POMDPs, the first partially observable neuro-symbolic model for an agent operating in continuous state space and perceiving the environment using a {\revise neural perception mechanism}. %NNs.
Motivated by the need for safety guarantees for such systems, we focus on {\em optimal} policy synthesis with discounting.
By placing mild assumptions on %the structure of 
NS-POMDPs, we are able to exploit their structure to approximate the value function from below and above using a representation of PWC $\alpha$-functions and belief-value induced functions. Using NS-HSVI, a variant of the classical HSVI algorithm, we synthesised optimal strategies for an agent parking a car and safe strategies for an agent using an aircraft collision avoidance system, employing the popular particle-based and novel region-based beliefs. 

Our main achievement is demonstrating the practicality of the methodology for small systems with realistic neural network components. 
To make progress in this challenging problem domain,
similarly to other POMDP approaches, we initially focus on discounted objectives,
and aim to later extend to the more complex undiscounted case (which is already
undecidable for finite-state POMDPs). However, as the case studies demonstrate,
we can use our approach to synthesise strategies that can then be shown to
be safe in terms of provably avoiding ``unsafe'' parts of the state space.
Further work includes efficiency improvement by incorporating sampling, approximating preimage of NNs, and adapting NS-HSVI to more general {\revise probabilistic} neural perception mechanisms. 

\startpara{Acknowledgements}
This project was funded by the ERC under the European
Union’s Horizon 2020 research and innovation programme
(grant agreement No.834115, FUN2MODEL).
{\revise We also gratefully acknowledge various helpful suggestions from the anonymous reviewers of this paper.}

\newpage
\appendix 
\setcounter{lema}{0}
\setcounter{thom}{0}
\section{Proofs from Section~\ref{sec:value-iteration}}

\noindent
Before we give the proofs of Section~\ref{sec:value-iteration} we require the following definition.

\begin{defi}
For FCPs $\Phi_1$ and $\Phi_2$ of S, we denote by $\Phi_1+\Phi_2$ the smallest FCP of $S$ such that $\Phi_1+\Phi_2$ is a refinement of both $\Phi_1$ and $\Phi_2$, which can be computed by all combinations of intersections between regions in $\Phi_1$ and $\Phi_2$.
\end{defi}

\begin{lema}[Perception FCP]
    There exists a smallest FCP of $S$, called the perception FCP, denoted $\Phi_{P}$, such that all states in any $\phi \in \Phi_{P}$ are observationally equivalent, i.e., if $(s_A,s_E),(s_A',s_E')\in \phi$, then $s_A=s_A'$ and we let $s_A^\phi= s_A$.
\end{lema}
\begin{proof}
   Since $\obs_A$ is PWC and $S_A$ is finite, using \defiref{defi:NS-POMDP} we have that for any $s_A = (\loc, \per) \in S_A$ the set $S_E^{s_A} = \{ s_E \in S_E \mid \obs_A(\loc, s_E) = \per \}$ can be expressed as a number of disjoint regions of $S_E$ and we let $\Phi_E^{s_A}$ be such a representation that minimises the number of such regions. It then follows that $\{ \{ (s_A, s_E) \mid s_E \in \phi_E \} \mid \phi_E \in \Phi_E^{s_A} \wedge s_A \in S_A \}$ is a smallest FCP of $S$ such that all states in any region are observationally equivalent. 
\end{proof}

% theorem 1
\begin{thom}[P-PWLC closure and convergence]\label{thom:PWC-consistency-app}
If $V \in \mathbb{F}(S_B)$ and P-PWLC, then so is $[TV]$. If $V^0 \in \mathbb{F}(S_B)$ and P-PWLC, then the sequence $(V^t)_{t = 0}^{\infty}$, such that $V^{t+1} = [TV^t]$, is P-PWLC and converges to $V^\star$.
\end{thom}

\begin{proof}
Consider any $V \in \mathbb{F}(S_B)$ that is P-PWLC, by \defiref{defi:PWLC} there exists a finite set $\Gamma \subseteq \mathbb{F}_{C}(S)$ such that:
\begin{equation}\label{eq:V-max-expression-app}
    V(s_A,b_E) = \max\nolimits_{\alpha \in \Gamma}\langle \alpha, (s_A,b_E) \rangle \; \mbox{for all $(s_A,b_E) \in S_B$.}
\end{equation}
Now consider any $(s_A, b_E), (s_A', b_E') \in S_B$ where $s_A' = (\loc', \per')$ and action $a  \in \Delta_A(s_A)$, and letting $P_1 \coloneqq P(s_A' \mid (s_A,b_E), a)$, by \eqnref{eq:V-max-expression-app} we have:
\begin{align}\
   \lefteqn{\!\!\!\!\!\!\! V (s_A', b_E^{s_A,a,s_A'}) \; = \; \max_{\alpha \in \Gamma}\langle \alpha, (s_A',b_E^{s_A,a,s_A'}) \rangle} \nonumber \\
   & \!\!\!\!\!\!\!\!\!\! = \max_{\alpha \in \Gamma} \int_{s_E' \in S_E} \alpha(s_A',s_E')b_E^{s_A,a,s_A'}(s_E') \textup{d} s_E' & \mbox{by \eqnref{expectation-eq}}  \nonumber\\
    & \!\!\!\!\!\!\!\!\!\! = \max_{\alpha \in \Gamma} \int_{s_E' \in S_E} \alpha(s_A',s_E') \frac{P((s_A',s_E') \mid (s_A,b_E), a)}{P(s_A' \mid (s_A,b_E), a)} \textup{d} s_E' & \mbox{by \eqnref{eq:belief-update1}} \nonumber \\
       & \!\!\!\!\!\!\!\!\!\! = \max_{\alpha \in \Gamma} \int_{s_E' \in S_E} \alpha(s_A',s_E') \frac{P((s_A',s_E') \mid (s_A,b_E), a)}{P_1} \textup{d} s_E'  & \mbox{by definition of $P_1$} \nonumber \\
   & \!\!\!\!\!\!\!\!\!\! \; \lefteqn{=\frac{1}{P_1} \max_{\alpha \in \Gamma} \int_{s_E' \in S_E} \alpha(s_A',s_E') P((s_A',s_E') \mid (s_A,b_E), a)  \textup{d} s_E'} & \mbox{rearranging} \nonumber \\
    & \!\!\!\!\!\!\!\!\!\! \lefteqn{\;=\frac{1}{P_1} \max_{\alpha \in \Gamma} \int_{s_E' \in S_E} \!\!\!\!\!\!\!\! \alpha(s_A',s_E') \delta_A(s_A, a)(\loc')  \left(\int_{s_E' \in S_E^{s_A'} \wedge s_E \in S_E} \!\!\!\!\!\!\!\!\!\!\!\!\!\!\!\!\!\!\!\!\!\! b_E(s_E) \delta_E(s_E,a)(s_E')  \textup{d} s_E\right) \textup{d} s_E'} \nonumber \\
    && \mbox{by \eqnref{eq:new-agent-obs}} \nonumber \\
    & \!\!\!\!\!\!\!\!\!\! \lefteqn{\;= \frac{1}{P_1}\max_{\alpha \in \Gamma} \int_{s_E \in E} \left( \delta_A(s_A, a)(\loc') \int_{ s_E' \in S_E^{s_A'}}   \!\! \alpha(s_A',s_E')  \delta_E(s_E,a)(s_E')  \textup{d} s_E'  \right) b_E(s_E) \textup{d} s_E} \nonumber \\
    &&\mbox{rearranging.} \label{eq:V-updated-belief-app}
\end{align}
Next, for any $\alpha \in \mathbb{F}_C(S)$, $s_A' \in S_A$ and $a \in \Act$ we let $\alpha^{a, s_A'} : S \rightarrow \mathbb{R}$ be the function where for any $s = (s_A,s_E) \in S$ if $a \in \Delta_A(s_A)$, then: 
\begin{align}
\lefteqn{\alpha^{a, s_A'}(s_A,s_E)  = \delta_A(s_A, a)(\loc') \left( \int_{ s_E' \in S_E^{s_{\scale{.75}{A}}'}}  \alpha(s_A',s_E')  \delta_E(s_E,a)(s_E')  \textup{d} s_E' \right)}
\nonumber \\
& \; = \left( \int_{ s_E' \in S_E^{s_{\scale{.75}{A}}'}}  \alpha(s_A',s_E')  \delta(s,a)(s_A', s_E')  \textup{d} s_E' \right) & \mbox{by \defiref{semantics-def}} 
\label{eq:alpha-new-middle}
\end{align}
and otherwise $\alpha^{a, s_A'}(s_A,s_E) = L$. 
% where $\Theta_{s_{\scale{.75}{E}}}^a = \{ s_E' \in S_E \mid \delta_E(s_E, a)(s_E') > 0 \}$.  
Now combining \eqref{eq:V-updated-belief-app} and \eqnref{eq:alpha-new-middle} we have:
\begin{equation}\label{eq:V-updated-belief-app2}
\!\!\! V(s_A', b_E^{s_A,a,s_A'}) =  \frac{1}{P(s_A' \mid (s_A,b_E), a)}\max_{\alpha \in \Gamma} \int_{s_E \in E} \alpha^{a, s_A'}(s_A,s_E) b_E(s_E) \textup{d} s_E \, .
\end{equation}
We next prove that $\alpha^{a, s_A'}$ is PWC, i.e., $\alpha^{a, s_A'} \in \mathbb{F}_{C}(S)$. Since $\alpha \in \mathbb{F}_C(S)$, there exists an FCP $\Phi'$ of $S$ such that $\alpha$ is constant in each region of $\Phi'$. {\revise According to \aspref{asp:transitions-rewards}, there exists a preimage FCP $\Phi$ of $S$ and a transition function $\delta_{\Phi}$ such that, for all $s \in \phi \in \Phi$, $a \in \Act$ and $s' \in \phi' \in \Phi'$, $\delta(s, a)(s') = \delta_{\Phi}(\phi, a)(\phi')$. Therefore, using \eqref{eq:alpha-new-middle}, we obtain that $\alpha^{a, s_A'}$ is constant in each region of $\Phi$, i.e., $\alpha^{a, s_A'}$ is PWC.}

Substituting \eqref{eq:V-updated-belief-app2} into \defiref{max-def} it follows that $[TV](s_A, b_E)$ equals:
\begin{align*}
& \lefteqn{\max_{a  \in \Delta_A(s_A)} \left\{ \langle R_{a}, (s_A,b_E) \rangle + \beta \mbox{$\sum\nolimits_{s_A' \in S_A}$}  \max_{\alpha \in \Gamma} \int_{s_E \in E} \alpha^{a, s_A'}(s_A,s_E) b_E(s_E) \textup{d} s_E \right\}}
\\
& = \; \max_{a  \in \Delta_A(s_A)} \left\{ \langle R_{a}, (s_A,b_E) \rangle + \beta \mbox{$\sum\nolimits_{s_A' \in S_A}$}  \max_{\alpha \in \Gamma} \langle \alpha^{a, s_A'},  (s_A,b_E) \rangle  \right\} & \mbox{by \eqnref{expectation-eq}.}
\end{align*}
Therefore letting $\Gamma^{a, s_A'} = \{ \alpha^{a, s_A'} \mid \alpha \in \Gamma \}$ and
\[
\alpha_{b_E}^{a, s_A'} \in \arg \max\nolimits_{\alpha^{a, s_{\scale{.75}{A}}'} \in \Gamma^{a, s_{\scale{.75}{A}}'}} \langle \alpha^{a, s_A'},  (s_A,b_E) \rangle
\]
where $\alpha_{b_E}^{a, s_A'}$ is independent of $s_A$ due to \eqref{eq:alpha-new-middle}, it then follows from \defiref{max-def} that:
\begin{align}
[TV](s_A, b_E) = & \; \lefteqn{\max_{a  \in \Delta_A(s_A)} \left\{ \langle R_{a}, (s_A,b_E) \rangle + \beta \mbox{$\sum_{s_A' \in S_A}$}  \langle \alpha^{a, s_A'}_{b_E},  (s_A,b_E) \rangle  \right\}}\nonumber \\
    = & \max_{a  \in \Delta_A(s_A)} \left\langle R_{a} + \beta \mbox{$\sum_{s_A' \in S_A}$}  \alpha^{a, s_A'}_{b_E}, (s_A,b_E) \right\rangle & \mbox{by \eqnref{expectation-eq}.} \label{expectation1-eqn}
\end{align}
Furthermore, we have that
\[
\Gamma_{(s_A, b_E)} = \left\{ R_{a} + \beta \mbox{$\sum_{s_A' \in S_A}$} \alpha^{a, s_A'}_{b_E} \mid a  \in \Delta_A(s_A) \right\} 
\]
and from \eqnref{expectation1-eqn}:
\[
[TV](s_A, b_E) = \max\nolimits_{\alpha \in \Gamma_{(s_{\scale{.75}{A}}, b_{\scale{.75}{E}})}} \langle \alpha, (s_A,b_E) \rangle \, .
\]
Finally, since $S_A$, $\Act$ and $\Gamma$ are finite, $R_{a}$ is PWC by \aspref{asp:transitions-rewards} and $\alpha^{a, s_A'}$ is PWC, defining $\Gamma'$ to be the set containing the PWC functions:
\[
R_{a} + \beta \mbox{$\sum_{s_A' \in S_A}$}  \alpha^{a, s_A'} \in  \mathbb{F}_C(S)
\]
for all $a  \in \Act$, $s_A' \in S_A$ and $ \alpha^{a, s_A'} \in \Gamma^{a, s_A'}$,
we have for any $(s_A,b_E) \in S_B$:
\begin{equation*}\label{eq:new-func-T}
[TV](s_A, b_E) = \max\nolimits_{\alpha \in \Gamma'} \langle \alpha, (s_A,b_E) \rangle.
\end{equation*}
Therefore, $[TV]$ is P-PWLC. As can be seen $|\Gamma'| = |\Act| |\Gamma|^{|S_A|}$, and hence the number of $\alpha$-functions representing $[TV]$ given grows exponentially in the number of agent states for those representing $V$.

The remainder of the proof follows from Banach's fixed point theorem and the fact we have proved that if $V \in \mathbb{F}(S_B)$ and P-PWLC, so is $[TV]$. 
\end{proof}

\begin{thom}[Convexity and continuity]\label{thom:continuity-app}
For any $s_A \in S_A$, the value function $V^{\star}(s_A, \cdot) : \mathbb{P}(S_E) \rightarrow \mathbb{R}$ is convex and for any $b_E, b_E' \in  \mathbb{P}(S_E)$:
\begin{equation}\label{eq:K-definition2}
    |V^{\star}(s_A, b_E) - V^{\star}(s_A, b_E')| \leq K(b_E, b_E')
\end{equation}
where $K(b_E, b_E') = \frac{1}{2}(U - L) \int_{s_E \in S_E^{s_{\scale{.75}{A}}}} | b_E(s_E) - b_E'(s_E)| \textup{d}s_E$.
\end{thom}
\begin{proof}
According to Theorem~\ref{thom:PWC-consistency} there exists a (possibly infinite) set $\Gamma \subseteq \mathbb{F}_C(S)$ such that for any $(s_A, b_E) \in S_B$:
\begin{equation}\label{eq:V-star-sup}
    V^{\star}(s_A, b_E) = \sup\nolimits_{\alpha \in \Gamma} \langle \alpha, (s_A, b_E) \rangle \, .
\end{equation}
Given $s_A \in S_A$, consider any $b_E, b_E' \in \mathbb{P}(S_E)$ and $\lambda \in [0, 1]$, and we have:
\begin{align*}
    \lefteqn{\lambda V^{\star}(s_A, b_E) + (1 - \lambda ) V^{\star}(s_A, b_E')} \\
    & = \; \lambda \sup\nolimits_{\alpha \in \Gamma} \langle \alpha, (s_A, b_E) \rangle + (1 - \lambda ) \sup\nolimits_{\alpha \in \Gamma} \langle \alpha, (s_A, b_E') \rangle & \mbox{by \eqnref{eq:V-star-sup}} \\
    & = \; \sup\nolimits_{\alpha \in \Gamma} \langle \alpha, (s_A, \lambda  b_E) \rangle +  \sup\nolimits_{\alpha \in \Gamma} \langle \alpha, (s_A, (1 - \lambda ) b_E') \rangle & \mbox{by \eqnref{expectation-eq}} \\
    & \ge \; \sup\nolimits_{\alpha \in \Gamma} \langle \alpha, (s_A, \lambda  b_E + (1 - \lambda ) b_E') \rangle & \mbox{rearranging} \\
    & = \; V^{\star}(s_A, \lambda  b_E + (1 - \lambda ) b_E') & \mbox{by \eqnref{eq:V-star-sup}}
\end{align*}
which proves that $V^{\star} (s_A, \cdot)$ is convex.

Next given $\alpha$ and $s_A$, let $V_{\alpha, s_A} (b_E) \coloneqq \langle \alpha, (s_A, b_E) \rangle$ for $(s_A, b_E) \in S_B$. For any $(s_A, b_E), (s_A, b_E') \in S_B$, without loss of generality, we can assume that $V_{\alpha, s_A}(b_E) \ge V_{\alpha, s_A} (b_E')$, and therefore:
\begin{align}
 & |V_{\alpha, s_A}(b_E) - V_{\alpha, s_A} (b_E')| = V_{\alpha, s_A}(b_E) - V_{\alpha, s_A} (b_E') \nonumber \\
 & = \langle \alpha, (s_A,b_E) \rangle - \langle  \alpha (s_A, b_E') \rangle   & \mbox{by definition of $V_{\alpha, s_A}$} \nonumber \\
 & \; \lefteqn{= \int_{s_E \in S_E^{s_A}} \alpha(s_A, s_E) b_E(s_E) \textup{d}s_E - \int_{s_E \in S_E^{s_A}}  \alpha (s_A, s_E) b'_E(s_E) \textup{d}s_E}  & \mbox{by \eqnref{expectation-eq}} \nonumber \\
 & = \int_{s_E \in S_E^{s_A}} \alpha (s_A, s_E) (b_E(s_E) - b_E'(s_E)) \textup{d}s_E & \mbox{rearranging.} \label{lip3-eqn}
\end{align}
Since $b_E, b_E' \in \mathbb{P}(S_E)$ and $(s_A, b_E), (s_A, b_E') \in S_B$, we have: 
\begin{equation}\label{lip1-eqn}
\int_{s_E \in S_E^{s_A}} b_E(s_E) \textup{d} s_E = \int_{s_E \in S_E^{s_A}} b_E'(s_E) \textup{d} s_E = 1 \, .
\end{equation}
Now, letting $S_E^{+} = \{ s_E \in S_E^{s_A} \mid b_E(s_E) - b_E'(s_E) > 0 \}$ and $S_E^{-} = \{ s_E \in S_E^{s_A} \mid b_E(s_E) - b_E'(s_E) \leq 0 \}$, rearranging \eqnref{lip1-eqn} and using the fact that $S_E^{+} \cup S_E^{-} =S_E^{s_A}$ it follows that:
\begin{align}
\int_{s_E \in S_E^{-}} (b_E(s_E) - b_E'(s_E)) \textup{d} s_E = - \int_{s_E \in S_E^{+}} (b_E(s_E) - b_E'(s_E)) \textup{d} s_E. \label{lip-eqn}
\end{align}
Next, using \eqnref{lip3-eqn}, the definition of $V_{\alpha, s_A}$ and \eqnref{expectation-eq}, it follows that $|V_{\alpha, s_A}(b_E) - V_{\alpha, s_A}(b_E')|$ equals:
\begin{align}
  \lefteqn{\!\!\!\!\! \int_{s_E \in S_E^+}\!\! \alpha(s_A, s_E) (b_E(s_E) - b_E'(s_E)) \textup{d}s_E + \int_{s_E \in S_E^-}\!\! \alpha (s_A, s_E) (b_E(s_E) - b_E'(s_E)) \textup{d}s_E} \nonumber \\
  & \; \lefteqn{\leq \int_{s_E \in S_E^+} U (b_E(s_E) - b_E'(s_E)) \textup{d}s_E +  \int_{s_E \in S_E^-}  L (b_E(s_E) - b_E'(s_E)) \textup{d}s_E} \nonumber \\
  & & \mbox{by definition of $S_E^+$, $S_E^-$, $U$ and $L$} \nonumber \\
  & \; \lefteqn{=  U \int_{s_E \in S_E^+}  (b_E(s_E) - b_E'(s_E)) \textup{d}s_E - L \int_{s_E \in S_E^+}  (b_E(s_E) - b_E'(s_E)) \textup{d}s_E} \nonumber \\
  & & \mbox{rearranging and using \eqnref{lip-eqn}} \nonumber \\
  & = k \int_{s_E \in S_E^+}  (b_E(s_E) - b_E'(s_E)) \textup{d}s_E  & \mbox{rearranging and letting $k = U - L$.} \label{lip2-eqn}
\end{align}
We can now derive the following upper bound for $V^{\star} (s_A, b_E)$:
\begin{align*}
\lefteqn{V^{\star} (s_A, b_E)  = \sup\nolimits_{\alpha \in \Gamma} \langle \alpha, (s_A, b_E) \rangle} && \mbox{by \eqref{eq:V-star-sup}} \\
& = \sup_{\alpha \in \Gamma}  V_{\alpha, s_A} (b_E) & \mbox{by definition of $ V_{\alpha, s_A}$} \\
& = \sup_{\alpha \in \Gamma} ( V_{\alpha, s_A} (b_E') + V_{\alpha, s_A} (b_E) - V_{\alpha, s_A} (b_E') ) & \mbox{rearranging} \\
& = \sup_{\alpha \in \Gamma} ( V_{\alpha, s_A} (b_E') + | V_{\alpha, s_A} (b_E) - V_{\alpha, s_A} (b_E')  | ) & \mbox{rearranging} \\
& \leq \sup_{\alpha \in \Gamma} \left\{ V_{\alpha, s_A} (b'_E) + k \int_{s_E \in S_E^+} \!\!\!\!\!\!  (b_E(s_E) - b_E'(s_E)) \textup{d}s_E \right\} & \mbox{by \eqnref{lip2-eqn}} \\
& = \sup_{\alpha \in \Gamma} ( V_{\alpha, s_A} (b'_E) ) + k \int_{s_E \in S_E^+} \!\!\!\!\!\!  (b_E(s_E) - b_E'(s_E)) \textup{d}s_E & \mbox{rearranging} \\
& = \sup_{\alpha \in \Gamma} \langle \alpha, (s_A,b'_E) \rangle + k \int_{s_E \in S_E^+} \!\!\!\!\!\!  (b_E(s_E) - b_E'(s_E)) \textup{d}s_E & \mbox{by definition of $V_{\alpha, s_A}$} \\
 & = V^{\star} (s_A, b_E') +  k \int_{s_E \in S_E^+} \!\!\!\!\!\!  (b_E(s_E) - b_E'(s_E)) \textup{d}s_E  & \mbox{by \eqnref{eq:V-star-sup}} \\
 & = V^{\star} (s_A, b_E') +  \frac{1}{2}k \int_{s_E \in S_E^{s_A}} \!\!  |b_E(s_E) - b_E'(s_E)| \textup{d}s_E  & \mbox{by \eqnref{lip-eqn}}.
\end{align*}    
Using similar steps we can also show:
\[
V^{\star} (s_A, b'_E) \leq V^{\star} (s_A, b_E) +  \frac{1}{2}k \int_{s_E \in S_E^{s_A}} \!\!  |b_E(s_E) - b_E'(s_E)| \textup{d}s_E
\]
and therefore the second part of the theorem follows. 
\end{proof}

\section{Proofs from Section~\ref{sec:HSVI}}

\begin{lema}[Lower bound]\label{lema:new-pwc-alpha-app}
At belief $(s_A, b_E) \in S_{B}$, the function $\alpha^{\star}$ generated by \algoref{alg:point-based-update-belief} is a PWC $\alpha$-function satisfying \eqref{eq:update-lb-condition}, $V_{\mathit{lb}}^{\Gamma} \leq V_{\mathit{lb}}^{\Gamma'} \leq V^{\star}$ and $V_{\mathit{lb}}^{\Gamma'}(s_A, b_E) \ge [TV_{\mathit{lb}}^{\Gamma}](s_A, b_E)$.
\end{lema}
\begin{proof}
    By following the proof of \thomref{thom:PWC-consistency} and how $\bar{a}$ and $\alpha^{s_A'}$ are constructed for all $s_A' \in S_A$, we can easily verify that $\alpha^{\star}$ in \algoref{alg:point-based-update-belief} is a PWC $\alpha$-functions that satisfies \eqref{eq:update-lb-condition}. 

    For any $V_1, V_2 \in \mathbb{F}(S_B)$, we use the shorthand $V_1 \leq V_2$ to denote that $V_1(\hat{s}_A, \hat{b}_E) \leq V_2(\hat{s}_A, \hat{b}_E)$ for all $(\hat{s}_A, \hat{b}_E) \in S_B$. 
    Now, in the case of the lower bound consider any $V_{\mathit{lb}}^{\Gamma}$ such that $V_{\mathit{lb}}^{\Gamma} \leq V^{\star}$. Since $\Gamma' = \Gamma \cup \{\alpha^{\star}\}$ after updating $V_{\mathit{lb}}^{\Gamma}$ at $(s_A, b_E)$ through \algoref{alg:point-based-update-belief}, for any $(\hat{s}_A,\hat{b}_E) \in S_B$:
\[         \max\nolimits_{\alpha \in \Gamma}\langle \alpha, (\hat{s}_A,\hat{b}_E) \rangle \leq \max\nolimits_{\alpha \in \Gamma'}\langle \alpha, (\hat{s}_A,\hat{b}_E) \rangle \, .
\]
Therefore combining this with the fact that $V_{\mathit{lb}}^{\Gamma}$ is a P-PWLC function
% a point-wise maximum
for the finite set $\Gamma$, see \defiref{defi:PWLC}, we have $V_{\mathit{lb}}^{\Gamma}(\hat{s}_A,\hat{b}_E) \leq V_{\mathit{lb}}^{\Gamma'}(\hat{s}_A,\hat{b}_E)$
and since $(\hat{s}_A,\hat{b}_E)$ was arbitrary $V_{\mathit{lb}}^{\Gamma} \leq V_{\mathit{lb}}^{\Gamma'}$.

Next, again using the fact $V_{\mathit{lb}}^{\Gamma}$ is a P-PWLC function
% a point-wise maximum
for the finite set $\Gamma$ we have:
\begin{align*}
V_{\mathit{lb}}^{\Gamma'}(s_A, b_E) & = \max\nolimits_{\alpha \in \Gamma'}\langle \alpha, (s_A, b_E) \rangle \\
& \ge \langle \alpha^{\star}, (s_A, b_E) \rangle & \mbox{rearranging} \\
& = [TV_{\mathit{lb}}^{\Gamma}](s_A, b_E) & \mbox{by \eqref{eq:update-lb-condition}}.
\end{align*}
In \algoref{alg:point-based-update-belief}, if the backup at line 6 is executed, then the Bellman operator is applied for some states in $\phi$ which may result in non-optimal Bellman backup for the other states in $\phi$, and if the backup at line 7 is executed, $\alpha^{\star}$ is assigned the lower bound $L$ in $\phi$. Therefore we have for any $(\hat{s}_A,\hat{b}_E) \in S_B$:
\begin{align}
\langle \alpha^{\star}, (\hat{s}_A,\hat{b}_E) \rangle 
& \leq [TV_{\mathit{lb}}^{\Gamma}](\hat{s}_A,\hat{b}_E) \nonumber \\
& \leq [TV^{\star}](\hat{s}_A,\hat{b}_E) & \mbox{since $V_{\mathit{lb}}^{\Gamma} \leq V^{\star}$} \nonumber \\
& = V^{\star} (\hat{s}_A,\hat{b}_E) & \mbox{by \thomref{thom:PWC-consistency}.} \label{eq:lb-alpha-v-1}
    \end{align}
Combining this inequality with $V_{\mathit{lb}}^{\Gamma} \leq V^{\star}$, we have $V_{\mathit{lb}}^{\Gamma'} \leq V^{\star}$ as required.
\end{proof}

\begin{lema}[ISPP backup]\label{lema:ISPP-backup-app}
    The FCP $\Phi_{\textup{product}}$ returned by \algoref{alg:ISPP-backup} is a constant-FCP of $\phi$ for $\alpha^{\star}$ and the region-by-region backup for $\alpha^*$ satisfies \eqref{eq:backup-value}.
\end{lema}

\begin{proof}
     For the PWC $\alpha$-functions in the input of \algoref{alg:ISPP-backup}, let $\Phi = \sum_{s_A' \in \bar{S}_A} \Phi_{s_A'}$, where $\Phi_{s_A'}$ is an FCP of $S$ for $\alpha^{s_A'}$.
    
    According to \aspref{asp:transitions-rewards}, there exists a preimage-FCP of $\Phi$.
    % for action $\bar{a}$. 
    Through the image, split, preimage and product operations of \algoref{alg:ISPP-backup}, all the states in any region $\phi' \in \Phi_{\textup{product}}$ have the same reward and reach the same regions of $\Phi$. Since each $\alpha$-function $\alpha^{s_A'}$ is constant over each region in $\Phi$, all states in $\phi'$ have the same backup value from $\alpha^{s_A'}$ for $s_A' \in \bar{S}_A$. This implies that $\Phi_{\textup{product}}$ is a preimage FCP of $\Phi$ for action $\bar{a}$. Since the value backup \eqref{eq:backup-value} is used for each region in $\Phi_{\textup{product}}$ and the image is from the region $\phi$, then $\Phi_{\textup{product}}$ is a constant-FCP of $\phi$ for $\alpha^{\star}$, and thus the value backup \eqref{eq:backup-value} for $\alpha^{\star}$ is achieved by considering the regions of $\Phi_{\textup{product}}$. 
\end{proof}

\begin{lema}[Upper bound]
Given belief $(s_A, b_E) \in S_{B}$, if $p^{\star} = [TV_{\mathit{ub}}^{\Upsilon}](s_A, b_E)$, then $p^\star$ is an upper bound of $V^{\star}$ at $(s_A, b_E)$, i.e., $p^{\star} \ge V^{\star} (s_A, b_E)$, and if $\Upsilon' = \Upsilon \cup \{ ((s_A, b_E), p^{\star}) \}$, then $V_{\mathit{ub}}^{\Upsilon} \ge V_{\mathit{ub}}^{\Upsilon'} \ge V^{\star}$ and $V_{\mathit{ub}}^{\Upsilon'}(s_A, b_E) \leq [TV_{\mathit{ub}}^{\Upsilon}](s_A, b_E)$.
\end{lema}
\begin{proof}
    Consider an upper bound $V_{\mathit{ub}}^{\Upsilon}$ such that $V_{\mathit{ub}}^{\Upsilon} \ge V^{\star}$. By construction, each pair $((s_A^i, b_E^i), y_i)$ in $\Upsilon$ satisfies $V^{\star}(s_A^i, b_E^i) \leq y_i$. 

     Now suppose for belief $(s_A, b_E) \in S_{B}$ we let $p^{\star} = [TV_{\mathit{ub}}^{\Upsilon}](s_A, b_E)$ and $\Upsilon' = \Upsilon \cup \{((s_A, b_E), p^{\star})\}$. The new upper bound $V_{\mathit{ub}}^{\Upsilon'}$  after updating $V_{\mathit{ub}}^{\Upsilon}$ at $(s_A, b_E)$ through \algoref{alg:point-based-update-belief}, satisfies $V_{\mathit{ub}}^{\Upsilon} \ge V_{\mathit{ub}}^{\Upsilon'}$ by \eqref{eq:new-ub}. By construction of $p^{\star}$ we have:
    \begin{align*}
        p^{\star} & = [TV_{\mathit{ub}}^{\Upsilon}] (s_A, b_E) \\
        & \ge [TV^{\star}] (s_A, b_E) & \mbox{since $V_{\mathit{ub}}^{\Upsilon} \ge V^{\star}$} \\
        & = V^{\star} (s_A, b_E) & \mbox{by \thomref{thom:PWC-consistency}.}
    \end{align*}
    Next we have:
    \begin{align*}
        V_{\mathit{ub}}^{\Upsilon'}(s_A, b_E) & \leq p^{\star}  & \mbox{since $((s_A, b_E), p^{\star}) \in \Upsilon'$ and \eqref{eq:new-ub}} \\
        & = [TV_{\mathit{ub}}^{\Upsilon}] (s_A, b_E) & \mbox{by construction of $p^{\star}$.}
    \end{align*}
It therefore remains to prove the last part, i.e.\ that $V_{\mathit{ub}}^{\Upsilon'} \ge V^{\star}$. Now for any $(s_A',b_E') \in S_B$, if $s_A' \neq s_A$, then using the  fact that $\Upsilon' = \Upsilon \cup \{ ((s_A, b_E), p^{\star}) \}$ and \eqref{eq:new-ub} we have: 
\begin{align*}
 V_{\mathit{ub}}^{\Upsilon'}(s_A',b_E') & = V_{\mathit{ub}}^{\Upsilon}(s_A',b_E') \\
 & \geq V^{\star}(s_A',b_E') & \mbox{since $V_{\mathit{ub}}^{\Upsilon} \ge V^{\star}$.}
\end{align*}
On the other hand, if $s_A' = s_A$, then using \eqref{eq:new-ub} there exists $\langle \hat{\lambda}_i \rangle_{i \in I_{s_{\scale{.75}{A}}}}$ with $\hat{\lambda}_i \ge 0$ and $\mbox{$\sum\nolimits_{i \in I_{s_{\scale{.75}{A}}}}$} \hat{\lambda}_i = 1$ such that:
\begin{align}\label{ub-eqn}
     V_{\mathit{ub}}^{\Upsilon'}(s_A',b_E') = \mbox{$\sum\nolimits_{i \in I_{s_{\scale{.75}{A}}}}$} \hat{\lambda}_i y_i  + K_{\mathit{ub}} \left(b_E', \mbox{$\sum\nolimits_{i \in I_{s_{\scale{.75}{A}}}}$} \hat{\lambda}_i b_E^i \right) \, .
\end{align}
Now using \thomref{thom:continuity} we have:
\begin{align*}
\lefteqn{ V^{\star} (s_A',b_E') \leq V^{\star} (s_A,\mbox{$\sum\nolimits_{i \in I_{s_{\scale{.75}{A}}}}$} \hat{\lambda}_i b_E^i) + K \left(b_E', \mbox{$\sum\nolimits_{i \in I_{s_{\scale{.75}{A}}}}$} \hat{\lambda}_i b_E^i \right)} \\
& \; \lefteqn{\leq \mbox{$\sum\nolimits_{i \in I_{s_{\scale{.75}{A}}}}$} \hat{\lambda}_i  V^{\star} (s_A,b_E^i) + K \left(b_E', \mbox{$\sum\nolimits_{i \in I_{s_{\scale{.75}{A}}}}$} \hat{\lambda}_i b_E^i \right)} & \mbox{since $V^\star$ is convex in $S_B$}  \\
& \; \lefteqn{\leq \mbox{$\sum\nolimits_{i \in I_{s_{\scale{.75}{A}}}}$} \hat{\lambda}_i  V^{\star} (s_A,b_E^i) + K_{\mathit{ub}} \left(b_E', \mbox{$\sum\nolimits_{i \in I_{s_{\scale{.75}{A}}}}$} \hat{\lambda}_i b_E^i \right)} & \mbox{by \eqref{eq:K-UB-condition}}  \\
& \leq \mbox{$\sum\nolimits_{i \in I_{s_{\scale{.75}{A}}}}$} \hat{\lambda}_i  y_i + K_{\mathit{ub}} \left(b_E', \mbox{$\sum\nolimits_{i \in I_{s_{\scale{.75}{A}}}}$} \hat{\lambda}_i b_E^i \right) &   \mbox{since if $i \in I_{s_{\scale{.75}{A}}}$, then $((s_A,b_E^i),y_i) \in \Upsilon$}  \\
& = V_{\mathit{ub}}^{\Upsilon'}(s_A',b_E') & \mbox{by \eqnref{ub-eqn}.} 
\end{align*}
Therefore since these are the only cases to consider for $(s_A',b_E') \in S_B$ we have  $V_{\mathit{ub}}^{\Upsilon'} \ge V^{\star}$ as required.
\end{proof}

\begin{lema}[LP for upper bound]\label{lema:pb-upper-bound-app}
The function $K_{\mathit{ub}}$ from \eqref{eq:pb-k-ub} satisfies \eqref{eq:K-UB-condition}, and for particle-based belief $(s_A, b_E)$ represented by $\{ (s_E^i, w_i) \}_{i=1}^{N_b}$, we have that $V_{\mathit{ub}}^{\Upsilon}(s_A,b_E)$ is the optimal value of the LP: 
\[    \begin{array}{rl}
    \textup{minimize:} \; & \sum_{k \in I_{s_{\scale{.75}{A}}}} \lambda_k y_k  + (U - L) N_b c \\
  \textup{subject to:} \; & c \ge | w_i - \sum_{k \in I_{s_{\scale{.75}{A}}}} \lambda_k P(s_E^i; b_E^k) |  \textup{ for } 1 \leq i \leq N_b \\
  & \lambda_k \ge 0 \textup{ for } k \in I_{s_{\scale{.75}{A}}} \; \mbox{and} \; \sum_{k \in I_{s_{\scale{.75}{A}}}} \lambda_k = 1\, .
\end{array}
\]
\end{lema}

\begin{proof}
    Consider any particle-based beliefs $(s_A,b_E)$ and $(s_A,b_E')$ where $(s_A,b_E)$ is represented by the weighted particle set $\{ (s_E^i, w_i) \}_{i=1}^{N_b}$. Let  $S_E^{b_E>b_E'} = \{ s_E \in S_E^{s_A} \mid b_E(s_E) - b_E'(s_E) > 0 \}$. Now by definition of $K(b_E, b_E')$, see \thomref{thom:continuity}, we have:
    \begin{align*}
        \lefteqn{K(b_E, b_E')  = (U - L) \int_{s_E \in S_E^{b_{\scale{.75}{E}}>b_{\scale{.75}{E}}^{\scale{.75}{'}}}} (b_E(s_E) - b_E'(s_E)) \textup{d}s_E} \\
        & = (U - L) \int_{s_E \in S_E^{b_{\scale{.75}{E}}>b_{\scale{.75}{E}}^{\scale{.75}{'}}}} | b_E(s_E) - b_E'(s_E)| \textup{d}s_E & \mbox{by definition of $S_E^{b_E>b_E'}$} \\
        & \leq (U - L) \mbox{$\sum\nolimits_{i=1}^{N_b}$} \left| P(s_E^i;b_E) - P(s_E^i;b_E') \right|
        & \mbox{by \defiref{defi:particle-based-belief}} \\
        & \leq  (U - L)  N_b \max\nolimits_{1 \leq i \leq N_b} \left|P(s_E^i; b_E) - P(s_E^i; b_E') \right| & \mbox{rearranging} \\
        & \; \lefteqn{=  (U - L)  N_b \max\nolimits_{s_E \in S_E \wedge b_E(s_E) > 0} \left| P(s_E; b_E) - P(s_E; b_E') \right|} & \mbox{rearranging} \\
        & = K_{\mathit{ub}}(b_E, b_E') & \mbox{by \eqnref{eq:pb-k-ub}.}
    \end{align*}
Furthermore, \eqref{eq:pb-k-ub} implies $K_{\mathit{ub}}(b_E, b_E) = 0$, and therefore  $K_{\mathit{ub}}$ satisfies \eqref{eq:K-UB-condition}. 

 Next suppose $\Upsilon = \{ ((s_A^k, b_E^k), y_k) \mid k \in I\}$, letting $b_E' = \sum_{k \in I_{s_{\scale{.75}{A}}}} \lambda_k b_E^k$ and $c = \max\nolimits_{s_E \in S_E \wedge b_E(s_E) > 0} |P(s_E; b_E) - P(s_E; b_E')|$, by definition of $K_{\mathit{ub}}$, see \eqref{eq:pb-k-ub}, we have:
 \begin{align*}
\hspace*{-2cm} K_{\mathit{ub}}(b_E, b_E') & \; \lefteqn{= (U - L)  N_b \max\nolimits_{s_E \in S_E \wedge b_E(s_E) > 0} |P(s_E; b_E) - P(s_E; b_E')|} \\
& =  (U - L)  N_b c & \mbox{by definition of $c$}
\end{align*}
% by definition of $c$, 
and therefore:
    \begin{align*}
        \mbox{$\sum\nolimits_{k \in I_{s_{\scale{.75}{A}}}}$}\lambda_k y_k + K_{\mathit{ub}}(b_E, b_E') = \mbox{$\sum\nolimits_{k \in I_{s_{\scale{.75}{A}}}}$}\lambda_k y_k + (U - L)  N_b c \, .
    \end{align*}
Furthermore, for any $1 \leq i \leq N_b$, by definition of $c$ and since $(s_A,b_E)$ is represented by the weighted particle set $\{ (s_E^i, w_i) \}_{i=1}^{N_b}$ we have:
    \begin{align*}
        c & \ge \left|P(s_E^i; b_E) - P(s_E^i; b_E') \right| \\
        & = \left|w_i - P(s_E^i; b_E') \right| & \mbox{since $\{ (s_E^i, w_i) \}_{i=1}^{N_b}$ represents $b_E$} \\
        & = \left|w_i - P \left(s_E^i; \mbox{$\sum_{k \in I_{s_{\scale{.75}{A}}}}$} \lambda_k b_E^k \right) \right| & \mbox{by definition of $b_E'$.}
    \end{align*}
    Therefore, the optimisation problem \eqref{eq:new-ub} can be equivalently written as the LP of \lemaref{lema:pb-upper-bound}, i.e., the optimal value is equal to $ V_{\mathit{ub}}^{\Upsilon}(s_A,b_E)$.
\end{proof}

\begin{lema}[Region-based belief closure]
    If $\delta_E^i(\cdot, a) : S_E \rightarrow \delta_E^i(S_E,a)$ is piecewise differentiable and invertible from $S_E$ to $T \subseteq S_E$, and the Jacobian determinant of the inverse function, i.e., for any $s_E' \in T$:
    \begin{align*}
        \textup{Jac}(s_E') \coloneqq \textup{det} \left( \frac{\textup{d} \delta_E^{i,-1}(s_E', a)}{ \textup{d} s_E'} \right)
    \end{align*}
    is PWC for $a \in \Act$ and $1 \leq i \leq N_e$, then region-based beliefs are closed under belief updates.
\end{lema}

\begin{proof}
    Since $\delta_E^i(\cdot, a)$ is piecewise differentiable and piecewise invertible, let $\phi_E \subseteq S_E$ be a region over which  $\delta_E^i(\cdot, a)$ is differentiable and invertible. Suppose that $X_E$ is a random variable taking values in $\phi_E$, and that $X_E$ has a continuous uniform distribution with probability density function $b_E$. Due to the differentiability and thus continuity of $\delta_E^i(\cdot, a)$, the image $\phi_E' = \{ s_E' \mid s_E' = \delta_E^i(s_E, a) \wedge s_E \in \phi_E \}$ is a region in $S_E$. Furthermore, suppose $X_E' = \delta_E^i(X_E, a)$ is a new random variable taking values in $\phi_E'$ and let $b_E'$ be the probability density function for $X_E'$ over $\phi_E'$. We next prove that $b_E'$ is a PWC uniform distribution under the given conditions.

    Let $\delta_E^{i,-1}(\cdot, a)$ be the inverse function of $\delta_E^i(\cdot, a)$ in $\phi_E$. If $\phi_1' \subseteq \phi_E'$, letting $\phi_1$ be the preimage of $\phi'_1$, then 
    \begin{align}
        P(X_E' \in \phi_1') & = P(\delta_E^i(X_E, a) \in \phi_1') & \mbox{since $X_E' = \delta_E^i(X_E, a)$} \nonumber \\
        & = \int_{s_E \in \phi_1} b_E(s_E) \textup{d} s_E & \mbox{by definition of $b_E$.} \label{region1-eqn}
    \end{align}
    Using the change of variables $s_E = \delta_E^{i, -1}(s_E', a)$ we have that: 
    \begin{align*}
        \textup{d} s_E & = \textup{det} \left( \frac{\textup{d} \delta_E^{i,-1}(s_E', a)}{ \textup{d} s_E'} \right) \textup{d} s_E' \\
        & = \textup{Jac}(s_E') \textup{d} s_E' & \mbox{by definition of the Jacobian determinant}   \end{align*}
and substituting this into \eqnref{region1-eqn} we have:
\[
P(X_E' \in \phi_1') = \int_{s_E' \in \phi_1'} b_E(\delta_E^{i,-1}(s_E', a)) \textup{Jac}(s_E') \textup{d} s_E' \, .
\]
Therefore we have that for any $s_E' \in \phi_1'$:
           \begin{align*} 
        b_E'(s_E') = b_E(\delta_E^{i,-1}(s_E', a)) \textup{Jac}(s_E')
    \end{align*}
and since $b_E(\delta_E^{i,-1}(s_E', a)) = b_E(s_E)$ for $s_E \in \phi_E$ is constant and by construction $ \textup{Jac}(s_E')$ is PWC, we have that $b_E'$ is PWC over $\phi_E'$ as required.

We conclude that $\delta_E^i(\cdot, a)$ transforms a random variable which has a continuous uniform distribution in a region into a new random variable which has a continuous uniform distribution over finitely many regions. Therefore, region-based belief are closed under $\delta_E^i(\cdot, a)$.
\end{proof}

\begin{lema}[Region-based belief update] For region-based belief $(s_A, b_E)$ represented by $\{ (\phi_E^i, w_i) \}_{i = 1}^{N_b}$, action $a$ and observation $s_A'$: $(s_A', b_E')$ returned by \algoref{alg:rb-belief-update} is region-based and $b_E' = b_E^{s_A, a, s_A'}$. Furthermore, if $h : S \to \mathbb{R}$ is PWC and $\Phi_E$ is a constant-FCP of $S_E$ for $h$ at $s_A$, then $\langle h, (s_A, b_E) \rangle = \sum_{i = 1}^{N_b} \sum_{\phi_E \in \Phi_E} h(s_A, s_E) w_i \textup{vol}(\phi_E^i \cap \phi_E)$ where $s_E \in \phi_E$.
\end{lema}
\begin{proof}
Consider a region-based belief $(s_A, b_E)$ represented by $\{ (\phi_E^i, w_i) \}_{i = 1}^{N_b}$, action $a$ and observation $s_A'$ and suppose that the belief $(s_A', b_E')$ is returned by \algoref{alg:rb-belief-update}.

Since $\delta_E^i(\cdot, a)$ is piecewise continuous by \lemaref{lema:uniform-closure}, then for any region $\phi_E \subseteq \Phi_E$, the image $\{ \delta_E^i(s_E, a) \mid s_E \in \phi_E \}$ can be represented as a union of regions. Furthermore, due to the invertibility of $\delta_E^i(\cdot, a)$, these regions are disjoint and the image is uniformly reached.  Letting $\phi_{ij} = \{ \delta_E^j(s_E, a) \mid s_E \in \phi_E^i \}$, according to the belief update \eqref{eq:belief-update} and the belief expression in \defiref{defi:region-based-belief}, we have:
    \begin{align*}
        \lefteqn{\int_{s_E \in S_E} \!\!\!\!\!\! b_E(s_E) \delta_E(s_E,a)(s_E') \textup{d} s_E  = \int_{s_E \in S_E} \left( \mbox{$\sum\nolimits_{i = 1}^{N_b}$} \chi_{\phi_E^i}(s_E) w_i \right) \delta_E(s_E,a)(s_E') \textup{d} s_E} \\
        & = \mbox{$\sum\nolimits_{i = 1}^{N_b}$} \left( \int_{s_E \in S_E}  \chi_{\phi_E^i}(s_E) w_i  \delta_E(s_E,a)(s_E') \textup{d} s_E \right) & \mbox{rearranging} \\
        & = \mbox{$\sum\nolimits_{i = 1}^{N_b}$} \left( \int_{s_E \in \phi_E^i}   w_i \delta_E(s_E,a)(s_E') \textup{d} s_E \right) & \mbox{by definition of $\chi_{\phi_E^i}$} \\
        & = \mbox{$\sum\nolimits_{i = 1}^{N_b}$} \left(  \int_{s_E \in \phi_E^i} w_i  \left( \mbox{$\sum\nolimits_{j = 1}^{N_e}$} \chi_{\phi_E^{ij}}(s_E') \frac{\mu_j}{\textup{vol}(\phi_E^{ij})}\textup{d} s_E \right) \right) \\
        & \lefteqn{\hspace*{0.68cm} \mbox{by definition of $\phi_{ij}$ and since it is uniformly reached by \lemaref{lema:uniform-closure}}} \\
        & = \mbox{$\sum\nolimits_{i = 1}^{N_b}$} \mbox{$\sum\nolimits_{j = 1}^{N_e}$}  \left( \int_{s_E \in \phi_E^i} w_i  \chi_{\phi_E^{ij}}(s_E') \frac{\mu_j}{\textup{vol}(\phi_E^{ij})}\textup{d} s_E \right) & \mbox{rearranging} \\
        & = \mbox{$\sum\nolimits_{i = 1}^{N_b}$} \mbox{$\sum\nolimits_{j = 1}^{N_e}$}  w_i  \chi_{\phi_E^{ij}}(s_E') \frac{\mu_j}{\textup{vol}(\phi_E^{ij})} \left( \int_{s_E \in \phi_E^i} \textup{d} s_E \right) &\mbox{rearranging} \\
        & = \mbox{$\sum\nolimits_{i = 1}^{N_b}$} \mbox{$\sum\nolimits_{j = 1}^{N_e}$} \chi_{\phi_E^{ij}}(s_E') \frac{ w_i   \mu_j \textup{vol}(\phi_E^i) }{\textup{vol}(\phi_E^{ij})} & \mbox{by definition of $\textup{vol}$.}
    \end{align*}
 Therefore, $b_E'$ can be constructed by normalizing $\int_{s_E \in S_E}b_E(s_E) \delta_E(s_E,a)(s_E') \textup{d} s_E$, which is a region-based belief. 
 
Next, consider any observation $s_A$ and PWC $h : S \to \mathbb{R}$ where $\Phi_E$ is a constant-FCP of $S_E$ for $h$ at $s_A$. By \eqref{expectation-eq} we have:
    \begin{align*}
        \lefteqn{\langle h, (s_A, b_E) \rangle = \int_{s_E \in S_E} h(s_A, s_E) b_E(s_E) \textup{d} s_E} \\
        & = \int_{s_E \in S_E} h(s_A, s_E) \mbox{$\sum\nolimits_{i = 1}^{N_b}$} \chi_{\phi_E^i}(s_E) w_i \textup{d} s_E & \mbox{by \defiref{defi:region-based-belief}}\\
        & = \mbox{$\sum\nolimits_{i = 1}^{N_b}$} \left( \int_{s_E \in S_E} h(s_A, s_E) \chi_{\phi_E^i}(s_E)  w_i  \textup{d} s_E \right) & \mbox{rearranging} \\
        & = \mbox{$\sum\nolimits_{i = 1}^{N_b}$}  \mbox{$\sum\nolimits_{\phi_E \in \Phi_E}$} \left( \int_{s_E \in \phi_E} h(s_A, s_E) \chi_{\phi_E^i}(s_E)  w_i  \textup{d} s_E \right) & \mbox{since $\Phi_E$ is an FCP} \\
        & = \mbox{$\sum\nolimits_{i = 1}^{N_b}$}  \mbox{$\sum\nolimits_{\phi_E \in \Phi_E}$} h(s_A, \phi_E) \left( \int_{s_E \in \phi_E} \chi_{\phi_E^i}(s_E)  w_i  \textup{d} s_E \right) \\
        & \lefteqn{\hspace{4cm} \mbox{since $\Phi_E$ is a constant-FCP of $S_E$ for $h$ at $s_A$}}  \\
        & = \mbox{$\sum\nolimits_{i = 1}^{N_b}$}  \mbox{$\sum\nolimits_{\phi_E \in \Phi_E}$} h(s_A, \phi_E) w_i \textup{vol}(\phi_E^i \cap \phi_E) & \mbox{by definition of $\textup{vol}$}
    \end{align*}  
which completes the proof. 
\end{proof}

\begin{lema}[Region-based upper bound]\label{rb-upper-bound-app}
For region-based belief $(s_A, b_E)$ represented by $\{ (\phi_E^i, w_i) \}_{i = 1}^{N_b}$ and $\Upsilon = \{ ((s_A^k, b_E^k), y_k) \mid k \in I\}$, if  $K_{\mathit{ub}} = K$, $(\phi_E^{\textup{max}},p)$ is returned by \algoref{alg:rb-upper-bound}, $b_E' = \sum_{k \in I_{s_{\scale{.75}{A}}}}  \lambda_k^{\star}  b_E^k$ and $b_E(s_E) > b_E'(s_E)$ for all $s_E \in \phi_E^{\textup{max}}$ where $\lambda_k^{\star}$ is a solution to the LP of \algoref{alg:rb-upper-bound}, then $p$ is an upper bound of $V_{\mathit{ub}}^{\Upsilon}$ at $(s_A, b_E)$. Furthermore, if $N_b = 1$, then  $p =  V_{\mathit{ub}}^{\Upsilon}(s_A,b_E)$.
\end{lema}
\begin{proof}
Suppose that $(s_A, b_E)$ is a region-based belief represented by $\{ (\phi_E^i, w_i) \}_{i = 1}^{N_b}$, $\Upsilon = \{ ((s_A^k, b_E^k), y_k) \mid k \in I\}$ and suppose  $K_{\mathit{ub}} = K$ and $(\phi_E^{\textup{max}},p)$ is returned by \algoref{alg:rb-upper-bound}, $b_E' = \sum_{k \in I_{s_{\scale{.75}{A}}}}  \lambda_k^{\star}  b_E^k$ where $\lambda_k^{\star}$ is a solution to the LP of \algoref{alg:rb-upper-bound} and $\phi_E^{\textup{max}} \subseteq S_E^{b_E>b_E'}$. 
Furthermore for each $k \in I$ suppose that $(s_A^k, b_E^k)$ is a region-based belief represented by $\{ (\phi_E^{kj}, w_{kj}) \}_{j = 1}^{N_b^k}$.

    Since $K_{\mathit{ub}} = K$, by definition of $K$ (see \thomref{thom:continuity}) we have:
    \begin{align}
        \lefteqn{K_{\mathit{ub}}(b_E,b_E') = (U - L) \int_{s_E \in S_E^{b_{\scale{.75}{E}}>b_{\scale{.75}{E}}^{\scale{.75}{'}}}} (b_E(s_E) - b_E'(s_E)) \textup{d}s_E} \nonumber \\
        & \; \lefteqn{= (U - L) \left( \int_{s_E \in S_E^{b_{\scale{.75}{E}}>b_{\scale{.75}{E}}^{\scale{.75}{'}}}} b_E(s_E) - \int_{s_E \in S_E^{b_{\scale{.75}{E}}>b_{\scale{.75}{E}}^{\scale{.75}{'}}}} b_E'(s_E) \textup{d}s_E \right)} & \mbox{rearranging} \nonumber \\
        & \; \lefteqn{\leq (U - L) \left( \int_{s_E \in S_E} b_E(s_E) - \int_{s_E \in S_E^{b_{\scale{.75}{E}}>b_{\scale{.75}{E}}^{\scale{.75}{'}}}} b_E'(s_E) \textup{d}s_E \right)} & \mbox{rearranging} \nonumber \\
        & = (U - L) \left( 1 - \int_{s_E \in S_E^{b_{\scale{.75}{E}}>b_{\scale{.75}{E}}^{\scale{.75}{'}}}} b_E'(s_E) \textup{d}s_E \right) \qquad \qquad \quad & \mbox{since $b_E \in \mathbb{P}(S_E)$.} \label{region3-eqn}
    \end{align}
Now since $\phi_E^{\textup{max}} \subseteq S_E^{b_E>b_E'}$ we have:
    \begin{align}
        \lefteqn{\int_{s_E \in S_E^{b_{\scale{.75}{E}}>b_{\scale{.75}{E}}^{\scale{.75}{'}}}} b_E'(s_E) \textup{d}s_E \; \geq \; \int_{s_E \in \phi_E^{\textup{max}}} b'(s_E) \textup{d}s_E} \nonumber \\
        & = \int_{s_E \in \phi_E^{\textup{max}}} \left( \mbox{$\sum_{k \in I_{s_{\scale{.75}{A}}}}$} \lambda_k^{\star} b_E^k(s_E) \textup{d}s_E \right)  & \mbox{by definition of $b_E'$}\nonumber  \\
        & = \mbox{$\sum\nolimits_{k \in I_{s_{\scale{.75}{A}}}}$} \left( \int_{s_E \in \phi_E^{\textup{max}}} \lambda_k^{\star} b_E^k(s_E) \textup{d}s_E \right) & \mbox{rearranging} \nonumber \\
        & \; \lefteqn{=  \mbox{$\sum\nolimits_{k \in I_{s_{\scale{.75}{A}}}}$} \left( \int_{s_E \in \phi_E^{\textup{max}}} \lambda_k^{\star} \mbox{$\sum\nolimits_{j = 1}^{N_{b}^k}$} \chi_{\phi_E^{kj}}(s_E) w_{kj} \textup{d}s_E \right)} & \nonumber \\
        & & \mbox{since $\{ (\phi_E^{kj}, w_{kj}) \}_{j = 1}^{N_b^k}$ represents $b^k_E$} \nonumber \\
        & \; \lefteqn{=  \mbox{$\sum\nolimits_{k \in I_{s_{\scale{.75}{A}}}}$} \mbox{$\sum\nolimits_{j = 1}^{N_{b}^k}$}  \lambda_k^{\star} \left( \int_{s_E \in \phi_E^{\textup{max}}} \chi_{\phi_E^{kj}}(s_E) w_{kj} \textup{d}s_E \right)} \nonumber & \mbox{rearranging} \\
        & \; \lefteqn{= \mbox{$\sum\nolimits_{k \in I_{s_{\scale{.75}{A}}}}$} \mbox{$\sum\nolimits_{j = 1}^{N_{b}^k}$} \lambda_k^{\star} w_{kj} \textup{vol}(\phi_E^{kj}  \cap \phi_E^{\textup{max}} )} \label{region2-eqn} & \mbox{by definition of $\textup{vol}$.}
    \end{align}
    Thus, substituting \eqnref{region2-eqn} into \eqnref{region3-eqn} we have:
    \[
    K_{\mathit{ub}}(b_E,b_E') \leq (U - L) \left( 1 - \mbox{$\sum\nolimits_{k \in I_{s_{\scale{.75}{A}}}}$} \mbox{$\sum\nolimits_{j = 1}^{N_{b}^k}$} \lambda_k^{\star} w_{kj} \textup{vol}(\phi_E^{kj}  \cap \phi_E^{\textup{max}} ) \right)
    \]
    and using \eqref{eq:new-ub}, it follows that the optimal value $p$ to the LP of \algoref{alg:rb-upper-bound} is an upper bound of $V_{\mathit{ub}}^{\Upsilon}$ at $(s_A, b_E)$.

    Finally, suppose that $N_b = 1$. Therefore $\phi_E^{\textup{max}} = \phi_E^1$ and since $\phi_E^1$ is the unique region with positive probabilities for $b_E$, by definition of $S_E^{b_E>b_E'}$ it follows that $S_E^{b_E>b_E'} \subseteq \phi_E^1$. Combining these with $\phi_E^{\textup{max}} \subseteq S_E^{b_E>b_E'}$, we have that $S_E^{b_E>b_E'} = \phi_E^{\textup{max}} = \phi_E^1$. Therefore, all the inequalities above become equalities, and therefore  $p =  V_{\mathit{ub}}^{\Upsilon}(s_A,b_E)$. 
\end{proof}

\clearpage

\clearpage



 \bibliographystyle{elsarticle-num} 
 \bibliography{references}





% \input{appendix_case_studies}
% \end{thebibliography}
\end{document}